\documentclass[11pt]{article}

\usepackage[sc]{mathpazo}

\usepackage[margin=1in]{geometry}
\usepackage[english]{babel}
\usepackage[utf8x]{inputenc}
\usepackage[compact]{titlesec}

\usepackage{cmap}
\usepackage[T1]{fontenc}
\usepackage{bm}
\usepackage{amsmath}
\usepackage{amsfonts}
\usepackage{amssymb}
\usepackage{amsbsy}
\usepackage{amsthm}

\usepackage{mathtools}
\usepackage{xspace}

\usepackage{graphicx, ucs}

\usepackage{subcaption}
\usepackage{rotating}
\usepackage{float}
\usepackage{tikz}

\usepackage{algorithm}
\usepackage[noend]{algpseudocode}
\usepackage{listings}

\usepackage{enumitem}
\usepackage{hyperref}
\usepackage{multirow}
\usepackage{array}

\usepackage{accents}

\usepackage[outline]{contour}
\usepackage{xcolor}

\usepackage[
linewidth=2pt,
linecolor=gray,
middlelinecolor= black,
middlelinewidth=0.4pt,
roundcorner=1pt,
topline = false,
rightline = false,
bottomline = false,
rightmargin=0pt,
skipabove=0pt,
skipbelow=0pt,
leftmargin=0pt,
innerleftmargin=4pt,
innerrightmargin=0pt,
innertopmargin=0pt,
innerbottommargin=0pt,
]{mdframed}

\usepackage{chngcntr}
\usepackage{soul}
\usepackage{arydshln}

\usepackage{wrapfig}

\newcommand{\G}{\ensuremath{\mathbb{G}}}
\newcommand{\N}{\ensuremath{\mathbb{N}}}
\newcommand{\Q}{\ensuremath{\mathbb{Q}}}
\newcommand{\R}{\ensuremath{\mathbb{R}}}

\newcommand{\Z}{\ensuremath{\mathbb{Z}}}

\newcommand{\Zp}{\ensuremath{\Z_p}}

\DeclarePairedDelimiter\inner{\langle}{\rangle}
\DeclarePairedDelimiter\abs{\lvert}{\rvert}
\DeclarePairedDelimiter\set{\{}{\}}

\DeclarePairedDelimiter\floor{\lfloor}{\rfloor}
\DeclarePairedDelimiter\ceil{\lceil}{\rceil}
\DeclarePairedDelimiter\norm{\lVert}{\rVert}

 \newcommand{\matA}{\ensuremath{\mathbf{A}}}
\newcommand{\matB}{\ensuremath{\mathbf{B}}}

\newcommand{\matD}{\ensuremath{\mathbf{D}}}
\newcommand{\matE}{\ensuremath{\mathbf{E}}}

\newcommand{\matG}{\ensuremath{\mathbf{G}}}

\newcommand{\matI}{\ensuremath{\mathbf{I}}}

\newcommand{\matU}{\ensuremath{\mathbf{U}}}
\newcommand{\matV}{\ensuremath{\mathbf{V}}}
\newcommand{\matW}{\ensuremath{\mathbf{W}}}

\newcommand{\matzero}{\ensuremath{\mathbf{0}}}

\newcommand{\veca}{\ensuremath{\mathbf{a}}}
\newcommand{\vecb}{\ensuremath{\mathbf{b}}}
\newcommand{\vecc}{\ensuremath{\mathbf{c}}}

\newcommand{\vecg}{\ensuremath{\mathbf{g}}}

\newcommand{\vecr}{\ensuremath{\mathbf{r}}}
\newcommand{\vecs}{\ensuremath{\mathbf{s}}}
\newcommand{\vect}{\ensuremath{\mathbf{t}}}
\newcommand{\vecu}{\ensuremath{\mathbf{u}}}
\newcommand{\vecv}{\ensuremath{\mathbf{v}}}

\newcommand{\vecx}{\ensuremath{\mathbf{x}}}
\newcommand{\vecy}{\ensuremath{\mathbf{y}}}

\newcommand{\veczero}{\ensuremath{\mathbf{0}}}
 \theoremstyle{plain}
\newtheorem{theorem}{Theorem}[section]
\newtheorem{lemma}[theorem]{Lemma}
\newtheorem{corollary}[theorem]{Corollary}
\newtheorem{proposition}[theorem]{Proposition}
\newtheorem{claim}[theorem]{Claim}

\theoremstyle{definition}
\newtheorem{definition}[theorem]{Definition}

\newtheorem{openProblem}[theorem]{Open Problem}

\theoremstyle{remark}
\newtheorem{remark}[theorem]{Remark}

\numberwithin{equation}{section}

\newenvironment{proofb}[1]{\noindent \textit{Proof #1.}}{\hfill \rule{1.5ex}{1.5ex}}

\contourlength{0.1pt}
\contournumber{10}

\newenvironment{nproblem}[1][\unskip]{
\medskip
\begin{mdframed}
\noindent
\contour{black}{\underline{$#1$ Problem.}} \\
\noindent}
{\end{mdframed}}

\newcommand{\bit}{\ensuremath{\set{0,1}}}

\DeclareMathOperator{\poly}{poly}
\DeclareMathOperator{\polylog}{polylog}
\DeclareMathOperator{\negl}{negl}

\newcommand{\calH}{\ensuremath{\mathcal{H}}}

\newcommand{\algo}[1]{\ensuremath{\mathsf{#1}}\xspace}

\newcommand{\problem}[1]{\ensuremath{\mathsf{#1}}\xspace}

\newcommand{\class}[1]{\ensuremath{\mathsf{#1}}\xspace}

\renewcommand{\P}{\class{P}}
\newcommand{\FP}{\class{FP}}

\newcommand{\NP}{\class{NP}}
\newcommand{\FNP}{\class{FNP}}
\newcommand{\TFNP}{\class{TFNP}}
\newcommand{\coNP}{\class{coNP}}

\newcommand{\PPP}{\class{PPP}}
\newcommand{\wPPP}{\class{PWPP}}
\newcommand{\PWPP}{\wPPP}
\newcommand{\PPA}{\class{PPA}}
\newcommand{\PLS}{\class{PLS}}
\newcommand{\PPAD}{\class{PPAD}}
\newcommand{\PPADS}{\class{PPADS}}
\newcommand{\CLS}{\class{CLS}}

\newcommand{\Ppoly}{\class{P/poly}}

\newcommand{\lamperp}{\Lambda^{\perp}}
\newcommand{\lam}{\Lambda}

\newcommand{\lat}{\mathcal{L}}
\newcommand{\dual}{\mathcal{L}^{*}}
\newcommand{\piped}{\mathcal{P}}

\DeclareMathOperator{\vol}{vol}

\DeclareMathOperator{\dist}{dist}
\DeclareMathOperator{\spn}{span}

\DeclareMathOperator{\ord}{ord}

\newcommand{\svp}{\problem{SVP}}

\newcommand{\sivp}{\problem{SIVP}}

\newcommand{\cvp}{\problem{CVP}}

\newcommand{\bdd}{\problem{BDD}}

\newcommand{\sis}{\problem{SIS}}

\newcommand{\attacker}[1]{\ensuremath{\mathcal{#1}}}

\newcommand{\Adv}{\attacker{A}}

\newcommand{\scheme}[1]{\ensuremath{\text{#1}}}

\newcommand{\crhf}{\scheme{H}}
\newcommand{\crhfgen}{\algo{Gen}}

\newif\ifnotes\notestrue

\ifnotes
\usepackage{color}
\definecolor{mygrey}{gray}{0.50}
\newcommand{\notename}[2]{{\textcolor{red}{\footnotesize{\bf (#1:} {#2}{\bf
) }}}}

\else

\newcommand{\notename}[2]{{}}

\fi

\mathchardef\mdash="2D

\renewcommand{\epsilon}{\varepsilon}

\def\binset{\{0,1\}}

\def\compactify{\itemsep=0pt \topsep=0pt \partopsep=0pt \parsep=0pt}
\let\latexusecounter=\usecounter

\makeatletter
\newtheorem*{rep@theorem}{\rep@title}
\newcommand{\newreptheorem}[2]{
	\newenvironment{rep#1}[1]{
		\def\rep@title{#2 \ref{##1}}
		\begin{rep@theorem}}
		{\end{rep@theorem}}}
\makeatother

\newreptheorem{theorem}{Theorem}

\newenvironment{Enumerate}
  {\def\usecounter{\compactify\latexusecounter}
   \begin{enumerate}}
  {\end{enumerate}\let\usecounter=\latexusecounter}

\newenvironment{Itemize}
{\begin{itemize}
\setlength{\itemsep}{0pt}
\setlength{\topsep}{0pt}
\setlength{\partopsep}{0 in}
\setlength{\parskip}{0 pt}}
{\end{itemize}}

\newcommand{\paragr}[1]{\noindent \textbf{#1}}

\newcommand{\reals}{\mathbb{R}}

\newcommand{\integer}{\mathbb{Z}}

\newcommand{\ubar}[1]{\underaccent{\bar}{#1}}

\renewcommand{\vec}{\boldsymbol}
\newcommand{\circuit}[1]{\mathcal{#1}}
\newcommand{\bvec}[1]{\vec{\ubar{#1}}}
\newcommand{\idlabel}{\text{id}}
\newcommand{\gnand}{~\bar{\wedge}~}
\newcommand{\gnor}{~\bar{\vee}~}
\newcommand{\gxor}{\oplus}
\newcommand{\gand}{\wedge}
\newcommand{\gor}{\vee}

\newcommand{\pigeon}{\problem{PIGEONHOLE~CIRCUIT}}
\newcommand{\collision}{\problem{COLLISION}}
\newcommand{\Blichfeldt}{\problem{BLICHFELDT}}
\newcommand{\cSIS}{\problem{cSIS}}
\newcommand{\Minkowski}{\problem{MINKOWSKI}}

\newcommand{\h}{h}
\newcommand{\weakcSIS}{\problem{weak \mdash cSIS}}
\newcommand{\gone}{1}

\newcommand{\Cindex}{\mathcal{I}}
\newcommand{\Cvalue}{\mathcal{V}}
\newcommand{\Cchara}{\mathcal{CH}}

\renewcommand{\pmod}[1]{~(\bmod~#1)}

\newcommand{\coset}{\textsc{Co}}

\newcommand{\BitDecomp}{\ensuremath{\mathsf{bd}}\xspace}
\newcommand{\BitComp}{\ensuremath{\mathsf{bc}}\xspace}

\newcommand{\graph}{\mathcal{G}}

\newcommand\Tstrut{\rule{0pt}{3.2ex}}
\newcommand\Tstruts{\rule{0pt}{2.6ex}}

\newcommand{\nand}{\mathsf{NAND}}

\newcommand{\HcSIS}{\mathcal{H}_\cSIS}

\begin{document}
\title{$\PPP$-Completeness with Connections to Cryptography}
\date{}
\author{
  Katerina Sotiraki \thanks{email: katesot@mit.edu. The author was partly
  supported by NSF grants CNS-1350619, CNS-1718161, CNS-1414119 and by the
  Chateaubriand Fellowship of the Office for Science and Technology of the
  Embassy of France in the United States.} \\ MIT
  \and Manolis Zampetakis \thanks{email: mzampet@mit.edu. The
  author was supported by NSF grants CCF-1551875, CCF-1617730, CCF-1650733.} \\
  MIT
  \and Giorgos Zirdelis \thanks{email: zirdelis.g@husky.neu.edu. The author
  was supported by NSF grants CNS-1314722, CNS-1413964, CNS-1750795.} \\
  Northeastern University
}
\clearpage
\maketitle

\begin{abstract}
  Polynomial Pigeonhole Principle ($\PPP$) is an important subclass of $\TFNP$
  with profound connections to the complexity of the fundamental cryptographic
  primitives:
   \textit{collision-resistant hash functions} and \textit{one-way
     permutations}. In contrast to most of the other subclasses of $\TFNP$, no
   complete problem is known for $\PPP$. Our work identifies the first
  $\PPP$-complete problem without any circuit or Turing Machine given
  explicitly in the input: $\problem{constrained\mdash SIS}$, and thus we
  answer a longstanding open question from~\cite{Papadimitriou1994}.

  \begin{quote}
    $\problem{constrained\mdash SIS}$: a generalized version of the well-known
    Short Integer Solution problem ($\sis$) from lattice-based cryptography.
  \end{quote}

  \noindent In order to give some intuition behind our reduction for
  $\problem{constrained\mdash SIS}$, we identify another $\PPP$-complete problem
  with a circuit in the input but closely related to lattice problems:
  $\Blichfeldt$.

  \begin{quote}
    $\Blichfeldt$: the computational problem associated with Blichfeldt's
    fundamental theorem in the theory of lattices.
  \end{quote}

  Building on the inherent connection of $\PPP$ with collision-resistant hash
  functions, we use our completeness result to construct the first natural hash
  function family that captures the hardness of all collision-resistant hash
  functions in a worst-case sense, i.e. it is natural and \textit{universal in
  the worst-case}. The close resemblance of our hash function family with
  $\sis$, leads us to the first candidate collision-resistant hash function that
  is both natural and universal \textit{in an average-case sense}.

  Finally, our results enrich our understanding of the connections between $\PPP$,
  lattice problems and other concrete cryptographic assumptions, such as the
  discrete logarithm problem over general groups.
\end{abstract}
\thispagestyle{empty}

\addtocounter{page}{-1}\newpage

\section{Introduction} \label{sec:intro}

The fundamental task of \textit{Computational Complexity} theory is to classify
computational problems according to their inherent computational difficulty.
This led to the definition of \textit{complexity classes} such as $\NP$ which
contains the \textit{decision} problems with \textit{efficiently} verifiable
proofs in the ``yes'' instances. The \textit{search} analog of the class $\NP$,
called $\FNP$, is defined as the class of \textit{search} problems whose
decision version is in $\NP$. The same definition extends to the class $\FP$, as
the search analog of $\P$. The seminal works of \cite{JohnsonPY1988,
Papadimitriou1994} considered search problems in $\FNP$ that are \textit{total},
i.e. their decision version is always affirmative and thus a solution must
always exist. This totality property makes the definition of $\FNP$ inadequate
to capture the intrinsic complexity of total problems in the appropriate way as
it was first shown in \cite{JohnsonPY1988}. Moreover, there were evidences for
the hardness of total search problems e.g. in~\cite{HirschPV1989}. Megiddo and
Papadimitriou~\cite{MeggidoP1989} defined the class $\boldsymbol{\TFNP}$ that
contains the total search problems of $\FNP$, and Papadimitriou
\cite{Papadimitriou1994} proposed the following classification rule of problems
in $\TFNP$:

\begin{quote}
  Total search problems should be classified in terms of the profound
  mathematical principles that are invoked to establish their totality.
\end{quote}

\noindent
Along these lines, many subclasses for $\TFNP$ have been defined.
Johnson, Papadimitriou and Yannakakis \cite{JohnsonPY1988} defined the
class $\boldsymbol{\PLS}$. A few years later, Papadimitriou
\cite{Papadimitriou1994} defined the complexity classes $\boldsymbol{\PPA}$,
$\boldsymbol{\PPAD}$, $\boldsymbol{\PPADS}$ and $\boldsymbol{\PPP}$, each one
associated with a profound mathematical principle in accordance with the
above classification rule. More recently, the classes $\boldsymbol{\CLS}$ and
$\boldsymbol{\wPPP}$ were defined in \cite{DaskalakisP11} and \cite{Jerabek16},
respectively. In Section \ref{sec:relatedWork} we give a high-level description
of all these classes.

Finding complete problems for the above classes is important as it enhances our
understanding of the underlying mathematical principles. In turn, such
completeness results reveal equivalences between total search problems, that
seemed impossible to discover without invoking the definition of these classes.
Since the definition of these classes in \cite{JohnsonPY1988, Papadimitriou1994}
it was clear that the completeness results about problems that don't have
explicitly a Turing machine or a circuit as a part of their input are of
particular importance. For this reason it has been established to call such
problems \textit{natural} in the context of the complexity of total search
problems (see \cite{FilosG18}).

Many natural complete problems are known for $\PLS$ and $\PPAD$, and recently
natural complete problems for $\PPA$ were identified too (see Section
\ref{sec:relatedWork}). However, no natural complete problems are known for the
classes $\PPP$, $\wPPP$ that have profound connections with the hardness of
important cryptographic primitives, as we explain later in detail.

\medskip

\paragr{Our Contributions.} Our main contribution is to provide the first
natural complete problems for $\PPP$ and $\PWPP$, and thus solve a longstanding
open problem from~\cite{Papadimitriou1994}. Beyond that, our $\PPP$ completeness
results lead the way towards answering important questions in cryptography and
lattice theory as we highlight below.
\medskip

\noindent \textsc{Universal Collision-Resistant Hash Function.} Building on the
inherent connection of $\PWPP$ with \textit{collision-resistant hash functions},
we construct a natural hash function family $\HcSIS$ with the following
properties:
\begin{Enumerate}
  \item[-] \textbf{Worst-Case Universality.} No efficient algorithm can find a
           collision in every function of the family $\HcSIS$, unless worst-case
           collision-resistant hash functions do not exist.
           \smallskip

           Moreover, if an (average-case hard)
           collision-resistant hash function family exists, then there exists an
           efficiently samplable distribution $\mathcal{D}$ over $\HcSIS$, such
           that $(\mathcal{D}, \HcSIS)$ is an (average-case hard)
           collision-resistant hash function family.

  \item[-] \textbf{Average-Case Hardness.} No efficient algorithm can find a
    collision in a function chosen \textit{uniformly at random}
    from $\HcSIS$, unless we can efficiently find short lattice vectors in
    any (worst-case) lattice.
\end{Enumerate}

\noindent The first property of $\HcSIS$ is reminiscent of the existence of
\emph{worst-case} one-way functions from the assumption that $\P \neq \NP$
\cite{Selman1992}. The corresponding assumption for the existence of worst-case
collision-resistance hash functions is assuming $\FP \neq \PWPP$,
but our hash function family $\HcSIS$ is the first natural definition that does
not involve circuits, and admits this strong completeness guarantee in the
worst-case.

The construction and properties of $\HcSIS$ lead us to the first candidate of a
\textit{natural} and \textit{universal collision-resistant hash function
family}. The idea of universal constructions of cryptographic primitives was
initiated by Levin in \cite{Levin1987}, who constructed the first universal
one-way function and followed up by \cite{Levin2003, Kozhevnikov2009}. Using the
same ideas we can also construct collision a universal collision-resistant hash
function family as we describe in Appendix \ref{sec:app:universalHash}. The
constructed hash function though invokes in the input an explicit description of
a Turing machine and hence it fails to be \textit{natural}, with the definition
of naturality that we described before. In contrast, our candidate construction
is natural, simple, and could have practical applications.
\medskip

\noindent \textsc{Complexity of Lattice Problems in $\PPP$.}
The hardness of lattice problems in $\NP \cap \coNP$~\cite{FOCS:AhaReg04} has
served as the foundation for numerous cryptographic constructions in the past
two decades. This line of work was initiated by the breakthrough work of
Ajtai~\cite{Ajtai1996}, and later developed in a long series of works (e.g.
\cite{AjtaiD97, MicRegev07, Regev09, STOC:GenPeiVai08, STOC:Peikert09,
STOC:GorVaiWee13, STOC:BLPRS13, BrakerskiV14, C:GenSahWat13, STOC:GorVaiWic15,
FOCS:GoyKopWat17, WichsZ17, STOC:PeiRegSte17}). This wide use of search
(approximation) lattice problems further motivates their study.

We make progress in understanding this important research front by showing that:
\begin{Enumerate}
\item the computational problem $\Blichfeldt$ associated with Blichfeldt's
theorem, which can be viewed as a generalization of Minkowski's theorem, is
$\PPP$-complete,
\item the $\cSIS$ problem, a constrained version of the Short Integer Solution
($\sis$), is $\PPP$-complete,
\item we combine known results and techniques from lattice theory to show
that most approximation lattice problems are reducible to $\Blichfeldt$ and $\cSIS$.
\end{Enumerate}

These results create a new path towards a better understanding of lattice
problems in terms of complexity classes.
\medskip

\noindent \textsc{Complexity of Other Cryptographic Assumptions.} Besides
lattice problems, we discuss the relationship of other well-studied
cryptographic assumptions and $\PPP$. Additionally, we formulate a white-box
variation of the \textit{generic group model} for the discrete logarithm
problem~\cite{Shoup1997}; we observe that this problem is in $\PPP$ and is
another natural candidate for being $\PPP$-complete.
\smallskip

\subsection{Related Work} \label{sec:relatedWork}
  In this section we discuss the previous work on the complexity of total search
problems, that has drawn attention from the theoretical computer science
community over the past decades. We start with a high-level description of the
total complexity classes and then discuss the known results for each one of
them.

\begin{description}
  \item[$\boldsymbol{\PLS}$.] The class of problems whose totality is
                              established using a potential function argument.\\
                              \textit{Every finite directed acyclic graph
                              has a sink.}
  \item[$\boldsymbol{\PPA}$.] The class of problems whose totality is proved
                              through a parity argument.\\
                              \textit{Any finite graph has an even number of
                              odd-degree nodes.}
  \item[$\boldsymbol{\PPAD}$.] The class of problems whose totality is proved
                              through a directed parity argument.\\
                              \textit{All directed graphs of degree two or less
                              have an even number of degree one nodes.}
  \item[$\boldsymbol{\PPP}$.] The class of problems whose totality is proved
                              through a pigeonhole principle argument.\\
                              \textit{Any map from a set $S$ to itself either is
                              onto or has a collision.}
\end{description}

\noindent Using the same spirit two more classes were defined after
\cite{Papadimitriou1994}, in \cite{DaskalakisP11} and \cite{Jerabek16}.
\begin{description}
  \item[$\boldsymbol{\CLS}$.] The class of problems whose totality is
                              established using both a potential function
                              argument and a parity argument.
  \item[$\boldsymbol{\wPPP}$.] The class of problems whose totality is proved
                              through a weak pigeonhole principle. \\
                              \textit{Any map from a set $S$ to a strict subset
                              of $S$ has a collision.}
\end{description}

\begin{wrapfigure}{r}{0.34\textwidth}
  \begin{center}
    \includegraphics[width=0.30\textwidth]{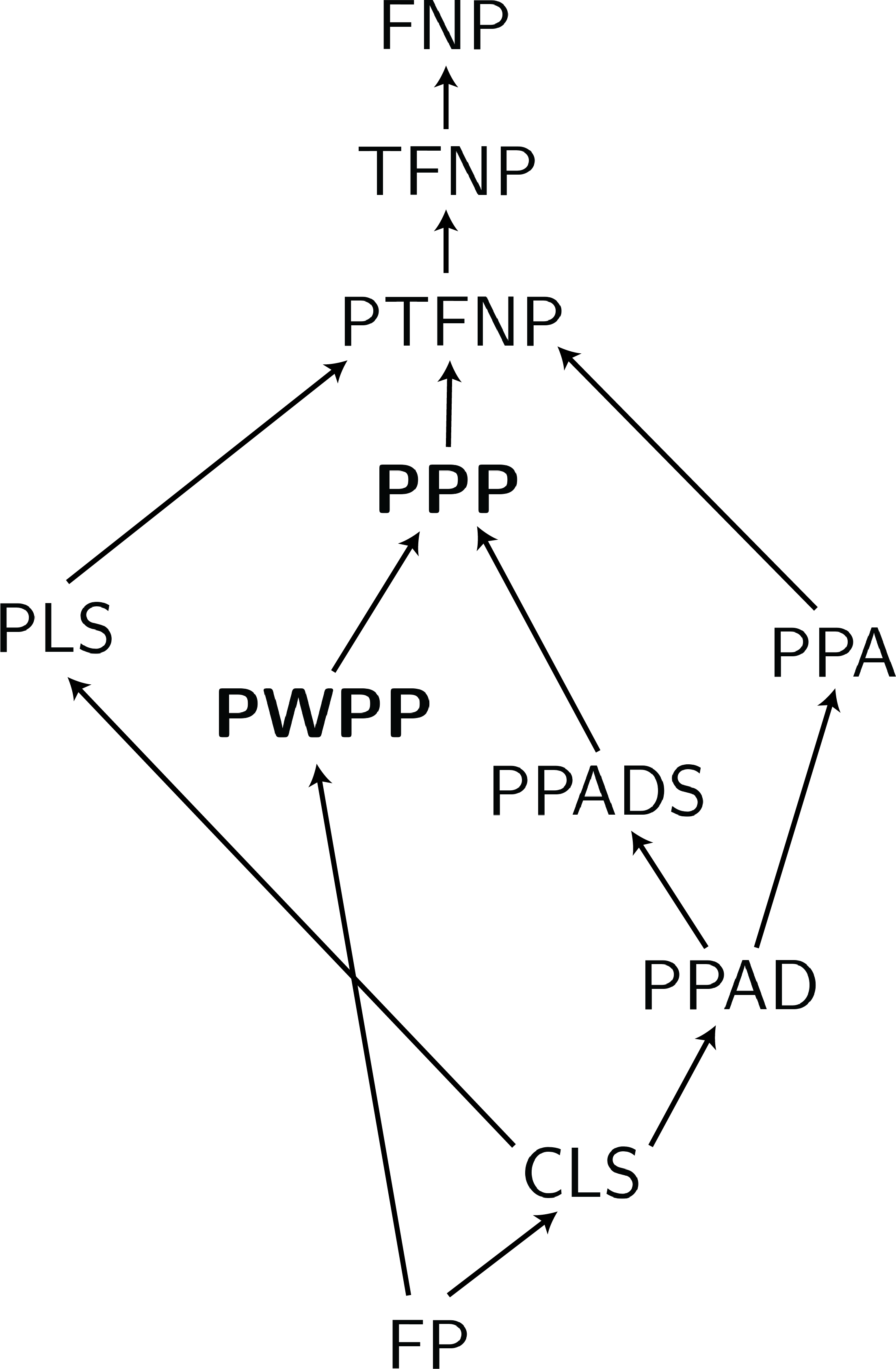}
  \end{center}
  \caption{The classes $\PPP$ and $\PWPP$ in the $\TFNP$ world.}
\end{wrapfigure}

\noindent Recently, a syntactic analog $\class{PTFNP}$ of the semantic class
$\TFNP$ has been defined in \cite{GoldbergP17}, and a complete problem for this
class has been identified. It has also been shown that all the classes we
described above are subsets of $\class{PTFNP}$. Oracle separations between all these
classes are known \cite{BeameCEIP1998}, with the only exception of whether
\class{PLS} is contained in \class{PPAD}.

\paragr{$\boldsymbol{\PLS}$-completeness.} The class $\PLS$ represents the
complexity of local optimization problems. Some important problems that have
been shown to be $\PLS$-complete are: Local Max-Cut \cite{SchafferY1991}, Local
Travelling Salesman Problem \cite{Papadimitriou1992}, and Finding a Pure Nash
Equilibrium \cite{FabrikantPT04}. Recently, important results for the smoothed
complexity of the Local Max-Cut problem were shown in~\cite{EtscheidR17, AngelBPW17}.

\paragr{$\boldsymbol{\PPAD}$-completeness.} Arguably, the most celebrated
application of the complexity of total search problems is the characterization
of the computational complexity of finding a Nash equilibrium
in terms of $\PPAD$-completeness \cite{DaskalakisGP09, ChenDT09}. This problem
lies in the heart of game theory and economics. The proof that Nash Equilibrium
is $\PPAD$-complete initiated a long line of research in the intersection of
computer science and game theory and revealed connections between the two
scientific communities that were unknown before (e.g. \cite{ElkindGG06, ChenDDT09,
VaziraniY11, KintaliPRST13, ChenO15, Rubinstein15, Rubinstein16, ChenPY17,
SchuldenzuckerSB17}).

\paragr{$\boldsymbol{\PPA}$-completeness.} $\PPA$-complete problems usually
arise as the undirected generalizations of their $\PPAD$-complete analogs. For
example, Papadimitriou \cite{Papadimitriou1994}
showed that Sperner's Lemma in a 3-D cube is $\PPAD$-complete and later Grigni
\cite{Grigni01} showed that Sperner’s Lemma in a 3-manifold consisting of the
product of a M{\"o}bius strip and a line segment is $\PPA$-complete. Since
M{\"o}bius strip is non-orientable, this indeed is a non-directed version of the
Sperner's Lemma. Similarly, other problems have been showed to be
$\PPA$-complete, all involving some circuit as an input in their definition
\cite{AisenbergBB15, DengEFLQX16, BelovsIQSY17}. Recently, the first natural
$\PPA$-complete problem, without a circuit as part of the input, has been
identified in \cite{FilosG18}. This illustrates an interesting relation between
$\PPA$ and complexity of social choice theory problems.

\paragr{$\boldsymbol{\CLS}$-completeness.} The $\CLS$ class was defined in
\cite{DaskalakisP11} to capture the complexity of problems such as
P-matrix LCP, computing KKT-points, and finding Nash equilibria in
congestion and network coordination games. Recently, it has been proved that
the problem of finding a fixed point whose existence invokes Banach's Fixed
Point Theorem, is $\CLS$-complete \cite{DaskalakisTZ18, FearnleyGMS17}.

\paragr{$\boldsymbol{\TFNP}$ and cryptography.} The connection of $\TFNP$ and
cryptography was illustrated by Papadimitriou in~\cite{Papadimitriou1994}, where
he proved that if $\PPP = \FP$ then \textit{one-way permutations} cannot exist.
In \cite{Buresh06}, a special case of integer factorization was shown to be in
$\PPA \cap \PPP$. This was generalized in \cite{Jerabek16} by proving that the
problem of factoring integers is in $\PPA \cap \PPP$ under randomized
reductions. Recently, strong cryptographic assumptions were used to prove the
average-case hardness of $\PPAD$ and $\CLS$ \cite{BitanskyPR15, GargPS16,
HubavcekY17}. In \cite{RosenSS17} it was shown that average-case $\PPAD$
hardness does not imply one-way function under black-box reductions, whereas in
\cite{HubacekNY17} it was shown that any hard on average problem in $\NP$
implies the average case hardness of $\TFNP$. Finally, in \cite{KomargodskiNY17}
it is proved that the existence of \textit{multi-collision resistant hash
functions} is equivalent with a variation of the total search problem
$\problem{RAMSEY}$, which is not known to belong to any of the above complexity
classes. Interestingly, they prove that a variation of $\problem{RAMSEY}$ called
\textit{colorful-Ramsey} ($\problem{C \mdash RAMSEY}$) is $\PWPP$-hard. Although
this an important result, the problem $\problem{C \mdash RAMSEY}$ still invokes
a circuit in the input and in not known to be in $\wPPP$, hence does not resolve
the problem of identifying a natural complete problem for $\wPPP$.

\paragr{$\boldsymbol{\TFNP}$ and lattices.} In~\cite{BanJPPR15} it was shown
that the computational analog of Minkowski's theorem (namely $\Minkowski$) is in
$\PPP$, was conjectured that it is also $\PPP$-complete. The authors justified
their conjecture by showing that $\problem{EQUAL \mdash SUMS}$, a problem
from~\cite{Papadimitriou1994} that is conjectured to be $\PPP$-complete, reduces
to $\Minkowski$. Additionally, they show that a number theoretic problem called
$\problem{DIRICHLET}$ reduces to $\Minkowski$, and thus is in $\PPP$.
In~\cite{conf/ipco/HobergRRY17} it is proven that the problem $\problem{NUMBER
\mdash BALANCING}$ is equivalent to a polynomial approximation of Minkowski's
theorem in the $\ell_2$ norm (via Cook reductions for both directions).

\subsection{Roadmap of the paper}
We start our exposition with a brief description of the results contained in
this paper. First we briefly describe the $\PPP$-completeness of $\Blichfeldt$
that illustrates some of the basic ideas behind our main result that $\cSIS$ is
$\PPP$-complete. The complete proof of the $\PPP$-completeness of $\Blichfeldt$
can be found in Section~\ref{sec:Blichfeldt}. We suggest to readers that have
experience with the fundamental concepts of the theory of lattices to skip the
details Section~\ref{sec:Blichfeldt}.

Then, we present a brief description of our main theorem and its proof. The
complete proof of the $\PPP$-completeness of $\cSIS$ can be found in
Section~\ref{sec:cSIS}.

In Section \ref{sec:collision} we describe the $\PPP$-completeness of a weaker
version of $\cSIS$ and its relation with the definition of the first natural
universal collision resistant hash function family in the worst-case sense. This
proof also provides the first candidate for a collision resistant hash function
family that is both natural and universal in the average-case sense.

Finally in Section \ref{sec:lattices} we present, for completeness of our
exposition, other lattice problems that are already known to belong to $\PPP$
and $\PWPP$ and in Section \ref{sec:cryptoHard} we present more general other
cryptographic assumptions that belong to $\PPP$ and $\PWPP$.

\subsection{Overview of the Results}
Before we describe our results in more detail we define the class $\PPP$ more
formally. The class $\PPP$ contains the set of problems that are reducible to
the $\pigeon$ problem. The input to $\pigeon$ is a binary circuit
$\circuit{C} : \binset^n \to \binset^n$ and its output is either an
$\bvec{x} \in \binset^n$ such that $\circuit{C}(\bvec{x}) = \bvec{0}$, or
a pair $\bvec{x}, \bvec{y} \in \binset^{n}$ such that $\bvec{x} \neq \bvec{y}$
and $\circuit{C}(\bvec{x}) = \circuit{C}(\bvec{y})$.

Our first and technically most challenging result is to identify and prove the
$\PPP$-completeness of two problems, both of which share similarities with
lattice problems. For our exposition, a lattice $\lat \subseteq \Z^n$ can be
viewed as a finitely generated additive subgroup of $\Z^n$. A lattice $\lat =
\lat(\matB) := \matB \cdot \Z^n$ is generated by a full-rank matrix $\matB \in
\Z^{n \times n}$, called \textit{basis}. In the rest of this section we also use
the \emph{fundamental parallelepiped} of $\lat$ defined as $\piped(\lat) :=
\matB \cdot [0,1)^n$.
\medskip

\subsubsection{\texorpdfstring{$\boldsymbol{\Blichfeldt}$}{$\Blichfeldt$} is
\texorpdfstring{$\boldsymbol{\PPP}$}{$\PPP$}-complete.}

We define the $\Blichfeldt$ problem as the computational analog of Blichfeldt's
theorem (see Theorem \ref{thm:BlichfeldtTheorem}). Its input is a basis for a
lattice $\lat \subseteq \Z^n$ and a set $S \subseteq \Z^n$ of cardinality
greater or equal to the volume of $\piped(\lat)$. Its output is either a point
in $S$ that belongs to $\lat$, or for two (different) points in $S$ such that
their difference belongs to $\lat$. In the overview below, we explain why such
an output always exists. Notice that finding a solution to $\Blichfeldt$ becomes
trivial if the input representation of $S$ has length proportional to its size,
i.e. one can iterate over all element pairs of $S$. The problem becomes
challenging when $S$ is represented succinctly. We introduce a notion for a
succinct representation of sets that we call \emph{value function}. Informally,
a value function for a set $S$ is a small circuit that takes as input
$\ceil{\log(|S|)}$ bits that describe an index $i \in \{0,\ldots,|S|-1\}$, and
outputs $\vecs_i \in S$.

We give a proof overview of our first main theorem, and highlight the obstacles
that arise, along with our solutions.

\begin{reptheorem}{thm:BlichfeldtPPPcompleteness}
	The $\Blichfeldt$ problem is $\PPP$-complete.
\end{reptheorem}

\smallskip

\noindent \textsc{$\PPP$ Membership of $\Blichfeldt$ Overview.}
We denote with $[n]$ the set $\{0,\ldots,n-1\}$. We define the map $\vec{\sigma}
: \Z^n \rightarrow \piped(\lat) \cap \Z^n$ that reduces any point in $\Z^n$
modulo the parallelepiped to $\piped(\lat) \cap \Z^n$, i.e.
$\pmod{\piped(\lat)}$. Using $\vec{\sigma}$ we can efficiently check the
membership of any $\vecv \in \Z^n$ in $\lat$, by checking if $\vec{\sigma}$ maps
$\vecv$ to the origin. Observe that if $\vec{\sigma}(\vecx) =
\vec{\sigma}(\vecy)$  then $\vecx - \vecy \in \lat$.

We show in Lemma~\ref{lm:inclusion-Blichfeldt} that $\vol(\piped(\lat)) =
|\piped(\lat) \cap \Z^n|$, hence the input requirement for $S$ is equivalent to
$|S| \geq |\piped(\lat) \cap \Z^n|$. Notice that the points of $S$ after
applying the map $\vec{\sigma}$, either have a collision in $\piped(\lat) \cap
\Z^n$ or a preimage of the origin exists in $S$. It follows by a pigeonhole
argument that a solution to $\Blichfeldt$ always exists. For the rest of this
part we assume that $|S| = |\piped(\lat) \cap \Z^n|$ and let $n =
\ceil{\log(|S|)}$.

We construct a circuit $\circuit{C}:\binset^{n} \to \binset^n$ that on input an
appropriate index $i$, evaluates the value function of $S$ to obtain $\vecs_i
\in S$, and computes $\vec{\sigma}(\vecs_i)$. The most challenging part of the
proof is to construct an efficient map from $\vec{\sigma}(S)$ to $[|\piped(\lat)
\cap \Z^n|]$ in the following way. We define an appropriate parallelepiped $D =
[L_1] \times [L_2] \times \dots \times [L_n]$ where the $L_i$ are non-negative
integers, and a bijection $\vec{\pi} : \piped(\lat) \cap \Z^n \rightarrow D$.
Because $D$ is a cartesian product, a natural efficient indexing procedure
exists as described in Lemma~\ref{lm:cubeSetSuccinctDescription}. This allows to
map $\vec{\pi}(\vec{\sigma}(\vecs_i))$ to $j \in [|\piped(\lat) \cap \Z^n|]$.
The circuit $\circuit{C}$ outputs the binary decomposition of $j$. It follows
that any $\bvec{x}$ such that $\circuit{C}(\bvec{x}) = \bvec{0}$ corresponds to
a vector $\vecx \in S$ such that $\vec{\sigma}(\vecx) = \veczero$. On the other
hand, a collision $\circuit{C}(\bvec{x}) = \circuit{C}(\bvec{y})$ with $\bvec{x}
\neq \bvec{y}$ corresponds to a collision $\vec{\sigma}_S(\vecx) =
\vec{\sigma}_S(\vecy)$, where $\vec{\sigma}_S$ is the restriction of
$\vec{\sigma}$ on $S$, and hence $\vecx - \vecy \in \lat$.
\medskip

\noindent \textsc{$\PPP$ Hardness of $\Blichfeldt$ Overview.} We start with a
circuit $\circuit{C}: \binset^n \rightarrow \binset^n$ that is an input to
$\pigeon$. We construct a set $S$ and a lattice $\lat$ as input to $\Blichfeldt$
in the following way. The set $S$ contains the elements
$\vecs_{\bvec{x}}=\begin{bmatrix} \bvec{x} \\ \circuit{C}(\bvec{x})
\end{bmatrix} \in \binset^{2n}$ and is represented succinctly with the value
function that maps $\bvec{x}$ to $\vecs_{\bvec{x}}$. Notice that $|S|=2^n$. The
lattice $\lat$ consists of all $\bvec{v} \in \binset^{2n}$ that satisfy the
equation $[\vec{0}_n~~\matI_n] \cdot \bvec{v} = \vec{0} \pmod{2}$. By
Lemma~\ref{lm:determinantQary}, one can efficiently obtain a basis from this
description of $\lat$ and in addition the volume of $\piped(\lat)$ is at most
$2^n$. Thus, $S$ and $\lat$ is a valid input for $\Blichfeldt$.

The output of $\Blichfeldt$ is either an $\vecs_{\bvec{x}} = \begin{bmatrix}
\bvec{x} \\  \circuit{C}(\bvec{x}) \end{bmatrix} \in S \cap \lat$ that (by
construction of $\lat$) implies $\circuit{C}(\bvec{x}) = \veczero$, or two
different elements of $S$, $\vecs_{\bvec{x}} = \begin{bmatrix} \bvec{x} \\
\circuit{C}(\bvec{x}) \end{bmatrix}$, $\vecs_{\bvec{y}} = \begin{bmatrix}
\bvec{y} \\  \circuit{C}(\bvec{y}) \end{bmatrix}$ with $\vecs_{\bvec{x}} -
\vecs_{\bvec{y}} \in \lat$ that implies $\bvec{x} \neq \bvec{y}$ and (by
construction of $\lat$) $\circuit{C}(\bvec{x}) = \circuit{C}(\bvec{y})$.

\medskip

\subsubsection{\texorpdfstring{$\boldsymbol{\cSIS}$}{$\cSIS$} is
\texorpdfstring{$\boldsymbol{\PPP}$}{$\PPP$}-complete.}
Part of the input to $\Blichfeldt$ is represented with a value function which
requires a small circuit. As we explained before this makes $\Blichfeldt$ a
non-natural problem with the respect to the definition of naturality in the
context of the complexity of total search problems. We now introduce a natural
problem that we call \emph{constrained Short Integer Solution} ($\cSIS$), and
show that it is $\PPP$-complete. The $\cSIS$ problem is a generalization of the
well-known \emph{Short Integer Solution} ($\sis$) problem, and discuss their
connection in Section~\ref{sec:ac-HcSIS}.

The input is $\matA \in \Z_q^{n \times m}$, $\matG \in \Z_q^{d \times m}$ and
$\vecb \in \Z_q^{d}$, for some positive integer $q$ and $m \geq
(n+d)\ceil{\log(q)}$. The matrix $\matG$ has the property that for every $\vecb$
we can efficiently find an $\vecx \in \binset^m$ such that $\matG \vecx = \vecb
\pmod{q}$. We define such matrices as \emph{binary invertible}. The output is
either a vector $\vecx \in \binset^m$ such that $\matA \vecx = \matzero
\pmod{q}$ and $\matG \vecx = \vecb \pmod{q}$, or two different vectors
$\vecx,\vecy \in \binset^m$ such that $\matA (\vecx - \vecy) = \matzero
\pmod{q}$ and $\matG \vecx = \matG \vecy= \vecb \pmod{q}$. We give a proof
overview of the next theorem, and a full proof in Section~\ref{sec:cSIS}.

\begin{reptheorem}{thm:cSISPPPcompleteness}
	The $\cSIS$ problem is $\PPP$-complete.
\end{reptheorem}
\smallskip

\noindent \textsc{$\PPP$ Membership of $\cSIS$ Overview.}
We show the membership of $\cSIS$ in $\PPP$ for the general class of binary
invertible matrices $\matG$ in Section~\ref{sec:cSIS}. In order to simplify the
exposition, we assume that $q= 2^{\ell}$ and $\matG$ to be the ``gadget'' matrix
concatenated with a random matrix $\matV$. That is, $\matG$ has the form
$\left[\matI_d \otimes \vec{\gamma}^T \,\,\, \matV\right]$ where $\vec{\gamma}^T
= [1, 2, \ldots, 2^\ell]$.

Let $\matA \in \Z_q^{n \times m}$, $\matG \in \Z_q^{d \times m}$, and $\vecb \in
\Z_q^d$ be the input to $\cSIS$. We now explain why $m \geq (n+d)\ell$ suffices
to always guarantee a solution to $\cSIS$. First, observe that the first $\ell
\cdot d$ columns of $\matG$, corresponding to the gadget matrix $[\matI_d
\otimes \vec{\gamma}^T]$, are enough to guarantee that for every $\vecr'\in
\Z_q^{m-\ell d}$ there exists an $\vecr$ such that $\matG\begin{bmatrix} \vecr
\\ \vecr' \end{bmatrix} = \vecb \pmod{q}$. Hence, there are at least $q^{m- \ell
d}$ solutions to the equation $\matG \vecx = \vecb \pmod{q}$. Also, there are
$2^{\ell n}$ possible values for $\matA \vecx \pmod{q}$. By a pigeonhole
argument a solution to $\cSIS$ always exists. To complete the membership proof,
issues similar to $\Blichfeldt$ with the circuit representation of the problem
instance appear, but we overcome them using similar ideas.
\medskip

\noindent \textsc{$\PPP$ Hardness of $\cSIS$ Overview.} We start with a circuit
$\circuit{C}: \binset^n \rightarrow \binset^n$ that is an input to $\pigeon$.
Since the input of $\pigeon$ is a circuit and the input of $\cSIS$ is a pair of
matrices and a vector, we need to represent this circuit in an algebraic way. In
particular, we device a way to encode the circuit in a binary invertible matrix
$\matG$ and a vector $\vecb$. To gain a better intuition of why this is
possible, we note that a $\mathsf{NAND}$ gate $x \gnand y = z$ can be expressed
as the linear modular equation $x + y + 2z -w = 2 \pmod{4}$, where $x,y,z,w \in
\binset$. By a very careful construction, we can encode these linear modular
equations in a binary invertible matrix $\matG$. For further details we defer to
Section~\ref{sec:cSIS}.

Since $\cSIS$ with $q = 4$ returns a vector such that $\matA \vecx = \vec{0}
\pmod{4}$ and $\pigeon$ asks for a binary vector such that $\circuit{\bvec{x}} =
\bvec{0}$, a natural idea is to let $\matA$ be of the form $\left[\vec{0}
\,\,\,\matI_n\right]$, where the identity matrix corresponds to the columns
representing the output of circuit $\circuit{C}$ in $\matG$. Finally, we argue
that a solution to $\cSIS$ with input $\matA, \matG$ and $\vecb$ as constructed
above, gives either a collision or a preimage of zero for the circuit
$\circuit{C}$ as required.

It can be argued that this reduction shares common ideas with the reduction of
$3\mdash\mathsf{SAT}$ to $\mathsf{SUBSET\mdash SUM}$; this shows the importance
of the input conditions for $\cSIS$ and hints to the numerous complications that
arise in the proof. Without these conditions, we could end up with a trivial
reduction to an $\NP$-hard problem! Fortunately, we are able to show that our
construction satisfies the input conditions of $\cSIS$.

\medskip

\subsubsection{Towards a Natural and Universal Collision-Resistant Rash Family.}
$\wPPP$ is a subclass of $\PPP$ in which a collision always exists; it is not
hard to show that variations of both $\Blichfeldt$ and $\cSIS$ are
$\wPPP$-complete. We tweak the parameters of valid inputs in order to guarantee
that a collision always exists. The $\wPPP$-complete variation of $\cSIS$, which
we denote by $\weakcSIS$, gives a function family which is a universal
collision-resistant hash function family in a worst-case sense: if there is a
function family that contains at least one function for which it is hard to find
collisions, then our function family also includes a function for which it is
hard to find collisions.

We now describe the differences of $\cSIS$ and $\weakcSIS$. As before we assume
that $q = 2^\ell$. The input to $\weakcSIS$ is a matrix $\matA \in \Z_q^{n
\times m}$, and a binary invertible matrix $\matG \in \Z_q^{d \times m}$. Notice
that there is no vector $\vecb$ in the input,  and the relation between $n,m,d$
and $\ell$ is that $m$ has to be \emph{strictly} greater that $\ell (n + d)$.
Namely, $m > \ell(n + d)$. This change in the relation of the parameters might
seem insignificant, but is actually very important, as it allows us to replace
$\vecb$ in $\cSIS$ by the zero vector. This transforms $\weakcSIS$ into a pure
lattice problem: on input matrices $\matA,\matG$ with corresponding bases
$\matB_{\matA}$ and $\matB_{\matG}$, where $\matG$ is binary invertible, find
two vectors $\vecx$ and $\vecy$ such that $\vecx, \vecy \in \lat(\matB_{\matG})$
and $\vecx - \vecy \in \lat(\matB_{\matA})$.

The great resemblance of $\weakcSIS$ with $\sis$ and its completeness for
$\wPPP$ lead us to the first candidate for a universal collision-resistant hash
function $\HcSIS = \{ \h_{\bvec{s}}: \binset^k \rightarrow \binset^{k'}\}$:
\begin{Enumerate}
\item[-] The key $\bvec{s}$ is a pair of matrices $(\matA, \matG)$, where $\matG$
is binary invertible.
\item[-] Given a key $\bvec{s} = (\matA, \matG)$ and a binary vector $\bvec{x} \in
\binset^k$, $\h_{\bvec{s}}(\bvec{x})$ is the binary decomposition of $\matA
\vecu \pmod{q}$,  where $\vecu= \begin{bmatrix} \vecr \\ \bvec{x} \end{bmatrix}$
such that  $\matG\vecu = \vec{0} \pmod{q}$.
\end{Enumerate}

Because lattice problems have worst-to-average case reductions and our hash
family is based on a lattice problem, this gives hope for showing that our
construction is universal in the average-case sense.

\medskip

\subsubsection{Other Lattice Problems Known to be in
\texorpdfstring{$\boldsymbol{\PPP}$}{$\PPP$}.}
We show that the computational analog of Minkowski's theorem, namely
$\Minkowski$, is in $\PPP$ via a Karp-reduction to $\Blichfeldt$. We note that a
Karp-reduction showing $\Minkowski \in \PPP$ was shown in~\cite{BanJPPR15}.
Based on these two problems and the known reductions between lattice problems,
we conclude that a variety of lattice (approximation) problems belong to $\PPP$;
the most important among them are $n \mdash \svp $, $\tilde{O}(n)$-$\sivp$ and
$n^{2.5} \mdash \cvp $ (see Figure~\ref{fig:problemDiagram}).

\subsubsection{Other Cryptographic Assumptions in
\texorpdfstring{$\boldsymbol{\PPP}$}{$\PPP$}.}
By the definition, the class $\wPPP$ contains all cryptographic assumptions that
imply collision-resistant hash functions. These include the factoring of Blum
integers, the Discrete Logarithm problem over $\Z_p^*$ and over elliptic curves,
and the $\sis$ lattice problem (a special case of $\weakcSIS$). Also,
Je\v{r}\'abek~\cite{Jerabek16} showed that the problem of factoring integers is
in $\wPPP$.

We extend the connection between $\PPP$ and cryptography by introducing a
white-box model to describe \emph{general groups}, which we define to be cyclic
groups with a succinct representation of their elements and group operation
(i.e. a small circuit). We show that the Discrete Logarithm over general groups
is in $\PPP$. An example of a general group is $\Z_q^*$. These connections are
also summarized in Figure~\ref{fig:problemDiagram}.

\begin{figure}[H]
  \centering
  \includegraphics[scale=0.3]{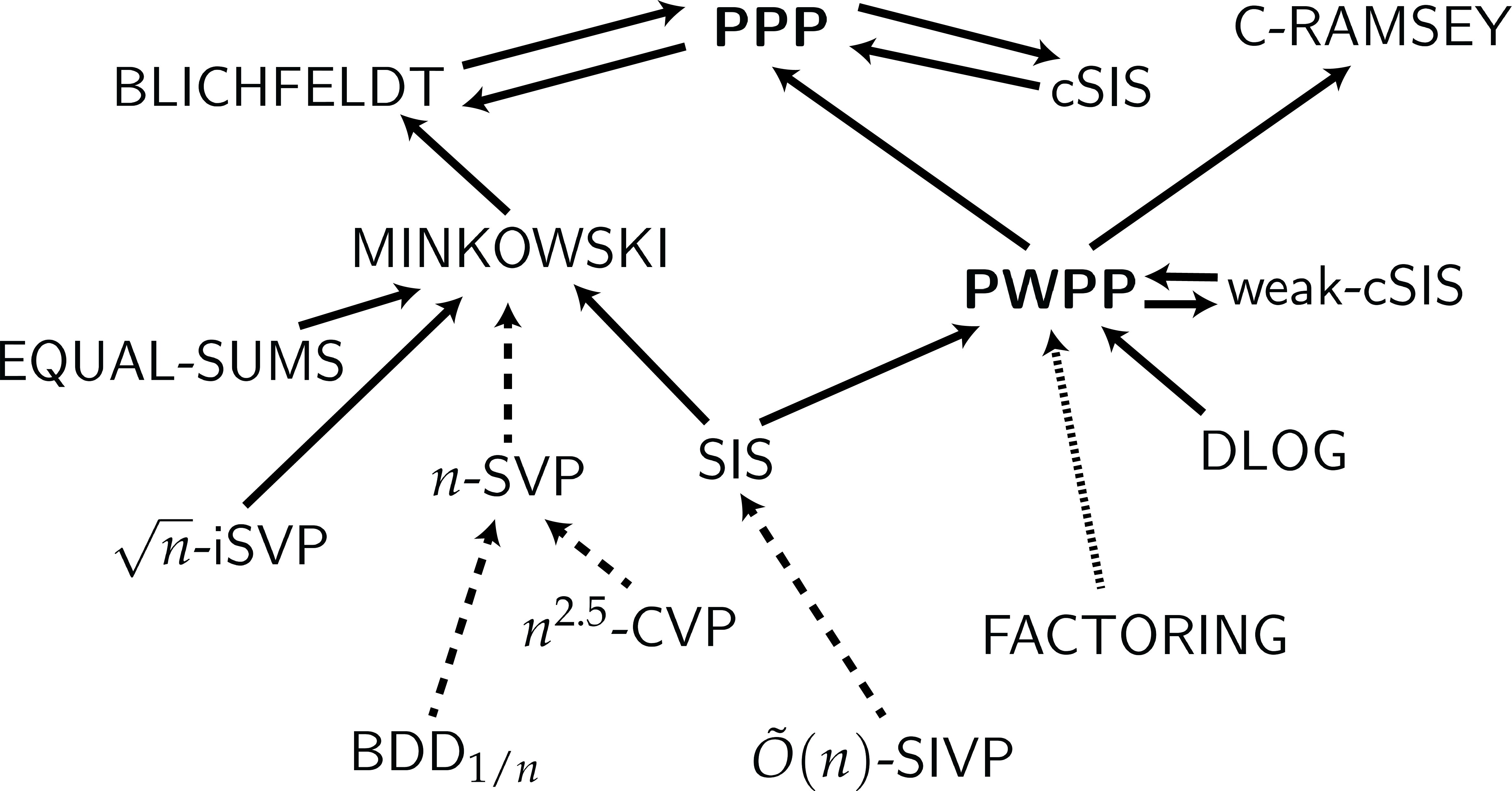}
  \caption{Solid arrows denote a Karp reduction, and dashed arrows denote a Cook
  reduction.}
  \label{fig:problemDiagram}
\end{figure}

\medskip

\subsection{Open questions.} Numerous new questions arise from our work and the
connections we draw between $\PPP$, cryptography and lattices. We summarize here
some of them.

\begin{openProblem}
  \textit{Is there a worst-to-average case reduction from $\weakcSIS$ to
  itself?} \\
  This result will provide the first \textit{natural}, in the sense that does
  not invoke explicitly a Turing machine in the input, and \textit{universal}
  collision resistant hash function family.
\end{openProblem}

\begin{openProblem}
  \textit{Is $\sis$ or $\Minkowski$ $\PPP$-hard?}
\end{openProblem}

\begin{openProblem}
  \textit{Is $\gamma$-$\svp$ in $\PPP$ for $\gamma = o(n)$?}
\end{openProblem}

\begin{openProblem}
  \textit{Is $\gamma$-$\cvp$ $\PPP$-hard for $\gamma = \Omega(\sqrt{n})$?}
\end{openProblem}

\begin{openProblem} \label{op:ellipticCurves}
  \textit{Is the discrete logarithm problem in $\PPP$ for general elliptic
  curves?}
\end{openProblem}

\section{Preliminaries} \label{sec:model}

\paragr{General Notation.} Let $[m]$ be the set $\{0, \dots, m - 1\}$, $\N =
\{0,1,2,\dots\}$ and  $\Z_+ = \{1,2,3,\dots\}$ We use small bold letters
$\vec{x}$ to refer to real vectors in finite dimension $\R^d$ and capital bold
letters $\matA$ to refer to matrices in $\R^{d \times \ell}$. For a matrix
$\matA$, we denote by $\veca^T_i$ its $i$-th row and by $a_{i,j}$ its $(i,j)$-th
element. Let $\matI_n$ denote the $n$-dimensional identity matrix. We denote
with $\matE_{i,j}$ the matrix that has all zeros except that $e_{i,j}=1$. A
function $\negl(k)$ is \textit{negligible} if $\negl(k)<1/k^c$ for any constant
$c>0$ and sufficiently large $k$. All logarithms $\log(\cdot)$ are in base 2.
\medskip

\paragr{Vector Norms.} We define the $\ell_p$-norm of $\vec{x} \in \reals^{d}$
to be $\norm{\vec{x}}_p = \left( \sum_i x_i^p \right)^{1/p}$ and the
$\ell_{\infty}$-norm of $\vec{x}$ to be $\norm{\vec{x}}_{\infty} = \max_{i}
\abs{x_i}$. For simplicity we use $\norm{\cdot}$ for the $\ell_2$-norm instead
of $\norm{\cdot}_2$. It is well known that $\norm{\vec{x}}_p \le n^{1/p - 1/q}
\norm{\vec{x}}_q$ for $p \le q$ and $\norm{\vec{x}}_p \le \norm{\vec{x}}_q$ for
$p > q$.
\medskip

\subsection{Complexity Classes and Reductions} \label{sec:prelims:complexity}

\paragr{Binary Strings and Natural Numbers.} We use bold and underlined small
letters to refer to binary strings. Binary strings $\bvec{x} \in \{0, 1\}^{k}$
of length $k$ can also be viewed as vectors in $\Z^k$. Every binary string
$\bvec{x} \in \{0, 1\}^k$ can be mapped to a non negative integer number through
the nonlinear map $\BitComp : \{0, 1\}^* \to \Z_+$ called \textit{bit
composition}, where $\BitComp(\bvec{x}) = \sum_{i = 0}^{k-1} \ubar{x}_{k - i}
2^i$. It is trivial to see that actually $\BitComp$ is a bijective mapping and
hence we can define the inverse mapping $\BitDecomp : \Z_+ \to \{0, 1\}^*$
called \textit{bit decomposition}, which is also trivial to compute for any
given number $m \in \Z_+$.
\smallskip

\noindent
The bit decomposition function $\BitDecomp$ is extended to integer vectors and
the result is the concatenation of the bit decomposition of each coordinate of
the vector. Similarly, this is also extended to integer matrices. Then of course
the bit composition function $\BitComp$ is no longer well defined because its
output can be either a number or a vector of numbers, but for simplicity we
still use the notation $\BitComp$ and it will be made clear from the context
whether the output is a number or a vector of numbers. When $\BitDecomp$ is
applied to a set $\{m_1, \dots, m_k\}$ the output is still a set with the bit
decomposition of each element $\{\BitDecomp(m_1), \dots, \BitDecomp(m_k)\}$.
\medskip

\paragr{Boolean Circuits.} A \emph{boolean circuit} $\circuit{C}$ with $n$
inputs and $1$ output is represented as a labeled directed acyclic graph with
in-degree at most $2$, with exactly $n$ source nodes and exactly $1$ sink node.
Each source node is an input of $\circuit{C}$ and the sink node is the output of
$\circuit{C}$. Each of the input nodes of $\circuit{C}$ is labeled with a number
in $[n]$ denoting the ordering of the input variables. Each node with in-degree
$1$ is labeled with one of the $2$ boolean functions with $1$ variable
$\{\idlabel, \neg\}$. Each node with in-degree $2$ can be labeled with one of
the $16$ boolean function with two variables, but for our purposes we are going
to use only the following five boolean functions: \textit{nand, nor, xor, and,
or}, with corresponding symbols $\{\gnand, \gnor, \gxor, \gand, \gor\}$. Every
boolean circuit defines a boolean function on $n$ variables $\circuit{C} : \{0,
1\}^n \to \{0, 1\}$. Let $\bvec{x}$ be a binary string of length $n$, i.e.
$\bvec{x} \in \{0, 1\}^n$. The value $\circuit{C}(\bvec{x})$ of the circuit on
input $\bvec{x}$ is computed by evaluating  of $\circuit{C}$ one by one in a
topological sorting of $\circuit{C}$, starting from the input nodes. Then,
$\circuit{C}(\bvec{x})$ is the value of the output node. The size
$\abs{\circuit{C}}$ of $\circuit{C}$ is the number of nodes in the graph that
represents $\circuit{C}$.
\medskip

\paragr{Circuits.} We can now define a circuit $\circuit{C}$ with $n$ inputs and
$m$ outputs as an ordered tuple of $m$ boolean circuits $\circuit{C} =
(\circuit{C}_1, \dots, \circuit{C}_m)$ which defines a function $\circuit{C} :
\{0, 1\}^n \to \{0, 1\}^m$, where $\circuit{C}(\bvec{x}) =
(\circuit{C}_1(\bvec{x}), \dots, \circuit{C}_m(\bvec{x}))$. The size
$\abs{\circuit{C}}$ of $\circuit{C}$ is equal to $\abs{\circuit{C}_1} + \cdots +
\abs{\circuit{C}_2}$. It is known that $\P \subseteq \Ppoly$
(see~\cite{arora09}), where $\Ppoly$ is the class of polynomial-sized circuits.
Thus, any polynomial time procedure we describe, implies an equivalent circuit
of polynomial size.
\medskip

\paragr{Search Problems.} A \emph{search problem} in $\FNP$ is defined by a
relation $\mathcal{R}$ that on input $x$ of size $n$ and for every $y$ of size
$\poly(n)$, $\mathcal{R}(x,y)$ is polynomial-time computable on $n$. A solution
to the search problem with input $x$ is a $y$ of size $\poly(n)$ such that
$\mathcal{R}(x,y)$ holds.

\noindent A search problem is \emph{total} if for every input $x$ of size $n$,
there exists a $y$ of size $\poly(n)$ such that $\mathcal{R}(x,y)$ holds.
\noindent The class of total search problems in $\FNP$ is called $\TFNP$.
\medskip

\paragr{Karp Reductions Between Search Problems.}  A search problem
$\mathcal{P}_1$ is \emph{Karp-reducible} to a search problem $\mathcal{P}_2$ if
there exist polynomial-time (in the input size of $\mathcal{P}_1$) computable
functions $f$ and $g$ such that if $x$ is an input of $\mathcal{P}_1$, then
$f(x)$ is an input of $\mathcal{P}_2$ and if $y$ is any solution of
$\mathcal{P}_2$ with input $f(x)$ then $g(x, f(x), y)$ is a solution of
$\mathcal{P}_1$.
\medskip

\paragr{Cook Reductions Between Search Problems.} A search problem
$\mathcal{P}_1$ is \emph{Cook-reducible} to a search problem $\mathcal{P}_2$ if
there exists a polynomial-time (in the input size of $\mathcal{P}_1$) oracle
Turing machine $\mathcal{T}$ such that if $x$ is an input of $\mathcal{P}_1$,
$\mathcal{T}$ computes a $y$ such that $y$ is a solution of $\mathcal{P}_1$
whenever all the oracle answers are solutions of $\mathcal{P}_2$. The set of all
search problems that are Cook-reducible to problem $\mathcal{P}$ is denoted by
$\mathsf{FP}^{\mathcal{P}}$.
\medskip

\paragr{The $\PPP$ Complexity Class.} The class $\PPP$ is a subclass of $\TFNP$
and consists of all search problems Karp-reducible to the following problem called
$\pigeon$.

\begin{nproblem}[\pigeon]
\textsc{Input:} A circuit $\circuit{C}$ with $n$ inputs and $n$ outputs. \\
\textsc{Output:} One of the following:
\begin{Enumerate}
  \item a binary vector $\bvec{x}$ such that $\circuit{C}(\bvec{x}) = \bvec{0}$, or
  \item two binary vectors $\bvec{x} \neq \bvec{y}$, such that
        $\circuit{C}(\bvec{x}) = \circuit{C}(\bvec{y})$.
\end{Enumerate}
\end{nproblem}

\paragr{The weak $\PPP$ Complexity Class.} The class $\wPPP$ is the set of all
search problems Karp-reducible to the following problem called $\collision$.

\begin{nproblem}[\collision]
\textsc{Input:} A circuit $\circuit{C}$ with $n$ inputs and $m$ outputs with $m < n$. \\
\textsc{Output:} Two binary vectors $\bvec{x} \neq \bvec{y}$, such that
$\circuit{C}(\bvec{x}) = \circuit{C}(\bvec{y})$.
\end{nproblem}

\subsection{Set Description Using Circuits} \label{sec:prelims:setDescription}
Let $S \subseteq \N^n$ and let $\BitDecomp(S) \subseteq \{0,1\}^k$, i.e. the
elements of $S$ can be represented using $k$ bits. As we will see later there is
an inherent connection between proofs of both the inclusion and the hardness of
$\PPP$ and the succinct representation of subsets $S$ using circuits. We define
here three such representations: the \textit{characteristic function}, the
\textit{value function} and the \textit{index function}.
\smallskip

\paragr{Characteristic Function.} We say that a circuit $\Cchara_S$ with $k$
binary inputs and one output is a characteristic function representation of $S$
if $\Cchara_S(\bvec{x}) = 1$ if and only if $\bvec{x} \in \BitDecomp(S)$.
\smallskip

\paragr{Value Function.} Let $(s, \Cvalue_S)$ be a tuple where $\Cvalue_S$ is a
circuit with $\ceil{\log (s)}$ binary inputs and $k$ outputs and $s \in \Z_+$.
Let $f_{(s, \Cvalue_S)} : [s] \to \{0, 1\}^k$ be a function such that $f_{(s,
\Cvalue_S)}(\BitComp(\bvec{x})) = \Cvalue_S(\bvec{x})$ for all $\bvec{x}$ with
$\BitComp(\bvec{x}) < s$. Then, $(s, \Cvalue_S)$ is a value function
representation of $S$ if and only if $f_{(s, \Cvalue_S)}$ is a bijective map
between $[s]$ and $\BitDecomp(S)$.  The value $\Cvalue_S(\bvec{x})$ can be
arbitrary when $\BitComp(\bvec{x}) \ge s$.
\medskip

\paragr{Index Function.} Let $(s, \Cindex_S)$ be a tuple where $\Cindex_S$ is a
circuit with $k$ binary inputs and $\ceil{\log (s)}$ outputs and $s \in \Z_+$.
Let $f_{(s, \Cindex_S)} : \BitDecomp(S) \to [s]$ be a function such that $f_{(s,
\Cindex_S)}(\bvec{x}) = \BitComp(\Cindex_S(\bvec{x}))$ for all $\bvec{x} \in
\BitDecomp(S)$. Then, $(s, \Cindex_S)$ is an index function representation of
$S$ if and only if $f_{(s, \Cindex_S)}$ is a bijective map between
$\BitDecomp(S)$ and $[s]$. The value $\Cindex_S(\bvec{x})$ can be arbitrary when
$\bvec{x} \not\in \BitDecomp(S)$.
\smallskip

Some remarks about the above definitions are in order. First, given a succinct
representation of $S$ it is computationally expensive to compute $\abs{S}$, thus
we provide it explicitly using $s$. Second, even though the input and the output
of each circuit have to be binary vectors, we abuse notation and let the input
of the index function and the characteristic function, and the output of the
value function to be a element in $\N^n$. Formally, according to the above
definitions the output of $\Cvalue_S$ is the bit decomposition of an element in
$S$, namely $\BitComp(\Cvalue_S(\bvec{x})) \in S$. In the rest of the paper, we
abuse notation and drop $\BitComp$ to denote by $\Cvalue_S(\bvec{x})$ the vector
in $S$. Similarly, we drop $\BitComp$ and $\BitDecomp$ for the characteristic
and the index functions.

To illustrate the definitions of succinct representations of sets, we explain
how to define them in the simple case of  the set $\left([0, L_1] \times \cdots
\times [0, L_n]\right) \cap \Z^n$. Although this is a simple example, it is an
ingredient that we need when we show the connection of lattice problems with the
class $\PPP$.

\begin{lemma} \label{lm:cubeSetSuccinctDescription}
    Let $L_1, \dots, L_n > 0$ and   $S = ([0, L_1] \times
  \cdots \times [0, L_n]) \cap \Z^n$. Then, the following exist:
  \begin{Enumerate}
    \item a characteristic function representation $\Cchara_S$ of $S$, where
          $\abs{\Cchara_S} = O(n \max_i \log L_i)$,
    \item a value function representation $(\Cvalue_S, s)$ of $S$, where
          $\abs{\Cvalue_S} = O(n^2 \max_i \log L_i)$ and
    \item an index function representation $(\Cindex_S, s)$ of $S$, where
          $\abs{\Cindex_S} = O(n^2 \max_i \log L_i)$.
  \end{Enumerate}
\end{lemma}

\begin{proof} Let
$\ell_1 = \floor{L_1} + 1$, $\dots$, $\ell_n = \floor{L_n} + 1$.
\par
  \begin{Enumerate}
    \item Given the bit decomposition of a vector $\vecx \in \Z^n$, we can
          easily test if $\vecx$ belongs to $S$ by checking whether
          $x_i \le \ell_i$ for all $i$. Such a comparison needs $O(\log L_i)$
          boolean gates and hence the size of $\Cchara_S$ is
          $O(n \max_i \log L_i)$.
    \item Let $k$ be the input number to our value function, the vector
          that we assing to $k$ is the $k$-th vector in the lexicographical
          ordering of the elements in $\BitDecomp(S)$. We compute this vector
          $\vecx \in \Z^n$ coordinate by coordinate. We start from $x_1$.
          Observe for any $t \in \Z_+$, the number of vectors in $S$ with $x_1 =
          t$ is equal to $\prod_{i = 2}^n \ell_i$  and hence the number of
          vectors with $x_1 \le t$ is equal to $(t + 1) \cdot \prod_{i = 2}^n
          \ell_i$. Therefore,  $x_1 = \floor{k/\prod_{i = 2}^n \ell_i}$. Then,
          the  dimension of the problem reduces by one and therefore we can
          repeat the same procedure to compute $x_2$ as the first coordinate of
          the $\left(k - \left( \prod_{i = 2}^n \ell_i \right) \cdot
          \floor{k/\prod_{i = 2}^n \ell_i} \right)$-th vector in the set $S' =
          ([0, L_2] \times \cdots \times [0, L_n]) \cap \Z^{n - 1}$. Then we
          apply the procedure recursively. Moreover, this whole task can be made
          into an iterative procedure, with a circuit of size
          $O(n^2 \max_i \log L_i)$.
    \item Let $\vecx \in S$ be the vector whose index we want to compute.
          The index that we assign to this vector is its position
          in the lexicographical ordering of the vectors in $\BitDecomp(S)$.
          The number of
          vectors $\vecy \in S$ with $y_1 < x_1$ that are before $\vecx$ in the
          lexicographical ordering is equal to $(x_1 - 1) \cdot \prod_{i =
            2}^n \ell_i$. Therefore, using the recursion of the
          construction of the circuit for the value function of $S$ above,
          we see that the lexicographical index can be computed by a circuit
          of size $O(n^2 \max_i \log L_i)$.
   \end{Enumerate}
\end{proof}

In the next sections, the constructions of value and index functions for sets
lie at the heart of our proofs for showing membership in $\PPP$, and we make
frequent use of the above Lemma~\ref{lm:cubeSetSuccinctDescription}.
Specifically, in Section \ref{sec:Blichfeldt}, we see that the set of integer
cosets of a lattice admits an index function that can be implemented with a
polynomial sized circuit, and in Section \ref{sec:cryptoHard} we see that any
cyclic group admits a value function that can be implemented with a
polynomial-size circuit. In the latter case, we demonstrate that an efficient
implementation of an index function of a group shows that the Discrete Logarithm
problem for this group belongs to the class $\PPP$.

\subsection{Lattice Basics} \label{ssec:modelLattices}

\paragr{Lattice.}
A $n$-dimensional \emph{lattice} $\lat \subset \reals^n$ is the set of all
integer linear combinations of $d$ linearly independent vectors $\matB =
(\vecb_1,\ldots,\vecb_d)$ in $\reals^n$, $\lat = \lat(\matB) = \left\{
\sum_{i=1}^{d} a_i \vecb_i ~:~ a_i \in \integer \right\}$. The integer $d$ is
the rank of the lattice, and if it is equal to the dimension $n$ we refer to
$\lat$ as \emph{full-rank}. The matrix $\matB$ is called the \emph{basis} of
lattice $\lat(\matB)$. Unless explicitly stated, the lattices on all
definitions, statements and proofs are assumed to be full-rank integer lattices.
Our results can be easily extended to $d$-rank lattices, but for ease of
exposition we present only the full-rank case. A useful lemma we will use is the
following:
\begin{lemma} \label{lm:unimodularBasis}
	  Let $\matB \in \Z^{n \times n}$, $\matU \in \Z^{n \times n}$ be a
	unimodular matrix, i.e. $\det(\matU) = \pm1$, then
  $\lat(\matB) =	\lat(\matB \matU)$.
\end{lemma}

\paragr{Smith Normal Form.}
A matrix $\matD \in \Z^{n\times n}$ is in \emph{Smith Normal Form (SNF)} if it
is diagonal and $d_{i+1,i+1}$ divides $d_{i,i}$ for $1 \leq i < n$. Moreover,
any non-singular matrix $\matA \in \Z^{n\times n}$ can be written as $\matA =
\matU \matD \matV$, where $\matD \in  \Z^{n\times n}$ is a unique matrix in SNF,
and $\matU,\matV \in  \Z^{n\times n}$ are unimodular matrices. Also, the
matrices $\matU,\matD,\matV$ can be computed in polynomial-time in
$n$~\cite{KaltofenV05}.
\smallskip

\paragr{Dual Lattice.} The \emph{dual lattice} $\dual$ is defined as the set of
all vectors in the span of $\matB$ that have integer inner product with $\lat$.
That is, $\dual = \{\vecy \in \spn(\matB) ~:~ \forall \vecx \in
\lat,~\inner{\vecy, \vecx} \in \integer \}$.
\smallskip

\paragr{$q$-ary Lattice.} A lattice $\lat \subseteq \Z^n$ is called a $q$-ary
lattice if $(q\integer)^n \subseteq \lat$. Let $\matA \in \Z^{m \times n}$,
 we define the following two types of $q$-ary lattices
\begin{align} \label{eq:qaryLatticesDefinition}
	\lam_q(\matA)     & = \left\{\vecx \in \integer^n \mid \vecx^T = \vecy^T
	\matA \pmod{q} \text{ where } \vecy \in \Z^m\right\}, \\
	\lamperp_q(\matA) & = \left\{\vecx \in \integer^n \mid \matA \vecx =
	\matzero
	\pmod{q}\right\}.
\end{align}
\noindent The following lemma for $q$-ary lattices is useful in our proofs. For
a simple proof of this Lemma \ref{lm:determinantQary} we refer to Sections 2.3
and 2.4 of \cite{AlwenP11}.

\begin{lemma} \label{lm:determinantQary}
    Let $\matA \in \Z_q^{m \times n}$, then:
  \begin{Enumerate}
    \item $\lamperp_q(\matA) = q \lam_q^*(\matA)$ and
          $\det\left(\lamperp_q(\matA)\right) \le q^n$,
    \item there exists a polynomial-size circuit $\circuit{BS}$ that on input
      $\BitDecomp(\matA)$ it outputs $\BitDecomp(\matB)$ such that
          $\matB \in \Z^{n \times m}$ and $\lamperp_q(\matA) = \lat(\matB)$.
  \end{Enumerate}
\end{lemma}

\paragr{Fundamental Parallelepiped.} The fundamental parallelepiped of
$\lat(\matB)$ is defined as the set $\piped(\matB)= \{\sum_{i=1}^{n} t_i \vecb_i
~:~ 0 \leq t_i < 1 \}$. Given a full-rank lattice $\lat(\matB)$, we can define
the operator $\pmod{\piped(\matB)}$ on vectors in $\R^n$ such that $\vecy =
\vecx \pmod{\piped(\matB)}$ if $\vecy = \matB \left( \matB^{-1} \vecx -
\floor{\matB^{-1} \vecx} \right)$.
\smallskip

\paragr{Determinant.} The \emph{determinant} of a lattice is the volume of the
fundamental parallelepiped, $\det(\lat) = \sqrt{\det(\matB^T \matB)}$ or simply
$|\det(\matB)|$ for full-rank lattices.
\smallskip

\paragr{Lattice Cosets.} For every $\vecc \in \reals^n$ we define the lattice
coset as $\lat + \vecc = \{ \vecx + \vecc ~:~ \vecx \in \lat \}$. Let
$\matB \in \Z^{n \times n}$ be a set of $n$ linearly independent
integer vectors. We define the set of \textit{integer cosets} of
$\lat(\matB)$ as follows $\coset(\lat(\matB)) = \{\lat(\matB) + \vecc \mid
\vecc \in \Z^n\}$.

\noindent We now state a fundamental relation between
$\coset(\lat(\matB))$, $\det(\lat(\matB))$ and
$\piped(\lat(\matB))$.
\begin{proposition} \label{prop:cosetsDeterminantRelation}
    Let $\matB \in \Z^{n \times n}$ be a set of $n$-dimensional linearly
  independent vectors. It holds that
  $\det(\lat(\matB)) = \abs{\coset(\lat(\matB))} = \abs{\piped(\lat(\matB))
  \cap \Z^d}$.
\end{proposition}
\noindent For a proof of Proposition \ref{prop:cosetsDeterminantRelation} see
the proof of Lemma \ref{lm:inclusion-Blichfeldt}.

\section{\texorpdfstring{$\boldsymbol{\Blichfeldt}$}{$\Blichfeldt$} is
\texorpdfstring{$\boldsymbol{\PPP}$}{$\PPP$}-Complete} \label{sec:Blichfeldt}

  The concept of lattices was introduced by Hermann Minkowski in his influential
book \textit{Geometrie der Zahlen} \cite{Minkowski1910}, first published at 1896.
In his book, Minkowski developed the theory of the geometry of numbers and
resolved many difficult problems in number theory. His fundamental theorem,
known as \emph{Minkowski's Convex Body Theorem}, was the main tool of these
proofs.

  Despite the excitement created by Minkowski's groundbreaking work, it was
only after 15 years that a new principle in geometry of numbers was discovered.
The credit for this discovery goes to Hans Frederik Blichfeldt, who in 1914
published a paper \cite{Blichfeldt1914} with his new theorem and some very
important applications in number theory \footnote{These introductory paragraphs
were inspired from Chapter 9 of \cite{OldsLD2001}.}.

  In this section, we characterize the computational complexity of Blichfeldt's
existence theorem and in later sections we discuss its applications (see
Section~\ref{sec:lattices}).
We recall the statement of Blichfeldt's theorem below, introduce its
computational search version, and prove that it is $\PPP$-complete.

\begin{theorem}[Blichfeldt's Theorem \cite{Blichfeldt1914}]
  \label{thm:BlichfeldtTheorem}
  Let $\matB \in \Z^{n \times n}$ be a set of $n$-dimensional linearly
independent integer vectors and a measurable set $S \subseteq \R^n$. If
$\vol(S) > \det(\lat(\matB))$, then there exist $\vec{x}, \vec{y} \in S$ with
$\vec{x} \neq \vec{y}$ and $\vec{x} - \vec{y} \in \lat(\matB)$.
\end{theorem}

\noindent A proof of Theorem \ref{thm:BlichfeldtTheorem} can be found in
Chapter 9 of \cite{OldsLD2001}. We now define the computational discrete version
of Blichfeldt's theorem.
\medskip

\begin{nproblem}[\Blichfeldt]
\textsc{Input:} An $n$-dimensional basis $\matB \in \Z^{n \times n}$ and a set
$S \subseteq \Z^n$ described by the value function representation
$(s, \Cvalue_S)$. \\
\textsc{Output:} If $s < \det(\lat(\matB))$, then the vector
$\vec{0}$. Otherwise, one of the following:
\begin{Enumerate}
	\setcounter{enumi}{-1}
	\item a number $z \in [s]$ such that  $\Cvalue_S(z) \not\in S$ or two
	numbers $z,w \in [s]$ such that $\Cvalue_S(z) = \Cvalue_S(w)$,
  \item a vector $\vecx$ such that $\vecx \in S \cap \lat$,
  \item two vectors $\vecx \neq \vecy$, such that $\vecx, \vecy \in S$ and
        $\vecx - \vecy \in \lat$.
\end{Enumerate}
\end{nproblem}

\begin{lemma} \label{lm:inclusion-Blichfeldt}
  $\Blichfeldt$ is in $\PPP$.
\end{lemma}

\begin{proof}
We show a Karp reduction from $\Blichfeldt$ to the $\pigeon$ problem.
\smallskip

\noindent Let $\matB$ and $(s, \Cvalue_S)$ be the inputs of $\Blichfeldt$. We
define the lattice $\lat = \lat(\matB)$ and $R = \piped(\matB) \cap \Z^n$ to be
the set of all the integer vectors in the fundamental parallelepiped
$\piped(\matB)$. Let $\ell = \ceil{\log(\det(\lat))}$ and $m = \ceil{\log(s)}$.
We remark that $\ell = \ceil{\log(\abs{R})}$ (see Proposition
\ref{prop:cosetsDeterminantRelation}) and $m$ is equal to the number of binary
inputs of $\Cvalue_S$. Our goal is to construct a circuit $\circuit{C}$ with
$\ell$ inputs and $\ell$ outputs such that given a solution of the $\pigeon$
problem with input $\circuit{C}$, we efficiently find a solution of
$\Blichfeldt$. If $s < \det(\lat)$, then we output $\vec{0}$ without invoking
$\pigeon$. Therefore, for the rest of the proof we focus on the case $s \geq
\det(\lat)$ which implies $m \ge \ell$. We can also assume without loss of
generality that $s = \det(\lat)$. If this is not the case, then we can
equivalently work with the set $S'$ which is described by the value function
representation $(\det(\lat), \Cvalue'_S)$, where $\Cvalue'_S$ has $\ell$ inputs
and outputs the value of $\Cvalue_S$ after padding the input with $m - \ell$
inputs in the most significant bits. We do this so that the input conditions of
$\pigeon$ are satisfied. For simplicity, we assume $S = S'$ for the rest of the
proof.

Before defining $\circuit{C}$, we note that numbers that satisfy case ``0.'' in
the output of $\Blichfeldt$ exist if and only if $(s,\Cvalue_S)$ is not a valid
value function of $S$. Therefore, an output of case ``0.'' certifies that the
input of $\Blichfeldt$ is not a valid input.

First, let us assume that $\det(\lat) = 2^\ell$. Intuitively, we want
$\circuit{C}$ to be a mapping from any vector $\vecx \in S$ to the integer coset
of $\vecx$ in $\coset(\lat)$. A useful observation is that if two vectors in $S$
belong to the same coset, then their difference belongs to $\lat$. The first
step to implement this idea is to find a set of representatives for
$\coset(\lat)$. The set $R$ of integer vectors in $\piped(\matB)$ is such a set.
Finally, to finish the construction of $\circuit{C}$ we need an index function
representation of $R$, so that the output of $\circuit{C}$ is in $\{0,
1\}^{\ell}$. When $\det(\lat) \neq 2^\ell$ we need to make some simple changes
in the reduction, we decribe these changes at the end of the proof.

The main difficulty that we encounter in formalizing the above intuition hides
in the construction of a polynomial-size circuit $\Cindex_R$ for the index
function of the set $R$. Let $r = \abs{R}$. We define the circuit $\Cindex_R:
\BitDecomp(R) \to \{0,1\}^\ell$ such that $\Cindex_R$ defines a bijective map
between $\BitDecomp(R)$ and $\BitDecomp([r])$. The circuit $\Cindex_R$ first
computes the Smith Normal Form of the basis $\matB = \matU \matD \matV$, which
can be done by a circuit that has size polynomial in $\abs{\BitDecomp(\matB)}$
\cite{KaltofenV05}, then uses the index function of the set $\piped(\matD) \cap
\integer^n$, as defined in Lemma~\ref{lm:cubeSetSuccinctDescription}, to map
each element of $R$ to a number in $[r]$.

\begin{claim} \label{claim:R-Dbijection}
    There exists a bijection $\vec{\pi}$ between $R$ and $\piped(\matD) \cap
  \integer^n$, which can be implemented by a polynomial-size circuit.
\end{claim}

\begin{proof}[Proof of Claim \ref{claim:R-Dbijection}.]
  In the Smith Normal Form of $\matB$, we have that $\matU \in \integer^{n
  \times n}$, $\matV  \in \integer^{n \times n}$ are unimodular matrices and
  $\matD \in \integer^{n \times n}$ is a diagonal matrix. Since $\matV$ is
  unimodular by Lemma~\ref{lm:unimodularBasis}, $\lat(\matD \matV) =
  \lat(\matD)$ and, hence, $\matD$ is a basis of $\lat(\matD \matV)$. This
  implies that the function $\vec{\phi}(\vecx) = \vecx \pmod{\piped(\matD)}$ is
  a bijection between $\piped(\matD \matV ) \cap \Z^n$ and $\piped(\matD ) \cap
  \Z^n$, with inverse map $\vec{\phi}^{-1}(\vecy) = \vecy \pmod{\piped(\matD
  \matV)}$. Observe that both $\vec{\phi}$ and $\vec{\phi}^{-1}$ can be
  implemented with a polynomial-size circuit. By the unimodularity of $\matU$,
  the map $\vec{h}(\vecx) = \matU^{-1} \vecx$ defines a bijection between
  $\piped(\matB) \cap \Z^n$ and $\piped(\matD \matV) \cap \Z^n$, with
  $\vec{h}^{-1}(\vecy) = \matU \vecy$. As above, both $\vec{h}$ and
  $\vec{h}^{-1}$ can be implemented by polynomial-size circuits. Hence, the
  function $\vec{\pi}(\vecx) = \vec{h}(\vec{\phi}(\vecx))$ is a bijection
  between $R = \piped(\matB) \cap \Z^n$ and $\piped(\matD) \cap \Z^n$ with
  inverse map $\vec{\pi}^{-1}(\vecy) = \vec{\phi}^{-1}(\vec{h}^{-1}(\vecy))$
  \footnote{The fact that $\vec{\pi}$ is a bijection shows that
  $\abs{\piped(\matD) \cap \Z^n} = \abs{\piped(\matB) \cap \Z^n}$.}. The size of
  the circuits that implement $\vec{\pi}$ and $\vec{\pi}^{-1}$ is also
  polynomial.
\end{proof}

\begin{proof}[Proof of Proposition~\ref{prop:cosetsDeterminantRelation}]
  By the unimodularity of $\matU$ and $\matV$ we have that $\det(\matD) =
  \det(\matB)$. Since $\matD$ is diagonal, the number of integer points in
  $\piped(\matD)$ is equal to $\det(\lat(\matD))$. Hence, $\abs{\piped(\matB)
  \cap \Z^n} = \abs{\piped(\matD) \cap \Z^n} = \det(\lat(\matD)) =
  \det(\lat(\matB))$.
\end{proof}

\noindent Now we are ready to describe $\Cindex_R$. Let $\matD =
\text{diag}(d_1, \dots, d_n)$, observe that \[ R_{\matD} = \{ \vecx \in \Z^n
\text{ such that } \vecx \in \piped(\matD)\} = \left([0, d_1] \times \cdots
\times [0, d_n] \right) \cap \Z^n. \] \noindent We use Lemma
\ref{lm:cubeSetSuccinctDescription} to construct an index function
$\Cindex_{R_{\matD}}$. Finally, we define the index function of $R$ as
$\Cindex_R(\vecx) = \Cindex_{R_{\matD}} \left( \vec{\pi}(\vecx)) \right)$.

Let $\vec{\sigma} : \Z^n \to R$ be the function that computes the modulo
$\piped(\matB)$, i.e. $\vec{\sigma}(\vecx) = \vecx \pmod{\piped(\matB)}$. Now,
we have all the components to define our circuit $\circuit{C}$ to reduce
$\Blichfeldt$ to $\PPP$. The circuit $\circuit{C}$ takes as input a boolean
vector $\bvec{x} \in \{0,1\}^\ell$ and computes
$\Cindex_{R}\left(\vec{\sigma}\left( \Cvalue_S(\bvec{x}) \right)\right)$, where
$\Cindex_{R}$ is as defined above. Namely, the circuit $\circuit{C}$ first
computes the vector $\vecx = \Cvalue_S(\bvec{x})$, which belongs in $S$, then it
computes a vector $\vecc \in \piped(\matB)$ such that $\vecx \in \vecc + \lat$,
maps it to a vector in $\piped(\matD)$ and lastly maps this vector to a boolean
vector using the index function of $R_{\matD}$. Since we assume $\abs{R} =
2^\ell$ we have that $\circuit{C}$ has $\ell$ inputs and $\ell$ outputs and
hence it is a valid input to $\pigeon$. Any solution of this $\pigeon$ instance
gives a solution to our $\Blichfeldt$ instance. The $\pigeon$ with input
$\circuit{C}$ returns one of the following:

\begin{enumerate}
  \item a boolean vector $\bvec{x} \in \{0,1\}^\ell$ such that
  $\circuit{C}(\bvec{x}) = \bvec{0}$.

  If $\Cvalue_S(\BitComp(\bvec{x})) \not\in S$, then $\BitComp(\bvec{x})$ is a
  solution to our initial $\Blichfeldt$ instance. Otherwise, let $\vec{y} =
  \vec{\sigma}\left(\Cvalue_S( \bvec{x})\right)$ and $\bvec{y} =
  \BitDecomp(\vec{y})$. In this case, we have that   $\Cindex_{R}(\bvec{y}) =
  \bvec{0}$, which implies that $\Cindex_{R_{\matD}}\left( \vec{\pi}(\vec{y})
  \right) = \bvec{0}$. From the definition of $\Cindex_{R_{\matD}}$ in Lemma
  \ref{lm:cubeSetSuccinctDescription}, we get that $\vec{\pi}(\vecy) = \vec{0}$
  and hence $\vecy = \vec{0}$. Finally,
  $\vec{\sigma}\left(\Cvalue_S(\bvec{x})\right) = \vec{0}$ implies
  $\Cvalue_S(\bvec{x}) \in \vec{0} + \lat$ and so $\vecx = \Cvalue_S(\bvec{x})$
  is a solution to our initial $\Blichfeldt$ instance.

  \item two boolean vectors $\bvec{x}, \bvec{y} \in \{0,1\}^\ell$, such that
  $\circuit{C}(\bvec{x}) = \circuit{C}(\bvec{y})$.

  Using the same reasoning as in the previous case we conclude that either
  $\BitComp(\bvec{x}), \BitComp(\bvec{y})$ is a solution to our initial
  $\Blichfeldt$ instance or there exists a $\vecc \in R$ such that
  $\Cvalue_S(\bvec{x}), \Cvalue_S(\bvec{y}) \in \vecc + \lat(\matB)$ and hence
  if $\vecx = \Cvalue_S(\bvec{x})$, $\vecy =\Cvalue_S(\bvec{y})$ then we have
  that $\vecx - \vecy \in \lat$ and so $\vecx, \vecy$ is a solution to our
  initial $\Blichfeldt$ instance.
\end{enumerate}

\noindent The only thing left to finish the proof is the case in which
$\det(\lat(\matB)) < 2^{\ell}$. In this case the circuit $\circuit{C}$ is
defined as
\begin{equation*}
  \circuit{C}(\bvec{x}) = \left\{ \begin{split}
                                     \bvec{x} & ~~~~~~ \text{if }\BitComp(\bvec{x}) \ge
                                       \det(\lat(\matB)) \\
                                       \Cindex_{R}\left(
                                     \vec{\sigma}\left(\Cvalue_S(
                                       \bvec{x}) \right)\right) &
                                       ~~~~~~ \text{if }\BitComp(\bvec{x}) <
                                       \det(\lat(\matB))
                                   \end{split} \right. .
\end{equation*}
\noindent When $\circuit{C}$ is as above, there are no solutions $\bvec{x}$ and
$\bvec{y}$ to the $\pigeon$ problem with input $\circuit{C}$ such that
$\BitComp(\bvec{x}), \BitComp(\bvec{y}) \ge \det(\lat(\matB))$. Hence, all
solutions of $\pigeon$ imply a solution to our initial $\Blichfeldt$ instance as
described before.
\end{proof}

\noindent We now proceed to show the $\PPP$-hardness of $\Blichfeldt$.

\begin{lemma} \label{lm:hardness-Blichfeldt}
  $\Blichfeldt$ is $\PPP$-hard.
\end{lemma}

\begin{proof}
  We prove that $\pigeon$ is reducible to the $\Blichfeldt$ problem.
  \smallskip

\noindent Let $\circuit{C}:\binset^n \rightarrow \binset^n$ be an arbitrary
instance of $\pigeon$. We construct an instance of $\Blichfeldt$ based on
$q$-ary lattices as follows. Fix $q = 2$ and let $\matA = [\vec0~~\matI_n] \in
\Z_2^{n \times 2n}$, where $\matI_n$ is the $n$-dimensional identity matrix. We
define the lattice $\lamperp_q(\matA)$ and using Lemma \ref{lm:determinantQary}
we  compute $\matB = \circuit{BS}(\matA)$ such that $\lat(\matB) =
\lamperp_q(\matA)$. Next, we define the set $S = \left\{ \begin{bmatrix}\bvec{x}
\\ \circuit{C}(\bvec{x})\end{bmatrix} \text{ such that } \bvec{x} \in \binset^n
\right\} \subseteq \Z_2^{2n}$ . Accordingly, the circuit $\Cvalue_S$  has
$\circuit{C}$ hardcoded, and on input $\bvec{x} \in \binset^n$ it outputs
$\Cvalue_S(\bvec{x}) = \begin{bmatrix}\bvec{x} \\
\circuit{C}(\bvec{x})\end{bmatrix}$, where $\begin{bmatrix}\bvec{x} \\
\circuit{C}(\bvec{x})\end{bmatrix}$ is viewed as a vector in $\Z_2^{2n}$ and not
as a binary string. Notice that $|S| = 2^{n}$ since for every $\bvec{x} \in
\binset^n$ there is one element in $S$ with its $n$-bit prefix equal to
$\bvec{x}$. Thus, the $\Blichfeldt$ instance is defined by $\matB$ and $(2^n,
\Cvalue_S)$. Notice that for $\vecy = \begin{bmatrix}\bvec{x} \\
\circuit{C}(\bvec{x})\end{bmatrix} \in S$ we have
\begin{equation} \label{eq:BlichfeldtHardness1Cosets}
  \matA \vecy = \circuit{C}(\bvec{x}) \pmod{2}.
\end{equation}

We will see that any solution to the above $\Blichfeldt$ instance gives a
solution for the $\pigeon$ problem  with input $\circuit{C}$. The problem
$\Blichfeldt$ returns one of the following:

\begin{enumerate}
  \item a single vector $\vecy =\begin{bmatrix}\bvec{x} \\
  \circuit{C}(\bvec{x})\end{bmatrix} \in S \cap \lat(\matB)$.

  Then, by \eqref{eq:BlichfeldtHardness1Cosets} it holds that
  $\circuit{C}(\bvec{x}) = \vec{0} \pmod{2}$ which means $\circuit{C}(\bvec{x})
  = \bvec{0}$ and hence $\bvec{x}$ is a solution to $\pigeon$.

  \item two vectors $\vecx, \vecy \in S$, such that $\vecx \neq \vecy$ and
  $\vecx - \vecy \in \lat(\matB)$.

  In this case, we have $\vecx = \begin{bmatrix}\bvec{x} \\
  \circuit{C}(\bvec{x})\end{bmatrix}$ and $\vecy = \begin{bmatrix}\bvec{x} \\
  \circuit{C}(\bvec{x})\end{bmatrix}$ such that $\vecx - \vecy \in
  \lamperp_q(\matA)$. Thus, $\matA(\vecx - \vecy) = \vec{0} \pmod{2}$ which by
  \eqref{eq:BlichfeldtHardness1Cosets} implies $\circuit{C}(\bvec{x}) =
  \circuit{C}(\bvec{y})$, and because $\vecx \neq \vecy$ it has to be that
  $\bvec{x} \neq \bvec{y}$. So, the solution to $\pigeon$ is the pair $\bvec{x},
  \bvec{y}$.
\end{enumerate}

Finally, we argue that $2^n \ge \det(\lat(\matB))$ so that the $\Blichfeldt$
problem does not output $\vec{0}$ trivially. This follows directly from Lemma
\ref{lm:determinantQary} since $q = 2$.
\end{proof}

Combining Lemma~\ref{lm:inclusion-Blichfeldt} and
Lemma~\ref{lm:hardness-Blichfeldt}, we prove the following theorem.

\begin{theorem} \label{thm:BlichfeldtPPPcompleteness}
	$\Blichfeldt$ is $\PPP$-complete.
\end{theorem}

\section{Constrained Short Integer Solution is
\texorpdfstring{$\boldsymbol{\PPP}$}{$\PPP$}-Complete} \label{sec:cSIS}

In this section we define the first $\PPP$-complete problem that is
\textit{natural}, i.e. does not explicitly invoke any circuit as part of its
input in contrast to the $\Blichfeldt$ problem. We call this problem the
\emph{constrained Short Integer Solution} ($\cSIS$) problem because it shares a
similar structure with the well-known and well-studied \emph{Short Integer
Solution} problem that was defined in the seminal work of
Ajtai~\cite{Ajtai1996}, and later studied
in~\cite{Micciancio04,MicRegev07,STOC:GenPeiVai08,MicciancioP12}.

To expand on the complexity theoretic importance and potential of $\cSIS$, we
note the $\sis$ problem is contained in $\PPP$ by its collision-resistance
nature but it is unknown if it is also $\PPP$-hard. This poses a fascinating
open question, since it implies a unique characterization of a concrete
cryptographic assumption using a complexity class and vice versa. We view our
result in this section (as well as in the next section) as a first step towards
this direction, by showing that $\cSIS$ is $\PPP$-complete.

Similar to the presentation of the previous section; we first define $\cSIS$,
then prove its $\PPP$ membership and finally show its $\PPP$-hardness. In all
these steps, a matrix with a special structure, that we call \textit{binary
invertible} plays an important role, and thus we define it below and prove a key
property that we use.

\begin{definition}[\textsc{Binary Invertible Matrix}] \label{def:bInv}
Let $\ell \in \Z_+$, $q \le 2^{\ell} $ and $d, k \in \N$. First, we define the
$\ell$-th \textit{gadget vector} $\vec{\gamma}_\ell$ to be the vector
$\vec{\gamma}_\ell = \begin{bmatrix} 1 & 2 & 4 & \ldots &
                                    2^{\ell - 1}
                                  \end{bmatrix}^T \in \Z_q^{\ell}$.
Second, let $\matU \in \Z_q^{d \times (d \cdot \ell)} $ be a matrix with
non-zero elements only above the $(\ell + 1)$-diagonal and $\matV \in \Z_q^{d
\times k}$ be an arbitrary matrix. We define the matrix $\matG = \begin{bmatrix}
(\matI_d \otimes \vec{\gamma}^T_\ell + \matU) & \matV \end{bmatrix} \in \Z_q^{d
\times (d \cdot \ell + k)}$ to be a \emph{binary invertible} matrix.
\end{definition}

\noindent To illustrate the form of a binary invertible matrix we give an
example below with $\ell = 3$, i.e. in $\Z_8$, where the symbol $\star$
represents any element in $\Z_8$.
\begin{align*}
  \matG = & \left.
            \begin{bmatrix}
              1 & 2      & 4 & \star & \star      & \star & \star & \star      & \star & \cdots & \star & \star      & \star & \star & \star & \cdots & \star & \star \\
              0 & 0      & 0 & 1 & 2      & 4 & \star & \star      & \star & \cdots & \star & \star      & \star & \star & \star & \cdots & \star & \star \\
              0 & 0      & 0 & 0 & 0      & 0 & 1 & 2      & 4 & \cdots & \star & \star      & \star & \star & \star & \cdots & \star & \star \\
                & \vdots &   &   & \vdots &   &   & \vdots &   & \ddots &   & \vdots &   &   &   & \vdots &   &   \\
              0 & 0      & 0 & 0 & 0      & 0 & 0 & 0      & 0 & \cdots & 1 & 2      & 4 & \star & \star & \cdots & \star & \star
            \end{bmatrix}
            ~~~~ \right\} d ~ \text{rows} \\
          & ~~~ \underbrace{\begin{matrix} ~ & ~ & ~ & ~ & ~ & ~ & ~ & ~ & ~ & ~ & ~ & ~ & ~ & ~ & ~ & ~ & ~ & ~ & ~ & ~ & ~~ \end{matrix}}_{d \cdot \ell ~ \text{columns}}
          ~~~\underbrace{\begin{matrix} ~ & ~ & ~ & ~ & ~ & ~ & ~ & ~~~ \end{matrix}}_{k ~ \text{columns}}
\end{align*}

It is evident from the definition of a binary invertible matrix, that it is not
a fixed matrix but rather a collection of matrices. That is, after we fix $q$,
the exact values of $\matG$ depend on the choice of $\matU$ and $\matV$. For
example, a special case of a binary invertible matrix is the \textit{gadget
matrix} $\matG = \matI_d \otimes \vec{\gamma}^T_\ell$ with $\matU = \vec{0}^{d
\times (d \cdot \ell)}$ and $k=0$, that was defined in \cite{MicciancioP12} and
used in many cryptographic constructions (e.g.~\cite{C:GenSahWat13, EC:BGGHNS14,
STOC:GorVaiWic15,  C:GorVaiWee15, EC:MukWic16, C:BraPer16, EC:BonKimMon17,
TCC:BTVW17, conf/pkc/PeikertS18}).

Next, we formalize the main property of binary invertible matrices that is in
the core of our proof for the inclusion of $\cSIS$ in $\PPP$, and also explains
the name ``binary invertible''.

\begin{proposition} \label{prop:bInv-property}
Let $\matG = \begin{bmatrix}
               (\matI_d \otimes \vec{\gamma}^T_{\ell} + \matU) & \matV
             \end{bmatrix} \in \Z_q^{d \times (d \cdot \ell + k)}$ be a binary
invertible matrix and $\vecr'$ be an arbitrary vector in $\Z_q^k$. Then, for
every $\vecb \in \Z_q^d$, there exists a vector $\vecr \in \{0, 1\}^{d \cdot
\ell}$ such that $\matG \begin{bmatrix} \vecr \\ \vecr' \end{bmatrix} = \vec{b}
\pmod{q}$. Additionally, the vector $\vecr$ is computable by a polynomial-size
circuit and it is guaranteed to be unique when $q = 2^{\ell}$.
\end{proposition}

\begin{proof}
For a simple illustration of the proposition, for a moment assume that $\matG =
\begin{bmatrix} \matW & \matV \end{bmatrix} \in \Z_q^{d \times (d + k)}$, where
$q$ is prime, $\matW \in \Z_q^{d \times d}$ is an upper triangular matrix, and
$\matV \in \Z_q^{k \times d}$ is arbitrary. Then, for every $\vecb \in \Z_q^d$,
using backwards substitution we can efficiently compute a vector $\vecx \in
\Z_q^d$ such that $\matG \begin{bmatrix} \vecx \\ \vecr' \end{bmatrix} = \vecb
\pmod{q}$.

For the general case where $\matG =
\begin{bmatrix}
  (\matI_d \otimes \vec{\gamma}^T_\ell + \matU) & \matV
\end{bmatrix}$, we use again backward substitution.
But because we require $\vecr$ to be binary, we make $\vecr$ to be the the
binary decomposition of the corresponding solution in $\Z_q$. More precisely, we
divide $\vecr$ into $d$ parts of $\ell$ coordinates each, such that $\vecr = [
\vecr_1 \, \dots \, \vecr_d ]^T$. We define $\vecg_i^T$ to be the $i$-th row of
the matrix $\matG$. Then, the $d$-th part of $\vecr$ is equal to $\vecr_d =
\BitDecomp\left( b_d - \vecg^T_d \begin{bmatrix} \vec{0} \\ \vecr' \end{bmatrix}
\pmod{q} \right)$. Next, we recursively compute the $t$-th part $\vecr_t$ of
$\vecr$, assuming we have already computed the parts $\vecr_{t + 1}, \dots,
\vecr_{d}$. The recursive relation for $\vecr_t$ is
\begin{align} \label{eq:backwardSubstitution}
  \vecr_t = \BitDecomp\left( b_t - \vecg^T_t
      \begin{bmatrix} \vec{0} \\ \vdots \\ \vec{0} \\ \vecr_{t + 1} \\
                      \vecr_{t + 2} \\ \dots \\ \vecr_d \\ \vecr'
      \end{bmatrix} \pmod{q} \right) \,,
\end{align}

\noindent where $\vecb = \begin{bmatrix}
                           b_1 & b_2 & \dots & b_d
                         \end{bmatrix}^T$.
The fact that $\vecr$ can be computed by a polynomial sized circuit, follows
easily from \eqref{eq:backwardSubstitution}.

In the special case of $q = 2^{\ell}$, it is easy to see that for every $x \in
[q]$ there exists a unique $\vecr_t \in \{0, 1\}^{\ell}$ such that
$\vec{\gamma}^T_{\ell} \vecr_t = x \pmod{q}$. Additionally, (when $q = 2^\ell$)
for any $\vecr_t \in \{0, 1\}^\ell$, it holds that $\vec{\gamma}^T_{\ell}
\vecr_t < q$, and thus $\vec{\gamma}^T_{\ell} \vecr_t = x$ over $\Z$. But in
this case, $\vecr_t$ is the binary decomposition of $x$, and it is unique.
Because every $\vecr_t$ is unique, we get that $\vecr$ is also unique.
\end{proof}

\begin{remark}
	The property of Proposition~\ref{prop:bInv-property} is the only property of
	binary invertible matrices that we need for our proofs. We could potentially
	define binary invertible matrices in a more general way. For example, a binary
	invertible matrix could be a permutation of the columns of a matrix of
	Definition~\ref{def:bInv}. Our results follow immediately for the	more general
	class of matrices that satisfy the properties of
	Proposition~\ref{prop:bInv-property}. Let us denote by $\mathcal{S}$ this set
	of matrices. We focus on the more restrictive case of
	Definition~\ref{def:bInv}, not only for ease of the exposition, but also
	because given a matrix $\matA$ there is no known efficient procedure to check
	whether $\matA \in \mathcal{S}$. In fact, this problem is
	$\mathsf{NP}$-complete, since we can encode a $\problem{SUBSET \mdash SUM}$
	instance in $\matA$ and reduce $\problem{SUBSET \mdash SUM}$ to checking
	whether $\matA \in \mathcal{S}$. Alternatively, we could define a promise
	version of $\cSIS$ where $\matG$, is promised to be in $\mathcal{S}$. However,
	this would deprive us from a \emph{syntactic} definition of $\cSIS$.
\end{remark}

We now define the Constrained Short Integer Solution problem.
\begin{nproblem}[\cSIS]
\textsc{Input:} A matrix $\matA \in \Z_q^{n \times m}$, a binary invertible
matrix $\matG \in \Z_q^{d \times m}$ and a vector $\vecb \in \Z_q^d$ where
$\ell \in \Z_+$, $q \leq 2^{\ell}$ and $m \geq (n + d) \cdot \ell$. \\
\textsc{Output:} One of the following:
\begin{Enumerate}
\item a vector $\vecx \in \{0,1\}^{m}$ such that $\vecx \in
\lamperp_q(\matA)$ and
  $\matG \vecx = \vecb \pmod{q}$,
\item two vectors $\vecx, \vecy \in \{0,1\}^{m}$ such that $\vecx \neq \vecy$
      with $\vecx - \vecy \in \lamperp_q(\matA)$ and
      $\matG \vecx = \matG \vecy = \vecb \pmod{q}$.
\end{Enumerate}
\end{nproblem}

\begin{lemma} \label{lm:inclusion-cSIS}
  $\cSIS$ is in $\PPP$.
\end{lemma}

\begin{proof}
We show a Karp reduction from $\cSIS$ to the $\pigeon$ problem that works for
any positive integer $q \geq 2$.
\smallskip

\noindent Let $\matA \in \Z_q^{n \times m}$, $\matG \in \Z_q^{d \times m}$,
$\vecb \in \Z_q^{d}$ be the inputs to the $\cSIS$ problem and define $k = m
- d \cdot \ell$. From the definition of $\cSIS$, we have that
$k \ge n \cdot \ell$. We define the circuit $\circuit{C}$ that takes as
input a vector $\bvec{x} \in \{0,1\}^{n \cdot \ell}$ and outputs a vector
$\bvec{z} \in \{0,1\}^{n \cdot \ell}$. For any input
$\bvec{x} \in \{0,1\}^{n \cdot \ell}$, we define the vector $\vecr'
= \begin{bmatrix} \bvec{x} \\ \vec{0}^{k - n \cdot \ell} \end{bmatrix} \in
\Z_q^k$,
and by Proposition \ref{prop:bInv-property} we compute a vector
$\vecr \in \binset^{d \cdot \ell}$ such that
$\matG \begin{bmatrix} \vecr \\ \vecr' \end{bmatrix} = \vec{b} \pmod{q}$. Let
$\circuit{C}_1 : \binset^{n \cdot \ell} \to \binset^{d \cdot \ell}$ be the
circuit that on input $\bvec{x}$ computes $\vecr$.

The circuit $\circuit{C}$ on input $\bvec{x}$, first use
$\circuit{C}_1$ to compute $\vecr=\circuit{C}_1(\bvec{x})$, and then outputs the
binary decomposition of the vector
$\matA \begin{bmatrix} \vecr \\ \vecr' \end{bmatrix} \pmod{q}$, where $\vecr'
=  \begin{bmatrix} \bvec{x} \\ \vec{0}^{k - n \cdot \ell}
\end{bmatrix}$. Overall, the description of $\circuit{C}$ is
\[ \circuit{C}(\bvec{x}) = \BitDecomp \left( \matA \begin{bmatrix}
   \circuit{C}_1(\bvec{x}) \\ \bvec{x} \\ \vec{0}^{k - n \cdot \ell}
   \end{bmatrix} \pmod{q} \right). \]
Clearly, this is a polynomial-time computation, and thus the circuit
$\circuit{C}$ is of polynomial size. To complete the proof,
we show that a solution to $\pigeon$ with input $\circuit{C}$,
gives a solution for the $\cSIS$ instance.
The output of $\pigeon$ with input $\circuit{C}$ is one of the following:
\begin{Enumerate}
\item a vector $\bvec{x} \in \{0,1\}^{n \cdot\ell}$ such that
  $\circuit{C}(\bvec{x}) = \vec{0}^{n \cdot \ell}$.

  In this case, for the vector $\vecx = \begin{bmatrix}
  \circuit{C}_1(\bvec{x}) \\ \bvec{x} \\ \vec{0} \end{bmatrix}$,
  we get $\circuit{C}(\bvec{x}) = \BitDecomp\left( \matA \vecx \pmod{q} \right)  =
  \bvec{0}$. Because the binary decomposition $\BitDecomp$ defines a bijective
  map, this implies that $\matA \vecx =  \vec{0} \pmod{q}$. Also, by the
  definition of $\circuit{C}_1$, we get that $\matG\vecx = \vecb \pmod{q}$.
  Hence, $\vecx$ is a solution of $\cSIS$ with input $(\matA, \matG, \vecb)$.

\item two vectors $\bvec{x}, \bvec{y} \in \{0,1\}^{n \cdot\ell}$,
  such that $\bvec{x} \neq \bvec{y}$ and
  $\circuit{C}(\bvec{x}) = \circuit{C}(\bvec{y})$.

  In this case, we define the vectors
  $\vecx = \begin{bmatrix} \circuit{C}_1(\bvec{x}) \\ \bvec{x} \\
  \vec{0} \end{bmatrix}$ and $\vecy = \begin{bmatrix} \circuit{C}_1(\bvec{y}) \\
  \bvec{y} \\ \vec{0} \end{bmatrix}$ such that
  \[\BitDecomp\left( \matA \vecx \pmod{q} \right) = \circuit{C}(\bvec{x}) =
    \circuit{C}(\bvec{y}) = \BitDecomp\left( \matA \vecy \pmod{q} \right). \]
  Because the binary decomposition $\BitDecomp$ defines a bijective map, this implies that
  $\matA \vecx = \matA \vecy \pmod{q}$ and hence
  $\matA (\vecx - \vecy) = \veczero \pmod{q}$.
  Therefore, $\vecx - \vecy \in \lamperp_q(\matA)$.
  Also, it has to be the case that $\vecx \neq \vecy$, because
  $\bvec{x} \neq \bvec{y}$ and by the definition of $\circuit{C}_1$ we get
  $\matG\vecx = \matG \vecy = \vecb \pmod{q}$.  Therefore, $\vecx$, $\vecy$ form a
  valid solution for the $\cSIS$ problem with input $(\matA, \matG, \vecb)$.
\end{Enumerate}
\end{proof}

\begin{lemma}\label{lm:hardness-cSIS}
  $\cSIS$ is $\PPP$-hard.
\end{lemma}

\begin{proof}
We show a Karp reduction from $\pigeon$ to the $\cSIS$ problem.
\smallskip

\noindent Let $\circuit{C} = (\circuit{C}_1, \dots, \circuit{C}_n)$ be the input
circuit to the $\pigeon$ problem with $n$ inputs and $n$ outputs. Also, let $d =
\abs{\circuit{C}}$ be the size of $\circuit{C}$. As we explained in
Section~\ref{sec:prelims:complexity}, we may assume without loss of generality,
that $\circuit{C}$ consists of gates in the set $\{\gnand, \gnor, \gxor, \gand,
\gor\}$ \footnote{In fact, it is well known that only the $\nand$ ($\gnand$)
gates suffice, but we discuss here the implementation of all these five gates,
because we are going to use them in Section~\ref{sec:collision}.}. The circuit
$\circuit{C}$ is represented as $n$ directed acyclic graphs. We first describe
how to construct from the $i$-th circuit $\circuit{C}_i$ part of a $\cSIS$
instance and then we combine these parts to form a $\cSIS$ instance. Let $d_i =
\abs{\circuit{C}_i}$ be the size of $\circuit{C}_i$ and $\graph^{(i)} =
\left(V^{(i)}, E^{(i)}\right)$ be its directed acyclic graph, where $V^{(i)}$ is
the set of nodes of $\circuit{C}_i$ . Let $\left(v^{(i)}_1, v^{(i)}_2, \dots,
v^{(i)}_{d_i}\right)$ be a topological ordering of the graph $\graph^{(i)}$,
where the first $n$ nodes are the source nodes of $\graph^{(i)}$ and the last
node is the unique sink of $\graph^{(i)}$. As we have already explained in
Section~\ref{sec:prelims:complexity}, the source nodes of $\graph^{(i)}$
correspond to the inputs of $\circuit{C}_i$ and the sink of $\graph^{(i)}$
corresponds to the output of $\circuit{C}_i$.

We denote by $\matG^{(i)}$ and $\vecb^{(i)}$ the part of the final $\cSIS$
instance that corresponds to circuit $\circuit{C}_i$. We prove our hardness
result for $\ell = 2$ and $q = 4$ \footnote{At the end of the proof, we also
show how to generalize the result for any $\ell \in \Z_+$ and $q = 2^\ell$}. We
introduce two variables for each node of $\graph^{(i)}$. The set of the first
variable in each pair represents the value of the corresponding node in the
evaluation of $\circuit{C}_i$ and we call it the set of \textit{value variables}
and the set of the second variable in each pair is the set of \textit{auxiliary
variables}. Let us remind that in the topological ordering of $\graph^{(i)}$ we
start with the $n$ input nodes $v_1, \dots, v_n$, for which we use $x_1, \dots,
x_n$ to represent their corresponding value variables. The last node of the
topological ordering is the output node, and since it is the $i$-th output of
the circuit $\circuit{C}$, we denote its value variable by $y_i$. We denote the
remaining value variables by $z^{(i)}_{n + 1}, \dots, z^{(i)}_{d_i - 1}$.
Additionally, we  denote by $w_1, \dots, w_n$ the auxiliary variables that
correspond to the input nodes, by $t_i$ the auxiliary variable of the output
node and by $r^{(i)}_{n + 1}, \dots, r^{(i)}_{d_i - 1}$ the remaining auxiliary
variables. We summarize the notation for the variables in the
Table~\ref{tbl:variablesForHardnesscSIS}. We observe that we may use $z^{(i)}_j
= x_j$ for $j \le n$ and $z^{(i)}_{d_i} = y_i$ and the same holds for the
auxiliary variables. Each column of $\matG^{(i)}$ corresponds to one of these
variables.

\begin{table}[!h]
  \centering
  \begin{tabular}{ l | c c c c c c c }
             nodes               & $v^{(i)}_1$ & $\dots$ & $v^{(i)}_n$ & $v^{(i)}_{n + 1}$ & $\dots$ & $v^{(i)}_{d_i - 1}$ & $v^{(i)}_{d_i}$ \\[3.4pt]
    \hline
             value variables     & $x_1$       & $\dots$ & $x_n$       & $z^{(i)}_{n + 1}$ & $\dots$ & $z^{(i)}_{d_i - 1}$ & $y_i$  \Tstrut  \\[3.4pt]
             auxiliary variables & $w_1$       & $\dots$ & $w_n$       & $r^{(i)}_{n + 1}$ & $\dots$ & $r^{(i)}_{d_i - 1}$ & $t_i$
  \end{tabular}
  \caption{The value and auxiliary variables that correspond to every node of the graph $\graph^{(i)}$.}
  \label{tbl:variablesForHardnesscSIS}
\end{table}

Since we focus on a fixed graph $\graph^{(i)}$, we occasionally drop the
superscript $(i)$ for simplicity. We reintroduce the superscripts when we
combine all matrices $\matG^{(i)}$ to a matrix $\matG$. Our goal is to define a
$\matG^{(i)}$ and a $\vecb^{(i)}$ such that every binary solution of
$\matG^{(i)} \vecs = \vecb^{(i)} \pmod{4}$ corresponds to a valid evaluation of
the circuit $\circuit{C}_i$.

As explained in Section~\ref{sec:prelims:complexity}, it suffices to assume that
the in-degree of every non-input node is two. Let $p_1(j)$ be the index of the
first, in the topological ordering, predecessor of the node $v_j$ and $p_2(j)$
be the index of the second. Since nodes are indexed in topological ordering we
have that $p_1(j) < p_2(j) < j$. Every row of $\matG^{(i)}$ corresponds to
a node $v_j$, with $j > n$, of $\graph^{(i)}$, and contains the
coefficients of the variables in the modular equation of a form that appears in
Table \ref{tbl:equationsForHardnesscSIS}, depending on the label of $v_j$. We
prove the correctness of these equations later in the text but it becomes
also clear from the following Claim \ref{claim:nand}. The proof of Claim
\ref{claim:nand} goes through a simple enumeration of the different values for
the boolean variables and can be found in Appendix
\ref{sec:app:proofOfClaimNand}. The equation of node $v_{d_i-j}$ defines also
the $j$-th element of $\vecb^{(i)}$ according to
Table~\ref{tbl:equationsForHardnesscSIS}.

\begin{table}[!h]
  \centering
  \begin{tabular}{ r | l | c }
             label of $v_j$ & equation of $v_j$                                     & $b_j$ \\[3.4pt]
    \hline
             $\gnand$       & $r_j + 2 z_j - z_{p_1(j)} - z_{p_2(j)} = 2 \pmod{4}$  & $2$ \Tstruts \\
             $\gnor$        & $r_j + 2 z_j - z_{p_1(j)} - z_{p_2(j)} = 3 \pmod{4}$  & $3$ \Tstruts \\
             $\gxor~$       & $z_j + 2 r_j - z_{p_1(j)} - z_{p_2(j)} = 0 \pmod{4}$  & $0$ \Tstruts \\
             $\gand~$       & $r_j + 2 z_j - z_{p_1(j)} - z_{p_2(j)} = 0 \pmod{4}$  & $0$ \Tstruts \\
             $\gor~$        & $r_j + 2 z_j + z_{p_1(j)} + z_{p_2(j)} = 0 \pmod{4}$  & $0$ \Tstruts \\
  \end{tabular}
  \caption{Forms of equation of a non-input node $v^{(i)}_j$ of the graph
           $\graph^{(i)}$, depending on its label.}
  \label{tbl:equationsForHardnesscSIS}
\end{table}

\begin{claim} \label{claim:nand}
  Let $x, y, z, w \in \binset$, then the following equivalences holds
  \begin{Enumerate}
    \item $w + 2 z - x - y = 2 \pmod{4} \Leftrightarrow x \gnand y = z, ~w = x \oplus y$
    \item $w + 2 z - x - y = 3 \pmod{4} \Leftrightarrow x \gnor  y = z, ~w = \neg (x \oplus y)$
    \item $z + 2 w - x - y = 0 \pmod{4} \Leftrightarrow x \gxor  y = z, ~w = x \gand y$
    \item $w + 2 z - x - y = 0 \pmod{4} \Leftrightarrow x \gand  y = z, ~w = x \oplus y$
    \item $w + 2 z + x + y = 0 \pmod{4} \Leftrightarrow x \gor   y = z, ~w = x \oplus y$.
  \end{Enumerate}
\end{claim}

So as we said, each column of $\matG^{(i)}$ corresponds to a variable, each row
of $\matG^{(i)}$ corresponds to an equation as described above and $\vecb^{(i)}$
is defined based on the label of each node of $\graph^{(i)}$ according to
Table~\ref{tbl:equationsForHardnesscSIS}. The order of both the rows and the
columns is specified by the topological sorting of $\graph^{(i)}$. Specifically,
the first row of $\matG^{(i)}$ describes the equation corresponding to node
$v^{(i)}_{d_i}$ (the output node of $\graph^{(i)}$), the second row of
$\matG^{(i)}$ describes the equation corresponding to node $v^{(i)}_{d_i - 1}$.
In general, the $k$-th row of $\matG^{(i)}$ describes the equation corresponding
to node $v^{(i)}_{d_i - k}$. We emphasize that, since there are no equations for
the input nodes of $\graph^{(i)}$, we have $d_i - n$ equations in total. The
order of columns follows a similar rule, i.e. it corresponds to the reverse
order of the topological ordering of $\graph^{(i)}$. The first two columns
correspond to variables $z_{d_i}$, $r_{d_i}$ of node $v_{d_i}$ followed by pairs
of rows corresponding to the variables of all non-input nodes of $\graph^{(i)}$.
Among the two columns of each node, the first corresponds to the auxiliary
variable and the second to the value variable, unless the label of the node is
$``\gxor"$. In the $``\gxor''$ case, the first corresponds to the value variable
and the second to the auxiliary variable. Finally, $\matG^{(i)}$ has $n$ columns
at the end for the variables that correspond to the input nodes of
$\graph^{(i)}$. For the last $2n$ columns, all the columns of the value
variables precede the columns of the auxiliary variables. These rules completely
define matrix $\matG^{(i)}$ (see Table~\ref{tbl:sketchOfIthGmatrix} for an
illustration).

\begin{table}[!h]
  \centering
  \begin{tabular}{ l | c c c c c c c : c c c c c c }
                            & $r_{d_i}$ & $z_{d_i}$ & $r_{d_i - 1}$ & $z_{d_i - 1}$ & $\dots$ & $r_{n + 1}$ & $z_{n + 1}$ & $x_n$ & $\dots$ & $x_1$ & $w_n$ &$\dots$& $w_1$ \\
    \hline
      eq. of $v_{d_i}$      & $1$       & $2$       & $\star$           &
      $\star$           & $\dots$ & $\star$         & $\star$         &
      $\star$   &
      $\dots$ & $\star$   & $0$   &$\dots$& $0$   \\
      eq. of $v_{d_i - 1}$  & $0$       & $0$       & $1$           &
      $2$           & $\dots$ & $\star$         & $\star$         & $\star$   &
      $\dots$ & $\star$   & $0$   &$\dots$& $0$   \\
      $\vdots$              &           &           &               &               & $\vdots$&             &             &       & $\vdots$&       &       &$\vdots$&      \\
      eq. of $v_{n + 1}$    & $0$       & $0$       & $0$           &
      $0$           & $\dots$ & $1$         & $2$         & $\star$   & $\dots$
      & $\star$   & $0$   &$\dots$& $0$
  \end{tabular}
  \caption{Illustration of the matrix $\matG^{(i)}$ (assuming that
  $\circuit{C}_i$ has no $``\gxor''$ gates).}
  \label{tbl:sketchOfIthGmatrix}
\end{table}

\noindent Before defining the final matrix $\matG$, we state and prove some
basic properties of $\matG^{(i)}$.

\begin{claim} \label{claim:ithGmatrixIsBinaryInvertible}
    The matrix $\matG^{(i)}$ is binary invertible.
\end{claim}

\begin{proofb}{of Claim \ref{claim:ithGmatrixIsBinaryInvertible}}
    We remind that the dimensions of $\matG^{(i)}$ are
  $(d_i - n) \times (2 d_i)$. Because of the order of rows and columns of
  $\matG^{(i)}$, and by the form of the equations of Table
  \ref{tbl:equationsForHardnesscSIS}, we have that the $1 \times 2$ vectors that
  appear in the diagonal of $\matG^{(i)}$ are equal to $\vec{\gamma}_2^T =
  \begin{bmatrix} 1 & 2 \end{bmatrix}$. The only other non-zero elements of the
  $k$-th row of $\matG^{(i)}$ appear in the columns corresponding to
  $v^{(i)}_{p_1(d_i - k)}$ and $v^{(i)}_{p_2(d_i - k)}$. But, by construction
  $v^{(i)}_{p_1(d_i - k)}$ and $v^{(i)}_{p_2(d_i - k)}$ are always before
  $v^{(i)}_{d_i - k}$ in the topological ordering of $\graph^{(i)}$. Therefore,
  their corresponding columns are after the columns of $v^{(i)}_{d_i - k}$.
  Hence, the only non-zero elements of $\matG^{(i)}$ are above its 3rd diagonal.
  This implies that $\matG^{(i)}$ has the form $\begin{bmatrix} \left(\matI_{d_i
  - n} \otimes \vec{\gamma}_2^T + \matU^{(d_i - n) \times (2 (d_i - n))}
  \right)& \matV^{n \times (2 n)} \end{bmatrix}$ and as shown $\matU^{(d_i - n)
  \times (2 (d_i - n))}$ has non-zero elements only above the 3rd diagonal. This
  concludes the claim that $\matG^{(i)}$ is binary invertible.
\end{proofb}

\begin{claim} \label{claim:correctnessOfithGmatrix}
    Let $\vecs \in \binset^{2 d_i}$ be a solution to the modular linear equation
  $\matG^{(i)} \vecs = \vecb^{(i)} \pmod{4}$. Let $\bvec{x}$ be a binary string
  consisting of the value variables of the input nodes, i.e. $\bvec{x} = (s_{2
  d_i - 2 n + 1}, s_{2 d_i - 2 n + 2}, \dots, s_{2 d_i - n + 1})$. Then, the
  second coordinate $s_2$ of $\vecs$ is equal to $s_2 =
  \circuit{C}_i(\bvec{x})$.
\end{claim}

\begin{proofb}{of Claim \ref{claim:correctnessOfithGmatrix}}
    We inductively prove that the value of the coordinate $s_{2 d_i - 2 j + 2}$
    of $\vecs$ is equal to the value of the non-input node $v_j$ in
    $\graph^{(i)}$ in the evaluation of $\circuit{C}_i(\bvec{x})$.
  \smallskip

  \noindent \textsc{Induction Base.} By the definition of $\bvec{x}$ we have
  that the coordinates $(s_{2 d_i - 2 n + 1}, \dots, s_{2 d_i - n + 1})$ are
  equal to the input values.
  \smallskip

  \noindent \textsc{Inductive Hypothesis.} Assume that for any $k$ such that $n
  < k < j$ the value of the coordinate $s_{2 d_i - 2 k + 2}$ of $\vecs$ is equal
  to the value of the non-input node $v_k$ in $\graph^{(i)}$ in the evaluation
  of $\circuit{C}_k(\bvec{x})$.
  \smallskip

  \noindent \textsc{Inductive Step.} The vector $\vecs$ has to satisfy the $(d_i
  - n - j + 1)$-th modular equation of the system $\matG^{(i)} \vecs =
  \vecb^{(i)} \pmod{4}$. Without loss of generality, we assume that the label of
  $v_j$ is $``\gnand"$. This equation then suggests that $s_{2 d_i - 2 j + 1} +
  2 s_{2 d_i - 2 j + 2} - s_{2 d_i - 2 p_1(j) + 2} - s_{2 d_i - 2 p_2(j) + 2} =
  2 \pmod{4}$ and by Claim \ref{claim:nand} we get that
  \begin{equation} \label{eq:claimCorrectnessNand}
  s_{2 d_i - 2 j + 2} = \left(s_{2 d_i - 2 p_1(j) + 2}\right) \gnand
   \left(s_{2 d_i - 2 p_2(j) + 2}\right).
  \end{equation}
  But, from inductive hypothesis we know that $s_{2 d_i - 2 p_1(j) + 2}$ and
  $s_{2 d_i - 2 p_2(j) + 2}$ take the correct values of $v_{p_1(j)}$ and
  $v_{p_2(j)}$ in the evaluation of $\circuit{C}_i(\bvec{x})$. Hence, from
  Equation~\eqref{eq:claimCorrectnessNand} we immediately get that $s_{2 d_i - 2
  j + 2}$ takes the value of node $v_j$. Similarly, we can show the inductive
  step for all other possible labels of $v_j$.
  \smallskip

  \noindent For $j = d_i$ the statement that we just proved through Induction
  implies that $s_2 = \circuit{C}_i(\bvec{x})$.
\end{proofb}

We are finally ready to describe our matrix $\matG$ and vector $\vecb$. We
remind that $d_i = \abs{\circuit{C}_i}$ and we define $d'_i = d_i - n$. Let $d =
\sum_{i = 1}^n d'_i$. The matrix $\matG$ is of dimension $d \times 2 (d + n)$.
We remind that from Claim \ref{claim:ithGmatrixIsBinaryInvertible} matrices
$\matG^{(i)}$ are binary invertible and hence let $\matU^{(i)}$ and
$\matV^{(i)}$ the matrices that satisfy the equation
\begin{equation} \label{eq:decompositionofIthGmatrix}
  \matG^{(i)} = \begin{bmatrix} \left(\matI_{d'_i} \otimes \vec{\gamma}_2^T +
                                      \matU^{(i)}\right) & \matV^{(i)}
                \end{bmatrix}.
\end{equation}
\noindent We define $\matG$ to be equal to
\begin{equation} \label{eq:definitionOfGHardnesscSIS}
  \matG = \begin{bmatrix}
            \left(\matI_{d'_1} \otimes \vec{\gamma}_2^T + \matU^{(1)}\right) &
            \matzero                                                & \dots  &
            \matzero                                                &
            \matV^{(1)} \\
            \matzero                                                &
            \left(\matI_{d'_2} \otimes \vec{\gamma}_2^T + \matU^{(2)}\right) &
            \dots  & \matzero                                                &
            \matV^{(2)} \\
            \vdots                                                  & \vdots                                                  & \ddots & \vdots                                                  & \vdots      \\
            \matzero                                                &
            \matzero                                                & \dots  &
            \left(\matI_{d'_n} \otimes \vec{\gamma}_2^T + \matU^{(n)}\right) &
            \matV^{(n)} \\
          \end{bmatrix},
\end{equation}
\noindent and the vector $\vecb$ to be
\begin{equation} \label{eq:definitionOfBHardnesscSIS}
  \vecb = \begin{bmatrix}
            \vecb^{(1)} \\
            \vecb^{(2)} \\
            \vdots      \\
            \vecb^{(n)}
          \end{bmatrix}.
\end{equation}

\noindent From the definition of $\matG$ and Claim
\ref{claim:ithGmatrixIsBinaryInvertible} it is immediate that $\matG$ is binary
invertible.

\begin{claim} \label{claim:GmatrixIsBinaryInvertible}
    The matrix $\matG$ defined in Equation~\eqref{eq:definitionOfGHardnesscSIS}
  is binary invertible.
\end{claim}

Additionally, let $k_i = 2 + \sum_{j = 1}^{i - 1} d'_j$, then the following
claim is a simple corollary of Equations~\eqref{eq:definitionOfGHardnesscSIS}
and \eqref{eq:definitionOfBHardnesscSIS} and Claim
\ref{claim:correctnessOfithGmatrix}.

\begin{claim} \label{claim:correctnessofGmatrix}
    Let $\vecs \in \binset^{2 (d + n)}$ be a solution to the modular linear
    equation $\matG \vecs = \vecb \pmod{4}$. Let also $\bvec{x}$ the binary
    string that is equal to the value of the value variables of the input nodes,
    i.e. $\bvec{x} = (s_{2 d + 1}, s_{2 d + 2}, \dots, s_{2 d + n})$. Then, the
    binary string $\bvec{z} = (s_{k_1}, s_{k_2}, \dots, s_{k_n})$ is equal to
    $\bvec{z} = \circuit{C}(\bvec{x})$.
\end{claim}

\noindent To complete the description of the $\cSIS$ instance to which we reduce
$\pigeon$, we have to describe also the matrix $\matA$. The dimensions of
$\matA$ are $n \times 2 (d + n)$. We describe as a concatenation of three
matrices $\matA_1 \in \Z_q^{n \times 2 d}$, $\matA_2 \in \Z_q^{n \times n}$,
$\matA_3 \in \Z_q^{n \times n}$, such that $\matA = \begin{bmatrix} \matA_1 &
\matA_2 & \matA_3 \end{bmatrix}$. Each row of $\matA_1$ corresponds to an output
of $\circuit{C}$ and has a single $1$ in the position $(i, k_i)$. More precisely
$\matA_1 = \sum_{i = 1}^n \matE_{i, k_i}$, where $\matE_{i, j}$ is the matrix
with all zeros except in position $(i, j)$. We also set $\matA_2 = \veczero$ and
$\matA_3 = 2 \matI_n$ which implies $\matA = \begin{bmatrix} \matA_1 & \veczero
& 2 \matI_n \end{bmatrix}$ (see Table~\ref{tbl:sketchOfAmatrix} for an
illustration).
\medskip

\begin{table}[!h]
  \centering
  \begin{tabular}{ l | c c c c c c c c : c c c : c c c }
                            & $r^{(1)}_{d_1}$
                            & $z^{(1)}_{d_1}$
                            & $\dots$
                            & $r^{(2)}_{d_2}$ & $z^{(2)}_{d_2}$ & $\dots$ & $r^{(n)}_{d_n}$ & $z^{(n)}_{d_n}$ & $x_n$ & $\dots$ & $x_1$ & $w_n$ &$\dots$& $w_1$ \\[3pt]
    \hline
      $y_1$                 & $0$
                            & $1$
                            & $\dots$
                            & $0$           & $0$           & $\dots$ & $0$         & $0$         & $0$   & $\dots$ & $0$   & $2$   &$\dots$& $0$   \\
      $y_2$                 & $0$
                            & $0$
                            & $\dots$
                            & $0$           & $1$           & $\dots$ & $0$         & $0$         & $0$   & $\dots$ & $0$   & $0$   &$\dots$& $0$   \\
      $\vdots$              &
                            &
                            &
                            &               &               & $\vdots$&             &             &       & $\vdots$&       &       &$\vdots$&      \\
      $y_n$                 & $0$
                            & $0$
                            & $\dots$
                            & $0$           & $0$           & $\dots$ & $0$         & $1$         & $0$   & $\dots$ & $0$   & $0$   &$\dots$& $2$
  \end{tabular}
  \caption{Illustration of the matrix $\matA$ (assuming that
  $\circuit{C}$ has no $``\gxor''$ gates).}
  \label{tbl:sketchOfAmatrix}
\end{table}

At a high level, the first $2s$ columns of $\matA$ and $\matG$ are meant for the
description of circuit $\circuit{C}$. In particular, there are two columns for
each gate of $\circuit{C}$, the first corresponds to an auxilirary variable and
the second to the output of the gate \footnote{This is ice versa in the case of
$\gxor$ gate, but we can assume without loos of generality that our circuit
consists of only $\gnand$ gates and then this holds.}. Then, there are $n$
columns corresponding to the input of $\circuit{C}$ and $n$ auxiliary columns,
one for each output of $\circuit{C}$. Similarly to
Lemma~\ref{lm:hardness-Blichfeldt}, $\matA$ has non-zero elements in the columns
corresponding to the $n$ outputs of $\circuit{C}$ and  $\matG$ and $\vecb$  play
the role of the set $S$, namely they guarantee that the output of $\cSIS$ will
encode $\bvec{x}$ and $\circuit{C}(\bvec{x})$.
\medskip

The output of $\cSIS$ on input $(\matA$, $\matG,\vec{b})$ is one of the
following:
\begin{Enumerate}
   \item a vector $\vecs \in \{0,1\}^{2(n + s)}$ such that $\vecs \in
     \lamperp_4(\matA)$ and $\matG \vecs \equiv  \vecb \pmod4$.
     \medskip

    Let $\bvec{x} \in \{0,1\}^n$ and $\bvec{y} \in \{0,1\}^n$ be the $n$ input
    and $n$ output coordinates of $\vecs$ respectively, i.e. $\bvec{x} = (s_{2 d
    + 1}, \dots, s_{2 d + n})$ and $\bvec{y} = (s_{k_1}, \dots, s_{k_n})$, as
    defined above. Let $\bvec{w} \in \{0,1\}^n$ be the $n$ last coordinates of
    $\vecs$. Then, since each row of $\matA$ has exactly one coordinate equal to
    1, corresponding to a value in $\bvec{y}$, and one coordinate equal to 2,
    corresponding to a coordinate in $\bvec{w}$, $\matA \vecx = \vec{0}
    \pmod{4}$ implies both that $\bvec{y} = \vec{0}$ and $\bvec{w} = \vec{0}$.
    We can now use Claim \ref{claim:correctnessofGmatrix}, and get that
    $\circuit{C}(\bvec{x}) = \bvec{y}$, which in turn implies that
    $\circuit{C}(\bvec{x}) = \vec{0}$. Hence, $\bvec{x}$ is a valid solution to
    $\pigeon$ with input $\circuit{C}$.

  \item two vectors $\vecs, \vect \in \{0,1\}^{2(n + s)}$, such that
    $\vecs \neq \vect$, $\vecs - \vect \in \lamperp_4(\matA)$ and $\matG \vecs
    \equiv \vecb \pmod 4$, $\matG \vect \equiv \vecb \pmod 4$.
    \medskip

    Let $\bvec{x}_1 ,\bvec{x}_2 \{0,1\}^n$ and $\bvec{y}_1 ,\bvec{y}_2
    \{0,1\}^n$ be the $n$ input and $n$ output coordinates of $\vecs$ and
    $\vect$ respectively, i.e. $\bvec{x}_1 = (s_{2 d + 1}, \dots, s_{2 d + n})$,
    $\bvec{y}_1 = (s_{k_1}, \dots, s_{k_n})$ and $\bvec{x}_2 = (t_{2 d + 1},
    \dots, t_{2 d + n})$, $\bvec{y}_2 = (t_{k_1}, \dots, t_{k_n})$. Let also
    $\bvec{w}_1 \in \{0,1\}^n$ and $\bvec{w}_2 \in \{0,1\}^n$ be the $n$ last
    coordinates of $\vecs$ and $\vect$ respectively. Then, similarly to the
    previous case, $\matA (\vecs - \vect) = \vec{0} \pmod{4}$ implies
    $\bvec{y}_1 = \bvec{y}_2$ and $\bvec{w}_1 = \bvec{w}_2$. From $\bvec{w}_1 =
    \bvec{w}_2$, $\vecs \neq \vect$ and the uniqueness guaranteed by Proposition
    \ref{prop:bInv-property} we can easily conclude that $\bvec{x}_1 \neq
    \bvec{x}_2$. Also, using Claim \ref{claim:correctnessofGmatrix} we get that
    $\circuit{C}(\bvec{x}_1) = \bvec{y}_1$, $\circuit{C}(\bvec{x}_2) =
    \bvec{y}_2$ and since $\bvec{y}_1 = \bvec{y}_2$ we get
    $\circuit{C}(\bvec{x}_1) = \circuit{C}(\bvec{x}_2)$ with $\bvec{x}_1 \neq
    \bvec{x}_2$. Therefore, the pair $\bvec{x}_1$, $\bvec{x}_2$ is a valid
    solution to $\pigeon$ with input $\circuit{C}$.
\end{Enumerate}

\noindent This completes the hardness proof for $q = 4$.
\medskip

For the $q = 2^\ell$ case, we need to augment $\matA$ and $\matG$. This is done
by introducing $\ell$ variables for every node of $\circuit{C}$. One of the is
still the value variable and the $\ell - 1$ rest are auxiliary variables. We
also have to concatenate a zero matrix of size $s \times (\ell-1)(n + s)$ on the
right of $\matG$. For the matrix $\matA \in \integer^{n \times \ell (n + s)}$,
we concatanate $\ell - 1$ matrices of the form $2^i\matI$ for $i \in \{2,3, 4,
\dots, \ell\}$ on the right. The vector $\vecb$ remains the same. The new tuple
is still a valid input for $\cSIS$, since the parameters are appropriately set
and $\matG$ remains binary invertible. Since only zero entries on the right have
been added to the matrix $\matG$, it decribes the circuit $\circuit{C}$ as we
argued above. Finally, let $\vecx \in \{0,1\}^{\ell(n+s)}$ such that $\matA
\vecx = 0 \pmod{4}$, then the last $(\ell - 1) n$ coordinates of $\vecx$ must be
$0$ and the rest of the proof is as above.
\end{proof}

\noindent Combining Lemma~\ref{lm:inclusion-cSIS} and
Lemma~\ref{lm:hardness-cSIS}, we prove the main theorem of this section.

\begin{theorem} \label{thm:cSISPPPcompleteness}
	The $\cSIS$ problem is $\PPP$-complete.
\end{theorem}

\noindent We provide an example of our reduction for very simple circuit
$\circuit{C}$ in Figure~\ref{fig:cSISHardnessSimpleExample}.

\begin{figure}[!ht]
  \centering
  \includegraphics[scale=0.3]{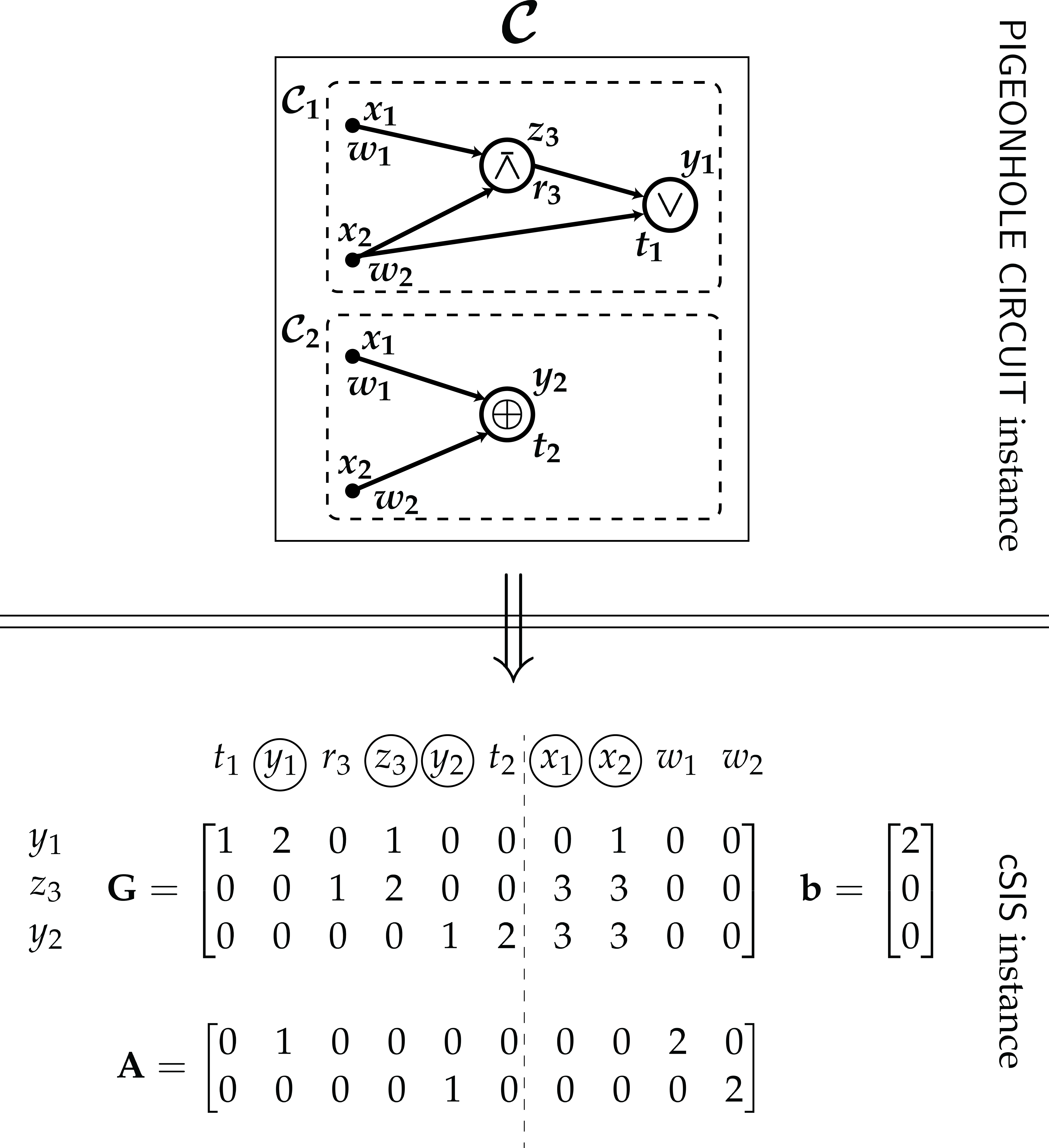}
  \caption{A simple example of the construction of
  Lemma~\ref{lm:hardness-cSIS}.}
  \label{fig:cSISHardnessSimpleExample}
\end{figure}

\section{Complete Collision Resistant Hash Function} \label{sec:collision}

The similarities between the $\cSIS$ problem and $\sis$ raise the question of
whether $\cSIS$ has cryptographic applications. In this section, we propose a
candidate family of \textit{collision-resistant} hash (CRH) functions based on
the \emph{average-case hardness} of the $\cSIS$ problem, and also discuss its
worst-case hardness. The computational problem $\weakcSIS$ associated with our
collision resistant hash function family is a variant of $\cSIS$ presented in
Section \ref{sec:cSIS}. In this case, the modular constraint are
\textit{homogeneous}, i.e. $\vecb = \veczero$, and the inequality constraints on
the dimension of the matrices are strict. This change in the relation of the
parameters might seem insignificant, but it is actually very important since it
transforms or problem into a purely lattice problem: On input matrices
$\matA,\matG$ with corresponding bases $\matB_{\matA}$ and $\matB_{\matG}$,
where $\matG$ is binary invertible, find two vectors $\vecx$ and $\vecy$ such
that $\vecx, \vecy \in \lat(\matB_{\matG})$ and $\vecx - \vecy \in
\lat(\matB_{\matA})$. Our proof that $\weakcSIS$ is $\PWPP$-complete, increases
our hope that $\sis$ is $\PWPP$-complete too, since it overcomes one important
difficulty towards reducing $\weakcSIS$ to $\sis$. We start with a formal
definition of collision-resistant hash function.
\medskip

\noindent \paragr{Collision-Resistant Hash Functions:}
Let $k$ be the security parameter. A family of functions
\[ \mathcal{H} = \left\{\crhf_{\bvec{s}}: \binset^{k} \rightarrow
\binset^{p(k)}\right\}_{\bvec{s} \in \binset^{p'(k)}} \]
where $p(k)$ and $p'(k)$ are polynomials, is \textit{collision-resistant} if:
\begin{itemize}
\item[] (\textsc{Shrinking}) The output of $\crhf_{\bvec{s}}$ is smaller than its
  input. Namely, $p(k) < k$.
\item[] (\textsc{Efficient Sampling}) There exists a probabilistic polynomial-time
  algorithm $\crhfgen$ that on input $1^k$ samples a uniform key $\bvec{s}$.
\item[] (\textsc{Efficient Evaluation}) There exists a deterministic polynomial-time
  algorithm that on input a key $\bvec{s}$ and an $\bvec{x} \in
  \binset^{k}$ outputs $\crhf_{\bvec{s}}(\bvec{x}) \in \binset^{p(k)}$.
\item[] (\textsc{Collision-Resistance}) For every probabilistic polynomial-time adversary
$\Adv$, there exists a negligible function $\nu(k)=\negl(k)$,
  such that for any $k \in \integer_+$:
  \[
  \Pr_{\bvec{s} \gets \crhfgen(1^{k})}\left[
    ( \bvec{x}_1,\bvec{x}_2 ) \gets \Adv(1^{k}, \bvec{s}) \text{ s.t. }
    \bvec{x}_1 \neq \bvec{x}_2 \text{ and }
    \crhf_{\bvec{s}}(\bvec{x}_1) = \crhf_{\bvec{s}}(\bvec{x}_2)\right] \leq \nu(k).
  \]
\end{itemize}

Next, we define our new CRH function family. Let $k$ be the security parameter.
Let $\ell \in \integer_+$, $q = 2^\ell$ and $d \in \integer_+$ be parameters.
Let $r  = \poly(k)$ such that $r \ell<k$. The family of hash functions,
\[ \HcSIS = \left\{\crhf_{\bvec{s}}: \binset^k \rightarrow \binset^{r\ell}
\right\}_{\bvec{s} \in \binset^{p(k)}} \]
is defined as follows.

\begin{enumerate}
  \item[-] $\crhfgen_{\cSIS}(1^k)$ samples a uniform
        $\bvec{s} \in \binset^{p(k)}$ and interprets it as a uniform matrix
        $\matA \in \integer_q^{r \times (k + \ell d)}$ and a uniformly chosen
        binary invertible matrix $\matG \in \integer_q^{d \times (k + \ell d)}$.
        This algorithm runs in polynomial-time in $k$.

  \item[-] $\crhf_{(\matA, \matG)}(\bvec{x})$: On input
        $\bvec{x} \in \binset^k$, compute the unique $\bvec{u} \in \binset^{\ell
        d}$ such that \\ $\matG \begin{bmatrix} \bvec{u} \\ \bvec{x}
        \end{bmatrix} = \vec{0} \pmod{q}$ as in
        Proposition~\ref{prop:bInv-property}, and output $\BitDecomp\left(\matA
        \begin{bmatrix} \bvec{u} \\ \bvec{x} \end{bmatrix} \pmod{q}\right)$.
\end{enumerate}

\begin{remark}\label{remark:ac-cSIStoCRHF}
We note that if $\sis$ \footnote{For a definition of $\sis$, see
Section~\ref{sec:ac-HcSIS}.} is hard, then $\HcSIS$ is collision-resistant. In
fact, a more general statement holds: if $\cSIS $ with $\vecb = \vec{0}$ is hard
on average, then $\HcSIS$ is collision-resistant. Also,  based on the above
description of $\crhf_{(\matA, \matG)}$, the vector $\bvec{z} = \begin{bmatrix}
\bvec{u} \\ \bvec{x} \end{bmatrix}$ is a solution to $\cSIS$ with input $(\matA,
\matG, \vec{0})$. Finally, note that for $\ell = 1$ finding collisions is
trivial and hence we assume that $\ell \ge 2$.
\end{remark}
\medskip

In the rest of this section, we analyze the hardness of finding collision in the
$\HcSIS$ family. We call this problem \emph{weak}-$\cSIS$ ($\weakcSIS$) because
of its similar nature to $\cSIS$. First, we consider its \emph{worst-case}
hardness and draw connections between the family $\HcSIS$, a restricted version
of the $\cSIS$ problem, and the complexity class $\wPPP$. Then, we move on to
the average-case hardness of $\weakcSIS$ and argue that it defines a canditate
for a \emph{universal} CRH function family.

\subsection{Worst-case Hardness of \texorpdfstring{$\HcSIS$}{$\HcSIS$}}
\label{sec:wc-HcSIS}

The class Polynomial Weak Pigeon Principle $\wPPP$ (a subclass of $\PPP$) is
particularly interesting for cryptography because it contains all
collision-resistant hash functions. In this part, we show that a generalized
version of the $\sis$ problem we define below, namely $\weakcSIS$, is complete
for the class $\wPPP$.
\medskip

\begin{nproblem}[\weakcSIS]
\textsc{Input:} Any key $\bvec{s} = (\matA, \matG) \in
\integer_q^{r \times (\ell d + k)} \times \integer_q^{d \times (\ell d + k)}$
that indexes a function $\crhf_{\bvec{s}} \in \HcSIS$. \\
\textsc{Output:} Two boolean vectors $\bvec{x}_1 \neq \bvec{x}_2$, such that
        $\crhf_{\bvec{s}}(\bvec{x}_1) = \crhf_{\bvec{s}}(\bvec{x}_2)$.
\end{nproblem}

The membership of $\weakcSIS$ in $\wPPP$ is straight-forward. The challenging
part of the proof is in showing its $\wPPP$-hardness, which  shares some common
ideas with the proof of Lemma~\ref{lm:hardness-cSIS}.

\begin{lemma}\label{lm:weakcSIS-inclusion}
	$\weakcSIS$ is in $\wPPP$.
\end{lemma}

\begin{proof}
	We show a Karp-reduction from $\weakcSIS$ to $\collision$. Let $\bvec{s}$ be
	an input to $\weakcSIS$ and let $\circuit{C}$ be the $\poly(|\bvec{s}|)$ size
	circuit that on input $\bvec{x}$ outputs $\crhf_{\bvec{s}}(\bvec{x})$. Because
	$r \ell<k$, the $\circuit{C}$ is a valid input for $\collision$. Let
	$\bvec{x}_1,\bvec{x}_2$ be the two boolean vectors that $\collision$ outputs.
	Thus, $\circuit{C}(\bvec{x}_1) = \circuit{C}(\bvec{x}_2)$, which directly
	implies that $(\bvec{x}_1, \bvec{x}_2)$ is a solution of $\weakcSIS$ with
	input $\bvec{s}$.
\end{proof}

Now, we move to the more challenging hardness proof. Even though some of the
proof techniques are reminiscent of the ones in Section~\ref{sec:cSIS}, we need
new ideas, mainly because of the homogeneity of the constraints in $\weakcSIS$.
Before presenting the proof, we define a special version of $\collision$
problem. Then, we state and prove an easy lemma. This lemma has appeared in
various previous works (see Lemma 2.2 in \cite{Jerabek16}), but it is useful for
us to state and prove it using our notation.
\medskip

\begin{nproblem}[\collision_{p(\kappa)}]
\textsc{Input:} A circuit $\circuit{C}$ with $\kappa$ inputs and $p(\kappa)$ outputs. \\
\textsc{Output:} Two boolean vectors $\bvec{x}_1 \neq \bvec{x}_2$, such that
$\circuit{C}(\bvec{x}_1) = \circuit{C}(\bvec{x}_2)$.
\end{nproblem}
\medskip

Then, the following lemma is easy to check:

\begin{lemma}\label{lm:shrinkByOneBit}
  The $\collision$ problem is Karp-reducible to $\collision_{\kappa-2}$.
\end{lemma}

\begin{proof}
  First, we claim that it suffices to show that $\collision_{\kappa-1}$ is
  Karp-reducible to $\collision_{\kappa-2}$. This is because every circuit
  $\circuit{C}$ with $\kappa$ inputs and $m$ outputs can be transformed into a
  circuit $\circuit{C}'$ with $\kappa$ inputs and $\kappa-1$ outputs by padding
  the output with zeros. We note that every collision of $\circuit{C}'$ is a
  collision of $\circuit{C}$ .

  Now, we show that $\collision_{\kappa-1}$ is Karp-reducible to
  $\collision_{\kappa-2}$. Let $\circuit{C}$ be a circuit with $n$ inputs and
  $n-1$ outputs, we create a new circuit $\circuit{C}'$ with $n+1$ inputs and
  $n-1$ outputs such that $\circuit{C}'(\bvec{x}, b) =
  \circuit{C}(\circuit{C}(\bvec{x}),b)$, where $\bvec{x} \in \binset^n$ and $b
  \in \binset$. Then, $\collision_{\kappa-2}$ with input $\circuit{C}'$ outputs
  $\bvec{y}_1 = (\bvec{x}_1,b_1)$ and $\bvec{y}_2 =  (\bvec{x}_2, b_2)$ such
  that $\bvec{y}_1 \neq \bvec{y}_2$ and $\circuit{C}'(\bvec{y}_1) = \circuit{C}'
  (\bvec{y}_2)$. We consider the following possible cases:
  \begin{Enumerate}
    \item[-] $b_1 \neq b_2$, then $\bvec{y}_1$ and $\bvec{y}_2$ form a collision
      for $\circuit{C}$.
    \item[-] $b_1 = b_2$, then $\bvec{x}_1 \neq \bvec{x}_2$ and one of the
      following holds: $\circuit{C}(\bvec{x}_1) \neq \circuit{C}(\bvec{x}_2)$ or
      $\circuit{C}(\bvec{x}_1) = \circuit{C}(\bvec{x}_2)$. In the first case,
      $\bvec{y}_1$ and $\bvec{y}_2$ form a collision for $\circuit{C}$.
      Otherwise, $\bvec{x}_1$ and $\bvec{x}_2$ form a collision for
      $\circuit{C}$.
  \end{Enumerate}
\end{proof}

The above lemma naturally generalizes to any polynomial shrinkage $p(k)$ of the
input by repeating the same construction. For our purposes shrinking the input
by two bits is enough.

\begin{lemma}\label{lm:weakcSIS-hardness}
	$\weakcSIS$ is $\wPPP$-hard.
\end{lemma}

\begin{proof}
	From Lemma~\ref{lm:shrinkByOneBit}, it  suffices to show a Karp-reduction from
	$\collision_{\kappa - 2}$ to $\weakcSIS$.

	\noindent Let $\circuit{C}$ be an input of $\collision_{\kappa-2}$ with $n$
	inputs, $r = n-2$ outputs and $d$ gates. Also, let $q = 4$. We can generalize
	to any $q = 2^\ell$ for $\ell>1$ similarly to the proof of
	Lemma~\ref{lm:hardness-cSIS}.

	If we set $\bvec{u} \in \binset^{2 d}$ to be the unique binary vector such
	that $\matG \begin{bmatrix} \bvec{u} \\ \bvec{x} \end{bmatrix}   = 2\cdot
	\vec{1} \pmod{4}$ in the construction of $\crhf_{\bvec{s}}(\bvec{x})$, then
	the proof follows exactly the same steps as the one of
	Lemma~\ref{lm:hardness-cSIS}. However, in this case we can prove a stronger
	statement, where $ \begin{bmatrix} \bvec{u} \\ \bvec{x} \end{bmatrix}$
	belongs to the lattice $\lamperp(\matG)$. First, let us restate the part of
	the Claim~\ref{claim:nand} that we need for our proof.

	\begin{claim} \label{claim:XOR-OR}
		The following equivalences hold:
		\begin{Itemize}
			\item $x \gxor y = z \iff \exists w \in \{0,1\} \text{ s.t. } -x -
			y + w +  2z\equiv 0 \pmod 4$
			\item $x \gor y = z \iff \exists w \in \{0,1\} \text{ s.t. } x + y
			+ w + 2z\equiv 0 \pmod 4$
		\end{Itemize}
		\end{claim}

	If we could implement a circuit using $\{\gxor, \gor\}$, then combining
	Claim~\ref{claim:XOR-OR} and the proof techniques of
	Lemma~\ref{lm:hardness-cSIS}, we would have the desired result. Even though
	this is not possible, we note that we can implement a circuit using $\{\gxor,
	\gor, 1\}$, since the implementation of $\gnor$ suffices and $(x \gor y) \gxor
	1 = x \gnor y$. Inspired by the above observations, our approach in showing
	the hardness of $\weakcSIS$ is to construct a valid input $\bvec{s} = (\matA,
	\matG)$	for $\weakcSIS$ such that $\matG$ encodes a circuit consisting of only
	$\mathsf{XOR}$ and $\mathsf{OR}$ gates and every solution of this instance of
	$\weakcSIS$ gives a solution for $\collision_{\kappa-2}$ with input
	$\circuit{C}$.

  First, we observe that the fact that $r = n-2$  guarantees the existence of at
  least $6\floor{k/4}$ pairs of $\bvec{x}_1$ and $\bvec{x}_2$ such that
  $\bvec{x}_1 \neq \bvec{x}_2$ and $\crhf_{\bvec{s}}(\bvec{x}_1) =
  \crhf_{\bvec{s}}(\bvec{x}_2)$. In particular, there exist collision pairs such
  that $\bvec{x}_1 \neq \bvec{0}$ and $\bvec{x}_2 \neq \bvec{0}$.

  As we explained, without loss of generality, we assume that $\circuit{C}$
  consists of only $\{\gor, \gxor, \gone\}$ gates. Before starting with the
  reduction we change our circuit $\circuit{C}$ to a circuit $\circuit{C}'$ such
  that any collision of $\circuit{C}'$ corresponds to a collision of
  $\circuit{C}$. We construct a circuit $\circuit{C}'$ such that it first
  computes the $\mathsf{OR}$ of all the inputs, $z = x_1 \gor x_2 \gor \cdots
  \gor x_n$. Then, in the circuit's graph we substitute the outgoing edges of
  nodes with label ``$\gone$'' with outgoing edges from the node $z$, that we
  introduced. Hence, $\circuit{C}$' only contains $\{\gor, \gxor\}$ gates. The
  next claim follows by construction of $\circuit{C}'$.

  \begin{claim} \label{cl:change1byOR}
    For every $\bvec{x} \in \binset^k \setminus{\bvec{0}}$, it holds that
    $\circuit{C}'(\bvec{x}) = \circuit{C}(\bvec{x})$ and $\circuit{C}'(\bvec{0})
    = \bvec{0}$.
  \end{claim}

  \begin{proof}
    Let $\bvec{x} \in \binset^k \setminus{\bvec{0}}$, then for this input the
    value of the node $z$ that we introduced is certainly equal to $1$. Since we
    replaced all the uses of the constant gate $\gone$ with the output of $z$
    and $z$ outputs one, the results of the circuit will be exactly the same.
    Also, it is easy  to observe that any circuit with only $\{\gor, \gxor\}$
    gates, on input $\bvec{0}$ outputs always $\bvec{0}$, and hence
    $\circuit{C}'(0)= \bvec{0}$.
  \end{proof}

  Let $\matG \in \Z_q^{d \times 2(n + r)}$ be the binary invertible matrix that
  encodes $\circuit{C}'$ (as described in Lemma \ref{lm:hardness-cSIS}) and
  $\matA \in \Z_q^{r \times 2(n + r)}$ be as in Lemma~\ref{lm:hardness-cSIS}.
  Observe that since we have only used the gates $\{\gor, \gxor\}$ in
  $\circuit{C}'$ the vector $\vecb$ in the reduction described in Lemma
  \ref{lm:hardness-cSIS} is always equal to $\veczero$. Thus, $(\matA, \matG)$
  is a valid input for the problem $\weakcSIS$.

  Let $(\bvec{x}_1, \bvec{x}_2)$ be a solution of $\weakcSIS$ with $\bvec{x}_1
  \neq \bvec{0}$ and $\bvec{x}_2 \neq \bvec{0}$, then the pair $\bvec{x}_1$ and
  $\bvec{x}_2$ is a collision for $\circuit{C}$. But, we cannot guarantee that
  $\bvec{x}_1 \neq \bvec{0}$ and $\bvec{x}_2 \neq \bvec{0}$ when we use
  $\circuit{C}'$ as input for $\weakcSIS$. Thus, we need a modification of
  $\circuit{C}'$ that guarantees that $\weakcSIS$ will not return the
  $\bvec{0}$-vector as part of a solution. To achieve this, we construct a new
  circuit $\circuit{C}''$ that is exactly the same as $\circuit{C}'$ with one
  more output variable set to be equal to $z = x_1 \gor x_2 \gor \dots \gor
  x_n$. This last output of $\circuit{C}''$ is equal to $1$ if and only if
  $\bvec{x} \neq \bvec{0}$. Observe that the output of $\circuit{C}''$ consists
  of $n - 1$ bits, and hence $\circuit{C}''$ still compresses its domain by one
  bit. Let $(\matA', \matG')$ be the matrices corresponding to $\circuit{C}''$
  according to the construction of proof of Lemma~\ref{lm:hardness-cSIS}, then
  $\bvec{s}' = (\matA', \matG')$ is a valid input for $\weakcSIS$.

  We conclude the proof by observing that every output of $\weakcSIS$ with input
  $\bvec{s}'$ gives a collision $\circuit{C}''(\bvec{x}) =
  \circuit{C}''(\bvec{y})$, with $\bvec{x} \neq \bvec{y}$. If $\bvec{x} \neq
  \bvec{0}$ and $\bvec{y} \neq \bvec{0}$, then it follows from Claim
  \ref{cl:change1byOR} and the construction of $\circuit{C}''$ that
  $\circuit{C}(\bvec{x}) = \circuit{C}(\bvec{y})$, and hence the pair
  $\bvec{x}$, $\bvec{y}$ is a solution to our initial $\collision$ instance.
  Additionally, there is no collision of the form $\bvec{0}$ and $\bvec{x}$.
  Assume that there was, then since the last bit of $\circuit{C}''(\bvec{0})$ is
  $0$, it must be that $\circuit{C}''(\bvec{x})$ is also $0$. However, by
  construction of $\circuit{C}''$, $\bvec{0}$ is the unique binary vector for
  which the last bit of the output is $0$, so it must be that $\bvec{x} = 0$
  which is a contradiction.  Therefore, $\bvec{x} \neq \bvec{0}$ and $\bvec{y}
  \neq \bvec{0}$, and the lemma follows.
\end{proof}

By combining Lemma~\ref{lm:weakcSIS-inclusion} and
Lemma~\ref{lm:weakcSIS-hardness}, we get the following theorem.

\begin{theorem}
  The $\weakcSIS$ problem is $\wPPP$-complete.
\end{theorem}

\subsection{Average-case Hardness of \texorpdfstring{$\HcSIS$}{$\HcSIS$}}
\label{sec:ac-HcSIS}

In the previous section, we showed that $\weakcSIS$ defines a \emph{worst-case}
collision resistant hash function, similar to the result of \cite{Selman1992}
for one-way functions. However, the definition of collision-resistance in
cryptography requires a stronger property. More specifically, it says that given
a random chosen function in the family, it is hard to find a collision for this
function. In this section, we investigate the average-case hardness of $\HcSIS$
in the search of a construction for a candidate of a both \textit{natural} and
\emph{universal collision-resistant hash function}.

We have briefly mentioned the connection between $\HcSIS$ and $\sis$. Now, we
describe it in more detail. We start by defining the $\sis$ problem.

\medskip

\begin{nproblem}[\sis_{q,n,m,\beta, p}]
\textsc{Input:} A uniformly random matrix $\matA \in \Z_q^{n \times m}$, where
$m > 2n \log(q)$.\\
\textsc{Output:} A vector $\vecr \in \Z_q^m$ such that $\norm{\vecr}_p \leq
\beta$ and $\matA \vecr = \vec{0} \pmod{q}$.
\end{nproblem}
Whenever the parameters are clear from the context, we drop them from the
subscripts for ease of notation.
\medskip

It is easy to see that if $\sis$ is hard on the average, then $\HcSIS$ is
collision-resistant. Let $\matA \in \Z_q^{n \times m}$ be uniformly random and
$\bvec{x}_1$ and $\bvec{x}_2$ be a collision of $\crhf_{(\matA, \vec{0})}$.
Then, $\bvec{x}_1 - \bvec{x}_2$ is a solution of $\sis$ with input $\matA$. This
remark, combined with the known reduction from $\tilde{O}(n) \mdash \sivp$ to
$\sis_{\tilde{O}(n),n,\Omega(n \polylog(n)), \sqrt{n}, 2}$  \cite{MicRegev07},
directly implies a reduction from $\tilde{O}(n) \mdash \sivp$ to $\weakcSIS$.

\begin{corollary}
	$\tilde{O}(n) \mdash \sivp$ is reducible to $\weakcSIS$, and thus is contained
	in $\wPPP$.
\end{corollary}

Finally we note that, since $\weakcSIS$ is $\wPPP$-complete and $\wPPP$ contains
all collision-resistant hash functions, the following statement is also true.

\begin{corollary} \label{cor:ac-HcSIS}
  If there exists a family of collision-resistant hash functions $\calH$, then
  there exists a distribution over keys $\bvec{s} \in \binset^{p(k)}$ for
  $\HcSIS$, where $\crhfgen$ draws key from this distribution and $\HcSIS$ is
  collision-resistant.
\end{corollary}

\section{Lattice Problems and \texorpdfstring{$\boldsymbol{\PPP}$}{$\PPP$}}
\label{sec:lattices}

The lattice nature of the $\cSIS$ problem raises also the next question:
\begin{center}
  What is the connection between the class $\PPP$ and lattices?
\end{center}
\noindent In this section we present the known result for the lattice problems
that are contained in $\PPP$ and we also mention how some of these results are
implied by the completeness results that we presented in the previous sections.
We start with the formal definition of the lattice problems that we are
interested in.
\medskip

\paragr{Lattice Problems.} We recall some of the most important lattice
problems: the Shortest Vector Problem ($\gamma \mdash \svp$), the Shortest
Independent Vectors Problem ($\gamma \mdash \sivp$) and the Closest Vector
Problem ($\gamma \mdash \cvp$). We start by defining some important lattice
quantities. For a lattice $\lat$,
\begin{align} \label{eq:lamba1Definition}
  \dist(\vect, \lat)     & = \min_{\vecx \in \lat} ~ \norm{\vecx - \vect} \\
  \lambda_i(\lat(\matB)) & = \min_{\vecx \in \lat \setminus \spn(\vecv_1,\ldots,\vecv_{i-1}), \vecx \neq \vec{0}} \norm{\vecx}, \quad \text{for } i=1,\ldots,n
\end{align}

where $\norm{\vecv_i} = \lambda_i(\lat(\matB))$.

\begin{nproblem}[\gamma \mdash \svp]
\textsc{Input:} A $n$-dimensional basis $\matB \in \Z^{n \times n}$. \\
\textsc{Output:} A lattice vector $\vecv \in \lat$ such that
$\norm{\vecv} \leq \gamma(n) \cdot \lambda_1(\lat(\matB))$.
\end{nproblem}
\bigskip

\begin{nproblem}[\gamma \mdash \sivp]
\textsc{Input:} A $n$-dimensional basis $\matB \in \Z^{n \times n}$. \\
\textsc{Output:} A set of $n$ linearly independent lattice vectors $\vecv_1,
\ldots, \vecv_n$ such that $\max_{i} \norm{\vecv_i} \leq \gamma(n) \cdot
\lambda_n(\lat(\matB))$.
\end{nproblem}
\bigskip

\begin{nproblem}[\gamma \mdash \cvp]
\textsc{Input:} A $n$-dimensional basis $\matB \in \Z^{n \times n}$ and a
target vector $\vect \in \Q^n$. \\
\textsc{Output:} A lattice vector $\vecv$ such that
$\norm{\vecv - \vect} \leq \gamma(n) \cdot \dist(\vecv, \lat(\matB))$.
\end{nproblem}
\medskip

\begin{nproblem}[\frac{1}{\gamma} \mdash \bdd]
\textsc{Input:} A $n$-dimensional basis $\matB \in \Z^{n \times n}$ and a target
  vector $\vect \in \Q^n$, with $\vect \leq \frac{\lambda_1(\lat)}{2\gamma(n)}$.
  \\
\textsc{Output:} A lattice vector $\vecv$ such that $\norm{\vecv - \vect} =
  \dist(\vect, \lat(\matB))$.
\end{nproblem}
\medskip

Where $\gamma(n) \geq 1$ is a non-decreasing function in the lattice dimension
$n$. For $\gamma=1$ we get the exact version of the problem.

We now show that well-studied (approximation) lattice problems are contained in
$\PPP$. First, we define the $\Minkowski$ problem and show a reduction to
$\Blichfeldt$. Second, using known reductions, we conclude the membership of
other lattice problems to $\PPP$.

\begin{nproblem}[\Minkowski_p]
\textsc{Input:} A $n$-dimensional basis $\matB \in \Z^{n \times n}$ for a
  lattice $\lat = \lat(\matB)$.\\
\textsc{Output:} A vector $\vecx \in \lat$ such that $\norm{\vecx}_p \leq
  n^{1/p} \det(\lat)^{1/n}$.
\end{nproblem}

The authors in~\cite{BanJPPR15} give a reduction from $\Minkowski$ to $\pigeon$.
We follow a different approach by showing a reduction from $\Minkowski_p$ to
$\Blichfeldt$. We emphasize that, even though a proof of Minkowski's theorem
uses Blichfedlt's theorem, our reduction differs from this proof technique. We
restrict to subsets of integer points, and thus the inherently continuous
techniques used in the original proof of Minkowski's theorem via Blichfeldt's
theorem cannot be applied in our case.

\begin{lemma} \label{lm:Minkowski-inclusion}
  For $p \geq 1$ and $p=\infty$, $\Minkowski_p$ is in $\PPP$.
\end{lemma}

\begin{proof}
  For any $p \geq 1$ it holds that $\norm{\vecx}_\infty \leq n^{1/p}
  \norm{\vecx}_p$. This implies a Karp-reduction from $\Minkowski_p$ to
  $\Minkowski_\infty$. Hence, it suffices to show a Karp-reduction from
  $\Minkowski_\infty$ to $\Blichfeldt$.

  Let $\matB \in \Z^{n \times n}$ be an input to $\Minkowski_p$, i.e. a basis
  for $\lat = \lat(\matB)$. Let $\ell = \floor{\det(\lat)^{1/n}}$, $S =
  ([0,\ell]^n \cap \Z^n) \setminus \{\vec{0}\}$ and $s = \abs{S} = (\ell + 1)^n
  - 1$. Define the value function representation for $S$ to be $(s,\Cvalue_S(x))
  = (s, \Cvalue_{[0,\ell]^n}(x+1))$, where the circuit $\Cvalue_{[0,\ell]^n}$ is
  constructed as in Lemma~\ref{lm:cubeSetSuccinctDescription}.

  To show that $\matB$ and $S$ are a valid input for $\Blichfeldt$, we need to
  show $\abs{S} \geq \det(\lat) \Rightarrow (\floor{\det(\lat)^{1/n}} + 1)^n
  \geq \det(\lat)^{1/n} + 1$. This follows from the next claim for $x =
  \det(\lat)$.

  \begin{claim}
    For any $x \in \Z$ and $n \geq 1$, $(\floor{x^{1/n}} + 1)^{n} \geq x + 1$.
  \end{claim}

  \begin{proof}
    If $x = k^n$ for some $k\in \Z$, then $\floor{x^{1/n}} = k$. Hence, $(k +
    1)^n \geq k^n +1 \Rightarrow (\floor{x^{1/n}} + 1)^{n} \geq x + 1$.
    Otherwise, let $k \in \integer$ be the smallest integer such that $x < k^n$.
    Then, $\floor{x^{1/n}} = k-1$ and $x+1 \leq k^n$. Hence, $ (\floor{x^{1/n}}
    + 1)^{n}  = k^n \geq x+1$.
  \end{proof}

  Finally, $\Blichfeldt$ on input $\matB$ and $S$ will output one of following:
  \begin{Enumerate}
    \item a vector $\vecx$ such that $\vecx \in S \cap \lat$. Since
          $\vecx \in S$, we get $\norm{x}_\infty \leq \det(\lat)^{1/n}$. Hence,
          $\vec{x}$ is a solution to $\Minkowski_\infty$.
    \item two vectors $\vecx \neq \vecy$, such that $\vecx, \vecy \in S$ and
          $\vecx - \vecy \in \lat$. Since $\vecx, \vecy \in S$, we get  $\vecx	-
          \vecy \in  [-\ell,\ell]^n$. Hence, $\norm{\vecx-\vecy}_\infty \leq
          \det(\lat)^{1/n}$ and $\vecx - \vecy$ is a solution to
          $\Minkowski_\infty$.
  \end{Enumerate}
\end{proof}

A direct corollary of the above lemma is that the most common version of the
$\Minkowski$ problem in the $\ell_2$-norm, is in $\PPP$.
\begin{corollary} \label{cor:Minkowski2-inclusion}
  $\Minkowski_2$ is in $\PPP$.
\end{corollary}

Furthermore, there are known connections between $\gamma \mdash \svp$ with
polynomial approximation factor and the class $\PPP$. Specifically, it is known
that $n \mdash \svp$ is Cook-reducible to $\Minkowski_\infty$ (see
\cite[Theorem 1.23]{Rothvoss15}).

This, along with the reduction from $\gamma\mdash\svp$ to $\sqrt{n}\gamma^2
\mdash \cvp$ in~\cite{Noah15} (for $\gamma=n$) and
Lemma~\ref{lm:Minkowski-inclusion}, implies the following.

\begin{corollary}
  $n \mdash \svp$ and $n^{2.5}\mdash\cvp$ are Cook-reducible to $\Blichfeldt$.
\end{corollary}

A special type of lattices, that have gained a lot of attention due to their
efficiency in cryptographic applications, are \emph{ideal lattices}. The
definition and cryptographic applications of ideal lattices are outside the
scope of this work and can be found in \cite{LyubashevskyPR13}. But, we include
the following lemma, which needs only the basic fact that $\lambda_1(\lat) =
\lambda_2(\lat) = \dots = \lambda_n(\lat)$ in ideal lattices, where $\lambda_i$
is the length of the $i$-th linearly independent vector. Let us denote by
$\gamma \mdash \problem{iSVP}$ the shortest vector problem on ideal lattices.

\begin{lemma}\label{lm:nSVPinPPPforIdealLattices}
	$\sqrt{n} \mdash \problem{iSVP}$ is  in $\PPP$.
\end{lemma}

\begin{proof}
	For ideal lattices, it holds that  $\lambda_1(\lat) = \lambda_2(\lat) = \dots
	= \lambda_n(\lat)$. Minkowski's second Theorem states that
  \[\lambda_1 \cdot \lambda_2 \cdot \dots \cdot \lambda_n \geq \det(\lat).\]
  Hence, on ideal lattices	$\lambda_1 \geq \det(\lat)^{1/n}$. Combining this
  with the first Minkowski's theorem, that states that $ \lambda_1 \leq \sqrt{n}
  \det(\lat)^{1/n}$, we get that $\Minkowski_2$ with input $\lat$ solves
  $\sqrt{n} \mdash \svp$ of $\lat$.
\end{proof}

\section{Cryptographic Assumptions in \texorpdfstring{$\boldsymbol{\PPP}$}{$\PPP$}}
\label{sec:cryptoHard}

The fact that the class $\PPP$ exhibits strong connections to cryptography was
already known since its introduction \cite{Papadimitriou1994}.
Papadimitriou~\cite{Papadimitriou1994} showed that if $\PPP = \FP$, then one-way
permutations do not exist. Since there are constructions of one-way permutations
based on the discrete logarithm problem on $\Zp^*$ for a prime $p$, the problem
of factoring Blum integers\footnote{A Blum integer is the product of two
distinct primes $p,q$ such that $p \equiv 3 \pmod{4}$ and $q \equiv 3
\pmod{4}$.} ~\cite{Rabin79}, the RSA assumption~\cite{RSA78} and a special type
of elliptic curves~\cite{Kaliski1991}, all these cryptographic assumptions
become insecure if $\PPP = \FP$.

Moreover, if $\wPPP = \FP$, then collision-resistant hash functions do not
exist; this follows directly from the definition of the class $\wPPP$.
Collision-resistant hash function families can be constructed based, for
example, on the hardness assumption of the discrete logarithm problem over
$\Zp^*$ for a prime $p$ \cite{C:ChaVanPfi91}, and from the $\sis$
problem~\cite{MicRegev07} problem, which implies that these assumptions become
insecure if $\PPP = \FP$ or even if $\wPPP = \FP$.

Additionally, recently it was shown that integer factorization is in $\PPP$
\cite{Jerabek16}. In fact, Je\v{r}\'{a}bek~\cite{Jerabek16} showed that integer
factorization lies in $\wPPP \cap \PPA$, which gives strong evidence that
factoring is not $\PPP$-complete.

We extend the connections between cryptographic assumptions and $\PPP$ by
presenting a more general formulation of discrete logarithm and showing its
membership to $\PPP$.  We call this formulation \emph{discrete logarithm over
general groups}.

\begin{definition} \label{def:generalGroup}
	A general group $(\G,\star)$ is a cyclic group
	for which there exists a circuit describing a bijection
	$\Cindex_{\G}: \G \to [\ord(\G)]$, which we call the
	\textit{index function of \G} in analogy to the definition of
	Section~\ref{sec:prelims:setDescription}, and a circuit
	$f: [\ord(\G)] \times [\ord(\G)] \to [\ord(\G)]$, which we call group
  operation function, such that for all  $\mathbf{x},\mathbf{y} \in \G$,
$\Cindex_\G(\mathbf{x} \star \mathbf{y}) =
f(\Cindex_{\G}(\mathbf{x}),\Cindex_{\G}(\mathbf{y}))$. The group $(\G, \star)$
is defined by
\begin{Enumerate}
	\item a number $g$ in $[\ord(\G)]$ such that $\Cindex_{\G}(\mathbf{g}) =
	g$, where
	$\mathbf{g}$ is the generator of $\G$,
	\item a number $id$ in $[\ord(\G)]$ such that $\Cindex_{\G}(\mathbf{id}) =
	id$, where
	$\mathbf{id}$ is the identity element of $\G$,
	\item the polynomial-size circuit for the group operation function $f$.
\end{Enumerate}
\end{definition}

In order to illustrate the above definition we show that the multiplicative
group $\Z_p^*$ for prime $p$ is a general group. First, note that $\ord(\Z_p^*)
= p-1$. Let $c \in \Z_p^*$ be a generator of $\Z_p^*$, then we set $g = c - 1
\in [p-1]$, $id = 0$ and we define $f:[p-1] \times [p-1] \to [p-1]$ as follows:
let $x,y \in \Z_p^*$, then $f(x,y) = (x+1)\cdot (y+1) - 1 \pmod{q}$. Now, we are
ready to define the discrete logarithm problem over a general group.

\begin{nproblem}[\problem{DLOG}]
  \textsc{Input:} An alleged general group $(\G, \star)$ represented by $(g,id,
  f_\G, \Cindex_{\G})$\footnote{Observe that in this definition $g$ and $id$
	are just elements of $\G$ and $f_\G, \Cindex_{\G}$ are given in the form of
	circuits.} and a target $y \in [s]$.  \\
\textsc{Output:} One of the following:
\begin{Enumerate}
	\item an $x \in [s]$ such that  $\Cindex_{\G}(\mathbf{g}^x) =
y$,
	\item two $x,y \in [s]$ such that $x \neq y$ and
	$\Cindex_{\G}(\mathbf{g}^x) = \Cindex_{\G}(\mathbf{g}^y)$.
	\end{Enumerate}
\end{nproblem}

\begin{remark} We aim at defining only syntantic problems, so even though we
could define the promise version of the problem where $(\G, \star)$ is always a
general group, we choose to define it as above. In the first case, the output is
the discrete logarithm of the target element, whereas in the second case, the
output is a proof that the tuple $(g, id, f_\G)$ does not represent a general
group.

Also, we observe that given an element $x \in [\ord(\G)]$ and the representation
$(g,id, f_\G)$ of a general group, there exists an efficient procedure for
computing $\Cindex_{\G}(\mathbf{g}^x)$. For instance, using repeating squaring
starting with $f_\G(g,g)$.
\end{remark}

\begin{theorem}
	$\problem{DLOG}$ is in $\PPP$.
\end{theorem}

\begin{proof}
	Let $(g, id, f_\G, y)$ be an input for $\problem{DLOG}$ and $n =
	\ceil{\log(s)}$, then we construct a circuit $\circuit{C}: \binset^n \to
	\binset^n$ as follows: For a binary vector $\bvec{x} \in \binset^n$,
	$\circuit{C}(\bvec{x}) = \abs{\Cindex_{\G}(\mathbf{g}^{\BitComp(\bvec{x})}) -
	y}$. As we mentioned in the remark above, we can compute
	$\circuit{C}(\bvec{x})$ efficiently using only $g$, $id$ and $f_\G$.

	Then, the output of $\pigeon$ with input $\circuit{C}$ is one of the
	following:
	\begin{Enumerate}
		\item a binary vector $\bvec{x}$ such that $\circuit{C}(\bvec{x}) =
		\bvec{0}$.

		In this case, $x = \BitComp(\bvec{x})$ is the discrete logarithm of
		$y$. Namely, $\Cindex_{\G}(\mathbf{g}^{x})=y$.

		\item two binary vectors $\bvec{x}$, $\bvec{y}$ such that $\bvec{x}
		\neq \bvec{y}$ and $\circuit{C}(\bvec{x}) = \circuit{C}(\bvec{y})$.

		In this case, $\Cindex_{\G}(\mathbf{g}^{\BitComp(\bvec{x})})  =
		\Cindex_{\G}(\mathbf{g}^{\BitComp(\bvec{y})})$ implies one of the following:
		(a) $f_{\G}$ does not implement a valid group operation over $\G$, (b)
		$\Cindex_{\G}$ is not a valid index circuit of the set $\G$, or (c)
		$\mathbf{g}^{\BitComp(\bvec{x})}= \mathbf{g}^{\BitComp(\bvec{y})}$. In
		either case this is a certificate that the input $(g, id, f_\G,
		\Cindex_{\G})$ does  not represent a valid general group.
	\end{Enumerate}
\end{proof}

\begin{openProblem}
	\textit{Is $\problem{DLOG}$ $\PPP$-hard?}
\end{openProblem}

We note that it is not known if elliptic curves are general groups, and hence
discrete logarithm over general elliptic curves is not known to belong in $\PPP$
(see Open Problem \ref{op:ellipticCurves}).
 \clearpage

\section*{Acknowledgements}
We thank the anonymous FOCS reviewers for their helpful comments.
We thank an anonymous reviewer and Nico D\"{o}ttling for bringing in our
attention a universal CRH following Levin's paradigm.
We thank Vinod Vaikuntanathan, Daniel Wichs and Constantinos Daskalakis for
helpful and enlightening discussions. We thank Christos-Alexandros Psomas
and his coauthors for sharing their unpublished manuscript \cite{BanJPPR15}.
MZ also thanks Christos-Alexandros Psomas and Christos Papadimitriou for
many fruitful discussions during his visit to Simons Institute at Berkeley at
Fall 2015.

\bibliographystyle{alpha}
\bibliography{abbrev3,crypto,ref,giorgos}

\newcommand{\etalchar}[1]{$^{#1}$}
\begin{thebibliography}{BTVW17}

\bibitem[AB09]{arora09}
Sanjeev Arora and Boaz Barak.
\newblock {\em Computational Complexity: A Modern Approach}.
\newblock Cambridge University Press, 2009.

\bibitem[ABB15]{AisenbergBB15}
James Aisenberg, Maria~Luisa Bonet, and Sam Buss.
\newblock 2-d tucker is ppa complete.
\newblock In {\em Electronic Colloquium on Computational Complexity (ECCC)},
  volume~22, page 163, 2015.

\bibitem[ABPW17]{AngelBPW17}
Omer Angel, S{\'e}bastien Bubeck, Yuval Peres, and Fan Wei.
\newblock Local max-cut in smoothed polynomial time.
\newblock In {\em Proceedings of the 49th Annual ACM SIGACT Symposium on Theory
  of Computing}, pages 429--437. ACM, 2017.

\bibitem[AD97]{AjtaiD97}
Mikl{\'o}s Ajtai and Cynthia Dwork.
\newblock A public-key cryptosystem with worst-case/average-case equivalence.
\newblock In {\em Proceedings of the Twenty-Ninth Annual {ACM} Symposium on the
  Theory of Computing}, pages 284--293, 1997.

\bibitem[Ajt96]{Ajtai1996}
Mikl{\'o}s Ajtai.
\newblock Generating hard instances of lattice problems.
\newblock In {\em Proceedings of the twenty-eighth annual ACM symposium on
  Theory of computing}, pages 99--108. ACM, 1996.

\bibitem[AP11]{AlwenP11}
Jo{\"e}l Alwen and Chris Peikert.
\newblock Generating shorter bases for hard random lattices.
\newblock {\em Theory of Computing Systems}, 48(3):535--553, 2011.

\bibitem[AR04]{FOCS:AhaReg04}
Dorit Aharonov and Oded Regev.
\newblock Lattice problems in {NP} cap {coNP}.
\newblock In {\em 45th FOCS}, pages 362--371. {IEEE} Computer Society Press,
  October 2004.

\bibitem[BCE{\etalchar{+}}98]{BeameCEIP1998}
Paul Beame, Stephen Cook, Jeff Edmonds, Russell Impagliazzo, and Toniann
  Pitassi.
\newblock The relative complexity of np search problems.
\newblock {\em Journal of Computer and System Sciences}, 57(1):3--19, 1998.

\bibitem[BGG{\etalchar{+}}14]{EC:BGGHNS14}
Dan Boneh, Craig Gentry, Sergey Gorbunov, Shai Halevi, Valeria Nikolaenko, Gil
  Segev, Vinod Vaikuntanathan, and Dhinakaran Vinayagamurthy.
\newblock Fully key-homomorphic encryption, arithmetic circuit {ABE} and
  compact garbled circuits.
\newblock In Phong~Q. Nguyen and Elisabeth Oswald, editors, {\em
  EUROCRYPT~2014}, volume 8441 of {\em {LNCS}}, pages 533--556. Springer,
  Heidelberg, May 2014.

\bibitem[BIQ{\etalchar{+}}17]{BelovsIQSY17}
Aleksandrs Belovs, G{\'a}bor Ivanyos, Youming Qiao, Miklos Santha, and Siyi
  Yang.
\newblock On the polynomial parity argument complexity of the combinatorial
  nullstellensatz.
\newblock In {\em Proceedings of the 32nd Computational Complexity Conference},
  page~30. Schloss Dagstuhl--Leibniz-Zentrum fuer Informatik, 2017.

\bibitem[BJP{\etalchar{+}}15]{BanJPPR15}
Frank Ban, Kamal Jain, Christos Papadimitriou, Christos~Alexandros Psomas, and
  Aviad Rubinstein.
\newblock Reductions in ppp.
\newblock {\em Unpublished Manuscript}, 2015.

\bibitem[BKM17]{EC:BonKimMon17}
Dan Boneh, Sam Kim, and Hart~William Montgomery.
\newblock Private puncturable {PRFs} from standard lattice assumptions.
\newblock In Jean{-}S{\'{e}}bastien Coron and Jesper~Buus Nielsen, editors,
  {\em EUROCRYPT~2017, Part~I}, volume 10210 of {\em {LNCS}}, pages 415--445.
  Springer, Heidelberg, May 2017.

\bibitem[Bli14]{Blichfeldt1914}
Hans~Frederik Blichfeldt.
\newblock A new principle in the geometry of numbers, with some applications.
\newblock {\em Transactions of the American Mathematical Society},
  15(3):227--235, 1914.

\bibitem[BLP{\etalchar{+}}13]{STOC:BLPRS13}
Zvika Brakerski, Adeline Langlois, Chris Peikert, Oded Regev, and Damien
  Stehl{\'e}.
\newblock Classical hardness of learning with errors.
\newblock In Dan Boneh, Tim Roughgarden, and Joan Feigenbaum, editors, {\em
  45th ACM STOC}, pages 575--584. {ACM} Press, June 2013.

\bibitem[BO06]{Buresh06}
Joshua Buresh-Oppenheim.
\newblock On the tfnp complexity of factoring.
\newblock {\em Unpublished manuscript}, 2006.
\newblock \url{http://www.cs.toronto.edu/~bureshop/factor.pdf}.

\bibitem[BP16]{C:BraPer16}
Zvika Brakerski and Renen Perlman.
\newblock Lattice-based fully dynamic multi-key {FHE} with short ciphertexts.
\newblock In Matthew Robshaw and Jonathan Katz, editors, {\em CRYPTO~2016,
  Part~I}, volume 9814 of {\em {LNCS}}, pages 190--213. Springer, Heidelberg,
  August 2016.

\bibitem[BPR15]{BitanskyPR15}
Nir Bitansky, Omer Paneth, and Alon Rosen.
\newblock On the cryptographic hardness of finding a nash equilibrium.
\newblock In {\em Foundations of Computer Science (FOCS), 2015 IEEE 56th Annual
  Symposium on}, pages 1480--1498. IEEE, 2015.

\bibitem[BTVW17]{TCC:BTVW17}
Zvika Brakerski, Rotem Tsabary, Vinod Vaikuntanathan, and Hoeteck Wee.
\newblock Private constrained {PRFs} (and more) from {LWE}.
\newblock In Yael Kalai and Leonid Reyzin, editors, {\em TCC~2017, Part~I},
  volume 10677 of {\em {LNCS}}, pages 264--302. Springer, Heidelberg, November
  2017.

\bibitem[BV14]{BrakerskiV14}
Zvika Brakerski and Vinod Vaikuntanathan.
\newblock Efficient fully homomorphic encryption from (standard) lwe.
\newblock {\em SIAM Journal on Computing}, 43(2):831--871, 2014.

\bibitem[CDDT09]{ChenDDT09}
Xi~Chen, Decheng Dai, Ye~Du, and Shang-Hua Teng.
\newblock Settling the complexity of arrow-debreu equilibria in markets with
  additively separable utilities.
\newblock In {\em Foundations of Computer Science, 2009. FOCS'09. 50th Annual
  IEEE Symposium on}, pages 273--282. IEEE, 2009.

\bibitem[CDO15]{ChenO15}
Xi~Chen, David Durfee, and Anthi Orfanou.
\newblock On the complexity of nash equilibria in anonymous games.
\newblock In {\em Proceedings of the forty-seventh annual ACM symposium on
  Theory of computing}, pages 381--390. ACM, 2015.

\bibitem[CDT09]{ChenDT09}
Xi~Chen, Xiaotie Deng, and Shang-Hua Teng.
\newblock Settling the complexity of computing two-player nash equilibria.
\newblock {\em Journal of the ACM (JACM)}, 56(3):14, 2009.

\bibitem[CPY17]{ChenPY17}
Xi~Chen, Dimitris Paparas, and Mihalis Yannakakis.
\newblock The complexity of non-monotone markets.
\newblock {\em Journal of the ACM (JACM)}, 64(3):20, 2017.

\bibitem[CvP92]{C:ChaVanPfi91}
David Chaum, Eug{\`e}ne {van Heijst}, and Birgit Pfitzmann.
\newblock Cryptographically strong undeniable signatures, unconditionally
  secure for the signer.
\newblock In Joan Feigenbaum, editor, {\em CRYPTO'91}, volume 576 of {\em
  {LNCS}}, pages 470--484. Springer, Heidelberg, August 1992.

\bibitem[DEF{\etalchar{+}}16]{DengEFLQX16}
Xiaotie Deng, Jack~R Edmonds, Zhe Feng, Zhengyang Liu, Qi~Qi, and Zeying Xu.
\newblock Understanding ppa-completeness.
\newblock In {\em LIPIcs-Leibniz International Proceedings in Informatics},
  volume~50. Schloss Dagstuhl-Leibniz-Zentrum fuer Informatik, 2016.

\bibitem[DGP09]{DaskalakisGP09}
Constantinos Daskalakis, Paul~W Goldberg, and Christos~H Papadimitriou.
\newblock The complexity of computing a nash equilibrium.
\newblock {\em SIAM Journal on Computing}, 39(1):195--259, 2009.

\bibitem[DP11]{DaskalakisP11}
Constantinos Daskalakis and Christos Papadimitriou.
\newblock Continuous local search.
\newblock In {\em Proceedings of the twenty-second annual ACM-SIAM symposium on
  Discrete Algorithms}, pages 790--804. Society for Industrial and Applied
  Mathematics, 2011.

\bibitem[DTZ18]{DaskalakisTZ18}
Constantinos Daskalakis, Christos Tzamos, and Manolis Zampetakis.
\newblock A converse to banach's fixed point theorem and its cls completeness.
\newblock {\em Proceedings of the 50th annual ACM symposium on Theory of
  computing (STOC)}, 2018.

\bibitem[EGG06]{ElkindGG06}
Edith Elkind, Leslie~Ann Goldberg, and Paul Goldberg.
\newblock Nash equilibria in graphical games on trees revisited.
\newblock In {\em Proceedings of the 7th ACM Conference on Electronic
  Commerce}, pages 100--109. ACM, 2006.

\bibitem[ER17]{EtscheidR17}
Michael Etscheid and Heiko R{\"o}glin.
\newblock Smoothed analysis of local search for the maximum-cut problem.
\newblock {\em ACM Transactions on Algorithms (TALG)}, 13(2):25, 2017.

\bibitem[FGMS17]{FearnleyGMS17}
John Fearnley, Spencer Gordon, Ruta Mehta, and Rahul Savani.
\newblock Cls: New problems and completeness.
\newblock {\em arXiv preprint arXiv:1702.06017}, 2017.

\bibitem[FPT04]{FabrikantPT04}
Alex Fabrikant, Christos Papadimitriou, and Kunal Talwar.
\newblock The complexity of pure nash equilibria.
\newblock In {\em Proceedings of the thirty-sixth annual ACM symposium on
  Theory of computing}, pages 604--612. ACM, 2004.

\bibitem[FRG18]{FilosG18}
Aris Filos-Ratsikas and Paul~W Goldberg.
\newblock Consensus halving is ppa-complete.
\newblock {\em Proceedings of the 50th annual ACM symposium on Theory of
  computing (STOC)}, 2018.

\bibitem[GKW17]{FOCS:GoyKopWat17}
Rishab Goyal, Venkata Koppula, and Brent Waters.
\newblock Lockable obfuscation.
\newblock In {\em 58th FOCS}, pages 612--621. {IEEE} Computer Society Press,
  2017.

\bibitem[Gol06]{Goldreich:2006:FCV:1202577}
Oded Goldreich.
\newblock {\em Foundations of Cryptography: Volume 1}.
\newblock Cambridge University Press, New York, NY, USA, 2006.

\bibitem[GP17]{GoldbergP17}
Paul~W Goldberg and Christos~H Papadimitriou.
\newblock Towards a unified complexity theory of total functions.
\newblock {\em Journal of Computer and System Sciences}, 2017.

\bibitem[GPS16]{GargPS16}
Sanjam Garg, Omkant Pandey, and Akshayaram Srinivasan.
\newblock Revisiting the cryptographic hardness of finding a nash equilibrium.
\newblock In {\em Annual Cryptology Conference}, pages 579--604. Springer,
  2016.

\bibitem[GPV08]{STOC:GenPeiVai08}
Craig Gentry, Chris Peikert, and Vinod Vaikuntanathan.
\newblock Trapdoors for hard lattices and new cryptographic constructions.
\newblock In Richard~E. Ladner and Cynthia Dwork, editors, {\em 40th ACM STOC},
  pages 197--206. {ACM} Press, May 2008.

\bibitem[Gri01]{Grigni01}
Michelangelo Grigni.
\newblock A sperner lemma complete for ppa.
\newblock {\em Information Processing Letters}, 77(5-6):255--259, 2001.

\bibitem[GSW13]{C:GenSahWat13}
Craig Gentry, Amit Sahai, and Brent Waters.
\newblock Homomorphic encryption from learning with errors:
  Conceptually-simpler, asymptotically-faster, attribute-based.
\newblock In Ran Canetti and Juan~A. Garay, editors, {\em CRYPTO~2013, Part~I},
  volume 8042 of {\em {LNCS}}, pages 75--92. Springer, Heidelberg, August 2013.

\bibitem[GVW13]{STOC:GorVaiWee13}
Sergey Gorbunov, Vinod Vaikuntanathan, and Hoeteck Wee.
\newblock Attribute-based encryption for circuits.
\newblock In Dan Boneh, Tim Roughgarden, and Joan Feigenbaum, editors, {\em
  45th ACM STOC}, pages 545--554. {ACM} Press, June 2013.

\bibitem[GVW15a]{C:GorVaiWee15}
Sergey Gorbunov, Vinod Vaikuntanathan, and Hoeteck Wee.
\newblock Predicate encryption for circuits from {LWE}.
\newblock In Rosario Gennaro and Matthew J.~B. Robshaw, editors, {\em
  CRYPTO~2015, Part~II}, volume 9216 of {\em {LNCS}}, pages 503--523. Springer,
  Heidelberg, August 2015.

\bibitem[GVW15b]{STOC:GorVaiWic15}
Sergey Gorbunov, Vinod Vaikuntanathan, and Daniel Wichs.
\newblock Leveled fully homomorphic signatures from standard lattices.
\newblock In Rocco~A. Servedio and Ronitt Rubinfeld, editors, {\em 47th ACM
  STOC}, pages 469--477. {ACM} Press, June 2015.

\bibitem[HNY17]{HubacekNY17}
Pavel Hub{\'a}cek, Moni Naor, and Eylon Yogev.
\newblock The journey from np to tfnp hardness.
\newblock In {\em LIPIcs-Leibniz International Proceedings in Informatics},
  volume~67. Schloss Dagstuhl-Leibniz-Zentrum fuer Informatik, 2017.

\bibitem[HPV89]{HirschPV1989}
Michael~D Hirsch, Christos~H Papadimitriou, and Stephen~A Vavasis.
\newblock Exponential lower bounds for finding brouwer fix points.
\newblock {\em Journal of Complexity}, 5(4):379--416, 1989.

\bibitem[HRRY17]{conf/ipco/HobergRRY17}
Rebecca Hoberg, Harishchandra Ramadas, Thomas Rothvoss, and Xin Yang.
\newblock Number balancing is as hard as minkowski's theorem and shortest
  vector.
\newblock In Friedrich Eisenbrand and Jochen K{\"o}nemann, editors, {\em IPCO},
  volume 10328 of {\em Lecture Notes in Computer Science}, pages 254--266.
  Springer, 2017.

\bibitem[HY17]{HubavcekY17}
Pavel Hub{\'a}{\v c}ek and Eylon Yogev.
\newblock Hardness of continuous local search: Query complexity and
  cryptographic lower bounds.
\newblock In {\em Proceedings of the Twenty-Eighth Annual ACM-SIAM Symposium on
  Discrete Algorithms}, pages 1352--1371. Society for Industrial and Applied
  Mathematics, 2017.

\bibitem[Jer16]{Jerabek16}
Emil Jer{\'a}bek.
\newblock Integer factoring and modular square roots.
\newblock {\em J. Comput. Syst. Sci.}, 82(2):380--394, 2016.

\bibitem[JPY88]{JohnsonPY1988}
David~S Johnson, Christos~H Papadimitriou, and Mihalis Yannakakis.
\newblock How easy is local search?
\newblock {\em Journal of computer and system sciences}, 37(1):79--100, 1988.

\bibitem[Kal91]{Kaliski1991}
Burton~S. Kaliski.
\newblock One-way permutations on elliptic curves.
\newblock {\em Journal of Cryptology}, 3(3):187--199, Jan 1991.

\bibitem[KN09]{Kozhevnikov2009}
A.~A. Kozhevnikov and S.~I. Nikolenko.
\newblock On complete one-way functions.
\newblock {\em Problems of Information Transmission}, 45(2):168--183, Jun 2009.

\bibitem[KNY17]{KomargodskiNY17}
Ilan Komargodski, Moni Naor, and Eylon Yogev.
\newblock White-box vs. black-box complexity of search problems: Ramsey and
  graph property testing.
\newblock In {\em 58th IEEE Annual Symposium on Foundations of Computer Science
  (FOCS)}. IEEE Canada, 2017.

\bibitem[KPR{\etalchar{+}}13]{KintaliPRST13}
Shiva Kintali, Laura~J Poplawski, Rajmohan Rajaraman, Ravi Sundaram, and
  Shang-Hua Teng.
\newblock Reducibility among fractional stability problems.
\newblock {\em SIAM Journal on Computing}, 42(6):2063--2113, 2013.

\bibitem[KV05]{KaltofenV05}
Erich Kaltofen and Gilles Villard.
\newblock On the complexity of computing determinants.
\newblock {\em computational complexity}, 13(3-4):91--130, 2005.

\bibitem[Lev87]{Levin1987}
Leonid~A. Levin.
\newblock One-way functions and pseudorandom generators.
\newblock {\em Combinatorica}, 7(4):357--363, 1987.

\bibitem[Lev03]{Levin2003}
L.~A. Levin.
\newblock The tale of one-way functions.
\newblock {\em Problems of Information Transmission}, 39(1):92--103, Jan 2003.

\bibitem[LPR13]{LyubashevskyPR13}
Vadim Lyubashevsky, Chris Peikert, and Oded Regev.
\newblock On ideal lattices and learning with errors over rings.
\newblock {\em J. {ACM}}, 60(6):43:1--43:35, 2013.

\bibitem[Mic04]{Micciancio04}
Daniele Micciancio.
\newblock Almost perfect lattices, the covering radius problem, and
  applications to ajtai's connection factor.
\newblock {\em SIAM J. Comput.}, 34(1):118--169, 2004.

\bibitem[Min10]{Minkowski1910}
Hermann Minkowski.
\newblock {\em Geometrie der zahlen}, volume~40.
\newblock 1910.

\bibitem[MP89]{MeggidoP1989}
N~Meggido and CH~Papadimitriou.
\newblock A note on total functions, existence theorems, and computational
  complexity.
\newblock Technical report, Tech. report, IBM, 1989.

\bibitem[MP12]{MicciancioP12}
Daniele Micciancio and Chris Peikert.
\newblock Trapdoors for lattices: Simpler, tighter, faster, smaller.
\newblock In {\em Annual International Conference on the Theory and
  Applications of Cryptographic Techniques}, pages 700--718. Springer, 2012.

\bibitem[MR07]{MicRegev07}
Daniele Micciancio and Oded Regev.
\newblock Worst-case to average-case reductions based on gaussian measures.
\newblock {\em SIAM J. Comput.}, 37(1):267--302, April 2007.

\bibitem[MW16]{EC:MukWic16}
Pratyay Mukherjee and Daniel Wichs.
\newblock Two round multiparty computation via multi-key {FHE}.
\newblock In Marc Fischlin and Jean-S{\'{e}}bastien Coron, editors, {\em
  EUROCRYPT~2016, Part~II}, volume 9666 of {\em {LNCS}}, pages 735--763.
  Springer, Heidelberg, May 2016.

\bibitem[OLD01]{OldsLD2001}
Carl~Douglas Olds, Anneli Lax, and Giuliana Davidoff.
\newblock {\em The geometry of numbers}, volume~41.
\newblock Cambridge University Press, 2001.

\bibitem[Pap92]{Papadimitriou1992}
Christos~H Papadimitriou.
\newblock The complexity of the lin--kernighan heuristic for the traveling
  salesman problem.
\newblock {\em SIAM Journal on Computing}, 21(3):450--465, 1992.

\bibitem[Pap94]{Papadimitriou1994}
Christos~H Papadimitriou.
\newblock On the complexity of the parity argument and other inefficient proofs
  of existence.
\newblock {\em Journal of Computer and system Sciences}, 48(3):498--532, 1994.

\bibitem[Pei09]{STOC:Peikert09}
Chris Peikert.
\newblock Public-key cryptosystems from the worst-case shortest vector problem:
  extended abstract.
\newblock In Michael Mitzenmacher, editor, {\em 41st ACM STOC}, pages 333--342.
  {ACM} Press, May~/~June 2009.

\bibitem[PRSD17]{STOC:PeiRegSte17}
Chris Peikert, Oded Regev, and Noah Stephens-Davidowitz.
\newblock Pseudorandomness of ring-{LWE} for any ring and modulus.
\newblock In Hamed Hatami, Pierre McKenzie, and Valerie King, editors, {\em
  49th ACM STOC}, pages 461--473. {ACM} Press, June 2017.

\bibitem[PS18]{conf/pkc/PeikertS18}
Chris Peikert and Sina Shiehian.
\newblock Privately constraining and programming prfs, the lwe way.
\newblock In Michel Abdalla and Ricardo Dahab, editors, {\em PKC (2)}, volume
  10770 of {\em Lecture Notes in Computer Science}, pages 675--701. Springer,
  2018.

\bibitem[Rab79]{Rabin79}
M.~O. Rabin.
\newblock Digitalized signatures and public-key functions as intractable as
  factorization.
\newblock Technical report, Cambridge, MA, USA, 1979.

\bibitem[Reg09]{Regev09}
Oded Regev.
\newblock On lattices, learning with errors, random linear codes, and
  cryptography.
\newblock {\em Journal of the ACM (JACM)}, 56(6):34, 2009.

\bibitem[Rot16]{Rothvoss15}
Thomas Rothvoss.
\newblock Integer optimization and lattices.
\newblock {\em Lecture Notes}, 2016.
\newblock
  \url{https://sites.math.washington.edu/~rothvoss/583D-spring-2016/IntOpt-and-Lattices.pdf}.

\bibitem[RSA78]{RSA78}
R.~L. Rivest, A.~Shamir, and L.~Adleman.
\newblock A method for obtaining digital signatures and public-key
  cryptosystems.
\newblock {\em Commun. ACM}, 21(2):120--126, February 1978.

\bibitem[RSS17]{RosenSS17}
Alon Rosen, Gil Segev, and Ido Shahaf.
\newblock Can ppad hardness be based on standard cryptographic assumptions?
\newblock In {\em Theory of Cryptography Conference}, pages 747--776. Springer,
  2017.

\bibitem[Rub15]{Rubinstein15}
Aviad Rubinstein.
\newblock Inapproximability of nash equilibrium.
\newblock In {\em Proceedings of the forty-seventh annual ACM symposium on
  Theory of computing}, pages 409--418. ACM, 2015.

\bibitem[Rub16]{Rubinstein16}
Aviad Rubinstein.
\newblock Settling the complexity of computing approximate two-player nash
  equilibria.
\newblock In {\em Foundations of Computer Science (FOCS), 2016 IEEE 57th Annual
  Symposium on}, pages 258--265. IEEE, 2016.

\bibitem[SD15]{Noah15}
Noah Stephens-Davidowitz.
\newblock Dimension-preserving reductions between lattice problems.
\newblock {\em Unpublished Manuscript}, 2015.
\newblock \url{http://www.noahsd.com/latticeproblems.pdf}.

\bibitem[Sel92]{Selman1992}
Alan~L. Selman.
\newblock A survey of one-way functions in complexity theory.
\newblock {\em Mathematical systems theory}, 25(3):203--221, Sep 1992.

\bibitem[Sho97]{Shoup1997}
Victor Shoup.
\newblock Lower bounds for discrete logarithms and related problems.
\newblock In {\em International Conference on the Theory and Applications of
  Cryptographic Techniques}, pages 256--266. Springer, 1997.

\bibitem[SSB17]{SchuldenzuckerSB17}
Steffen Schuldenzucker, Sven Seuken, and Stefano Battiston.
\newblock Finding clearing payments in financial networks with credit default
  swaps is ppad-complete.
\newblock In {\em LIPIcs-Leibniz International Proceedings in Informatics},
  volume~67. Schloss Dagstuhl-Leibniz-Zentrum fuer Informatik, 2017.

\bibitem[SY91]{SchafferY1991}
Alejandro~A Sch{\"a}ffer and Mihalis Yannakakis.
\newblock Simple local search problems that are hard to solve.
\newblock {\em SIAM journal on Computing}, 20(1):56--87, 1991.

\bibitem[VY11]{VaziraniY11}
Vijay~V Vazirani and Mihalis Yannakakis.
\newblock Market equilibrium under separable, piecewise-linear, concave
  utilities.
\newblock {\em Journal of the ACM (JACM)}, 58(3):10, 2011.

\bibitem[WZ17]{WichsZ17}
Daniel Wichs and Giorgos Zirdelis.
\newblock Obfuscating compute-and-compare programs under lwe.
\newblock In {\em Foundations of Computer Science (FOCS), 2017 IEEE 58th Annual
  Symposium on}, pages 600--611. IEEE, 2017.

\end{thebibliography}

\clearpage
\appendix

\section{Missing Proofs} \label{sec:app:missingProofs}

\subsection{Proof of Claim \ref{claim:nand}} \label{sec:app:proofOfClaimNand}

\begin{proof}[Proof of Claim \ref{claim:nand}.]
  The proof of this claim follows easily from the next Table
  \ref{tbl:proofOfClaimNand}, from which we conclude various relationships
  between boolean functions with two inputs and modular equations $\pmod{4}$.
  \begin{table}[!h]
    \centering
    \begin{tabular}{ c c c c | c | c | c }
               $\boldsymbol{x}$ & $\boldsymbol{y}$ & $\boldsymbol{z}$ & $\boldsymbol{w}$ & $\boldsymbol{w + 2z - x - y}$ & $\boldsymbol{z + 2w - x - y}$ & $\boldsymbol{w + 2z + x + y}$  \\[3.4pt]
      \hline
      \hline
               $0$              & $0$              & $0$              & $0$              & $0$                           & $0$                           & $0$                   \Tstruts \\[2pt]
               $0$              & $0$              & $0$              & $1$              & $1$                           & $2$                           & $1$                            \\[2pt]
               $0$              & $0$              & $1$              & $0$              & $2$                           & $1$                           & $2$                            \\[2pt]
               $0$              & $0$              & $1$              &
               $1$              & $-1$                           &
               $-1$                           & $-1$
               \\[2pt]
      \hline
               $0$              & $1$              & $0$              & $0$              & $-1$                          & $-1$                          & $1$                   \Tstruts \\[2pt]
               $0$              & $1$              & $0$              & $1$              & $0$                           & $1$                           & $2$                            \\[2pt]
               $0$              & $1$              & $1$              &
               $0$              & $1$                           &
               $0$                           & $-1$
               \\[2pt]
               $0$              & $1$              & $1$              &
               $1$              & $2$                           &
               $2$                           & $0$
               \\[2pt]
      \hline
               $1$              & $0$              & $0$              & $0$              & $-1$                          & $-1$                          & $1$                   \Tstruts \\[2pt]
               $1$              & $0$              & $0$              & $1$              & $0$                           & $1$                           & $2$                            \\[2pt]
               $1$              & $0$              & $1$              &
               $0$              & $1$                           &
               $0$                           & $-1$
               \\[2pt]
               $1$              & $0$              & $1$              &
               $1$              & $2$                           &
               $2$                           & $0$
               \\[2pt]
      \hline
               $1$              & $1$              & $0$              &
               $0$              & $2$                          &
               $2$                          & $2$                   \Tstruts
               \\[2pt]
               $1$              & $1$              & $0$              &
               $1$              & $-1$                          &
               $0$                           & $-1$
               \\[2pt]
               $1$              & $1$              & $1$              &
               $0$              & $0$                           &
               $-1$                          & $0$
               \\[2pt]
               $1$              & $1$              & $1$              &
               $1$              & $1$                           &
               $1$                           & $1$
               \\[2pt]
    \end{tabular}
    \caption{The values of specific modular expressions $\pmod{4}$ for all
    different binary
             values of the variables.}
    \label{tbl:proofOfClaimNand}
  \end{table}
\end{proof}

\section{Parameters of
  \texorpdfstring{$\boldsymbol{\cSIS}$}{$\cSIS$}
   and
  \texorpdfstring{$\boldsymbol{\weakcSIS}$}{$\weakcSIS$}}

We give a summary of the notation that we use for the parameters of
$\cSIS$ and $\weakcSIS$ (or $\HcSIS$).

\begin{table}[!h]
	\centering
	\begin{tabular}{ c c c }
		Notation &      Parameter 			 		& Problem 		\\[3.4pt]
	\hline
	\hline
		$n$ 	 &    circuit input size 		 	& both 					\\
		$k$		 & 	  $\HcSIS$ input size   	 	& $\weakcSIS$   		\\
		$d$ 	 &      rows of $\matG$ / gates of $\circuit{C}$	 	    & both 				    \\
		$m$ 	 &  columns of $\matA$ and $\matG$	& both					\\
	  $r\ell$  &    output of $\HcSIS$			& $\weakcSIS$
\end{tabular}
\label{tbl:notatioCSIS}
\end{table}

 Also, the known relations between these quantities for a valid $\cSIS$ or
 $\weakcSIS$ input are:
 \[ k = \ell \cdot n, ~~~~ m \geq d \ell + k = d \ell + n \ell, ~~~~ k > r \ell. \]

\section{Universal Collision Resistant Hash Function Family}
\label{sec:app:universalHash}

We sketch the construction of a universal (average-case hard) hash function,
following Levin's paradigm~\cite{Levin1987} for a universal one-way function.
Let $h$ be a hash function family that takes two inputs a key $k$ and a vector
$x \in \bit^n$, and compresses the input $x$ by one bit, i.e. $h(k,x) \in
\bit^{n-1}$. Let $p(\cdot)$ be a polynomial that bounds the running time of
$h(k, \cdot)$. First, using padding on the input we argue the existence of a
hash function family $h'$ that is defined as
\[ h'(k, x \circ y) = h(k, x) \circ y \]
such that $|x \circ y| = p (|x|)$. This implies that $h'(k,\cdot)$ runs in
quadratic time in $|x \circ y|$ (see~\cite[§2.4.1]{Goldreich:2006:FCV:1202577}).
Second, using standard domain extension we argue the existence of a hash
function family $h''$ that compresses the input for more than one bits, i.e.
$h'': \bit^{n'} \rightarrow \bit^{n - 1}$ (we exclude the key $k$ from the
description of the domain). Specifically, we require that the compressing ratio
is enough so that the concatenation of $m$ copies of $h''(k,\cdot)$ is smaller
than the input, meaning that $n' > m \cdot \left( p(n) - 1 \right)$. Let
$p''(\cdot)$ be a polynomial that bounds the running time of $h''$. The hash
function $h'(k,\cdot)$ runs in time quadratic to its input, and using this fact,
$p''(\cdot)$ can be made explicit depending on the length of the extended domain
that we require from $h''$, once that length is explicitly specified. This is
enough to get an upper bound for the running time of $h''(k, \cdot)$.

The universal hash function is described by a collection of $2 \cdot m + 1$
strings:
\[ h_{\mathsf{uni}}= (i_1, \ldots, i_m, k_1, \ldots, k_m, x). \]

The numbers $i_1,\ldots,i_m$ (represented as strings) are the indices to Turing
Machines (assume a canonical ordering of TMs) that describe $m$ different hash
function families $h''_{i_j}: \bit^{n'} \rightarrow \bit^{n - 1}$, $j = 1,
\ldots, m$. The keys $k_j$ define the hash functions $h''_{i_j}(k_j, \cdot)$.
Finally, for $j=1,\ldots,m$, we run each $h''_{i_j}(k_j, x)$ for at most
$p''(|x|)$ steps and output:
\[ h_{\mathsf{uni}}(x) = h''_{i_1}(k_1, x) \circ \cdots \circ h''_{i_m}(k_m, x).
\]

If $h''_{i_j}(k_j, x)$ does not terminate after $p''(|x|)$, we output $\bot$ for
that $j$. We can see that if at least one hash function family $h''_{i_j}$ is
collision-resistant, then so is $h_{\mathsf{uni}}$. At a high level, if at least
one hash function family $h_{i_j}$ is collision-resistant then so is $h'_{i_j}$,
and moreover so is $h''_{i_j}$  by the security of domain extension (e.g.
Merkle–Damgård). Without loss of generality we assume that all families
$h''_{i_j}$ are defined on the same domain and range. Finally, the simple hash
function combiner in which we concatenate the output of $m$ different hash
functions $h''_{i_j}(k_j, \cdot)$, is collision-resistant as long as at least
one hash function family $h''_{i_j}$ is collision-resistant.
 \end{document}